%% file: main.tex
\documentclass[UKenglish,cleveref,numberwithinsect,thm-restate,a4paper,11pt]{article}

\RequirePackage{doi}
\usepackage[margin=1in]{geometry}
\usepackage{hyperref}
\usepackage{blkarray}
\usepackage{bm}
\usepackage{rotating}
\usepackage{afterpage}
\usepackage{diagbox}
\usepackage{microtype}
\usepackage[fleqn]{amsmath}
\usepackage{amssymb,amsfonts,amsthm}
\usepackage{mathtools}
\usepackage{enumerate}
\usepackage{booktabs}
\usepackage{enumitem}
\usepackage[mathlines]{lineno}
\usepackage{thmtools, thm-restate}
\usepackage{tabularx}
\usepackage{url}

\setenumerate{itemsep=0pt,topsep=4pt,parsep=4pt}
\setitemize{itemsep=0pt,topsep=4pt,parsep=4pt}

\usepackage{tikz}
\usetikzlibrary{quotes, backgrounds, positioning, calc, arrows, arrows.meta, decorations.pathreplacing, shapes, shapes.misc}
\tikzstyle{vertex}=[circle,solid,draw=black,minimum size=14pt,inner sep=0pt]
\tikzstyle{cut}=[solid,draw=red]
\tikzset{>=stealth',every on chain/.append style={join},every join/.style={->}}

\tikzset{every picture/.style={draw=black, line width=.5pt},font=\sffamily}
\tikzset{every node/.style={circle, inner sep=0pt, minimum size=14},font=\sffamily}
\tikzset{>={Latex[color=,length=6pt,width=2.5pt 3]},font=\sffamily}
\tikzstyle{graph} = [fill=black!3, fill opacity=1, draw=black,font=\sffamily]
\tikzstyle{source} = [fill=yellow!30, fill opacity=1, draw=black,font=\sffamily]
\tikzstyle{joint} = [fill=black!3, fill opacity=1, draw=black,font=\sffamily]
\tikzstyle{tree} = [fill=red!30, fill opacity=1, draw=black,font=\sffamily]
\tikzstyle{stain} = [fill=orange, draw=black, line width=.7pt, fill opacity=.15, densely dotted]
\tikzstyle{stainEdge} = [fill=green, draw=black, line width=.7pt, fill opacity=.15, densely dotted]

\newtheorem{theorem}{Theorem}
\newtheorem{lemma}[theorem]{Lemma}
\newtheorem{corollary}[theorem]{Corollary}
\newtheorem{claim}[theorem]{Claim}
\newtheorem{fact}[theorem]{Fact}

\newtheorem{conjecture}[theorem]{Conjecture}

\newtheorem{definition}[theorem]{Definition}

\newtheorem*{conjecture*}{Conjecture}
\newtheorem*{question*}{Open Problem}

\newenvironment{claimproof}{\begin{proof}}{\end{proof}}

\newcommand{\N}{\ensuremath{\mathbb{N}}}

\newcommand{\ccFPT}{\ensuremath{\mathsf{FPT}}}
\newcommand{\W}[1]{\ensuremath{\mathsf{W[#1]}}}

\renewcommand{\to}{\!\rightarrow\!}

\newcommand{\homs}[2]{\mbox{\ensuremath{\mathsf{Hom}(#1 \to #2)}}}

\newcommand{\embs}[2]{\mbox{\ensuremath{\mathsf{Emb}(#1 \to #2)}}}

\newcommand{\subs}[2]{\mbox{\ensuremath{\mathsf{Sub}(#1 \to #2)}}}

\newcommand{\strembs}[2]{\mbox{\ensuremath{\mathsf{StrEmb}(#1 \to #2)}}}

\newcommand{\indsubs}[2]{\mbox{\ensuremath{\mathsf{IndSub}(#1 \to #2)}}}

\newcommand{\auts}[1]{\ensuremath{\mathsf{Aut}(#1)}}

\newcommand{\cphoms}[2]{\ensuremath{\mathsf{cp}\textsc{-}\mathsf{Hom}}(#1 \to #2)}

\newcommand{\cfhoms}[2]{\ensuremath{\mathsf{cf}\textsc{-}\mathsf{Hom}}(#1 \to #2)}

\newcommand{\cpindsubs}[3]{\ensuremath{\mathsf{cp}\text{-}\mathsf{IndSub}}(#1 \to_{\!#2} #3)}

\newcommand{\clique}{\ensuremath{\textsc{Clique}}}
\newcommand{\homsprobd}{\ensuremath{\textsc{Hom}_\textsc{d}}}
\newcommand{\homsprob}{\ensuremath{\textsc{Hom}}}
\newcommand{\cphomsprobd}{\ensuremath{\textsc{cp-Hom}_\textsc{d}}}
\newcommand{\cphomsprob}{\ensuremath{\textsc{cp-Hom}}}
\newcommand{\indsubsprob}{\ensuremath{\textsc{IndSub}}}
\newcommand{\indsubsprobd}{\ensuremath{\textsc{IndSub}_\textsc{d}}}
\newcommand{\mcolindsubsprobd}{\ensuremath{\textsc{MultiColIndSub}_\textsc{d}}}

\newcommand{\cpindsubsprobd}{\ensuremath{\textsc{cp-IndSub}_\textsc{d}}}

\newcommand{\subsprob}{\ensuremath{\textsc{Sub}}}

\newcommand{\subsprobd}{\ensuremath{\textsc{Sub}_\textsc{d}}}
\newcommand{\mcolsubsprobd}{\ensuremath{\textsc{MultiColSub}_\textsc{d}}}

\newcommand{\fptred}{\ensuremath{\leq^{\mathrm{fpt}}_{\mathrm{T}}}}

\newcommand{\gadgets}{\ensuremath{\mathcal{F}}}
\newcommand{\imn}{\ensuremath{\mathsf{imn}}} 
\newcommand{\igm}{\ensuremath{\mathsf{igm}}} 

\newcommand\defeq{\stackrel{\mathclap{\footnotesize\mbox{def}}}{=}}

\newcommand{\scO}{\mathcal{O}}
\newcommand{\scF}{\mathcal{F}}
\newcommand{\scS}{\mathcal{S}}
\newcommand{\scP}{\mathcal{P}}
\newcommand{\vt}{\mathcal{V}}
\newcommand{\bags}{\mathcal{B}}
\newcommand{\et}{\mathcal{E}}
\newcommand{\dtw}{\tau} 
\newcommand{\tw}{\operatorname{tw}} 
\newcommand{\cw}{\operatorname{cw}} 
\newcommand{\skel}{\Lambda}
\newcommand{\skelcw}{\mathsf{skel}\text{-}\mathsf{cw}} 
\newcommand{\vc}{\mathsf{vc}} 

\newcommand{\arrow}[1]{%
  \parbox{#1}{\scalebox{.8}{\tikz{\draw[->,>=stealth](0,0)--(#1,0);}}}
}
\newcommand{\orient}[1]{{#1}_{\arrow{5pt}}}
\newcommand{\sk}{\ensuremath{\Lambda}}

\newcommand{\algo}{\mathcal{A}} 

\newcommand{\scT}{\mathcal{T}}

\newcommand{\aspan}{\textsc{as}}
\newcommand{\relabel}{\textsc{relab}}
\newcommand{\link}{\textsc{clique}}
\newcommand{\unite}{\textsc{union}}
\newcommand{\create}{\textsc{create}}

\newcommand{\xon}{\ensuremath{x_{\operatorname{on}}}}
\newcommand{\xoff}{\ensuremath{x_{\operatorname{off}}}}
\newcommand{\Lab}{L}
\newcommand{\rt}{\operatorname{root}}

\newcommand{\ls}{\operatorname{ls}}
\newcommand{\poly}{\operatorname{poly}}

\title{Exact and Approximate Pattern Counting in Degenerate Graphs: New Algorithms, Hardness Results, and Complexity Dichotomies}

\begin{document}

\author{Marco Bressan \\ Department of Computer Science\\ University of Milan\\ Italy \\ marco.bressan@unimi.it
\and
Marc Roth \\ Merton College\\ University of Oxford\\ United Kingdom \\ marc.roth@merton.ox.ac.uk
}



\maketitle
\thispagestyle{empty}


\begin{abstract}
We consider the problem of counting the copies of a $k$-node graph $H$ in an $n$-node graph~$G$. Due to the sparsity of real-world graphs, there is considerable interest in understanding the complexity of this problem when the degeneracy $d$ of $G$ is small.
In this work we present several results on this topic. Our main contributions are:
\begin{itemize}[leftmargin=10pt]
    \item We fully classify the complexity of counting the copies and the induced copies of $H$ in $d$-degenerate graphs $G$, under the Exponential Time Hypothesis. We prove that the copies of $H$ in $G$ can be counted in time $f(k,d) \cdot n^{\max(\imn(H),1)}\cdot \log n$, where $\imn(H)$ is the size of the largest induced matching of $H$, and that this is essentially optimal: if the class of allowed patterns has unbounded induced matching number, the problem cannot be solved in time $f(k,d)\cdot n^{o(\imn(H)/\log \imn(H))}$ for any $f$. We show a similar result for counting \emph{induced} copies, in which case the relevant parameter of $H$ turns out to be its independence number $\alpha(H)$. In the language of parameterized complexity, this gives dichotomies in tractable and hard cases when the parameter is $k+d$, and implies that, unless ETH fails, several patterns cannot be counted in time $f(k,d)\cdot n^{o(k/\log k)}$, including $k$-matchings, $k$-independent sets, (induced) $k$-paths, (induced) $k$-cycles, and induced $(k,k)$-bicliques.
    \item We introduce a novel family of obstructions, that we call \emph{F-gadgets}, which are at the heart of our hardness results. We show that counting the homomorphisms from $H$ to a graph of bounded degeneracy is hard when $H$ has an $F$-gadget of unbounded treewidth. From this result, the hardness results for subgraphs and induced subgraphs above are derived using the complexity monotonicity principle of Curticapean, Dell and Marx (STOC 17) and Chen and Mengel (PODS 16), via a recent approach of Gishboliner, Levanzov and Shapira (ECCC 20). We conjecture that F-gadgets actually \emph{characterise} the hardness of counting homomorphisms in graphs of bounded degeneracy. Our work proves one part of this conjecture, leaving the other part as an open problem.
    \item We give novel algorithms for \emph{approximate} counting of subgraphs and induced subgraphs, based on recent reductions to the colourful decision version of the problem by Dell, Lapinskas and Meeks (SODA 20). We show an algorithm that computes an expected $\varepsilon$-multiplicative approximation of the number of copies of $H$ in $G$ in time $f(k,d)\cdot \varepsilon^{-2} \cdot n^{\tau_1(H)+o(1)}$, where $\tau_1(H)$ is the dag-treewidth of $H$. For induced copies, we give an algorithm with running time $(kd)^{O(k)}\varepsilon^{-2} \cdot n^{\imn(H)+1+o(1)}$. This implies, for instance, that we can efficiently approximate the number of induced $(k,k)$-bicliques in a degenerate graph, which for non-degenerate graphs is impossible under standard parameterized complexity assumptions.
\end{itemize}
\end{abstract}

\clearpage
\pagenumbering{arabic}

\input{intro.tex}

\paragraph*{Acknowledgements.}
M.\ Roth is grateful to Lior Gishboliner for introducing him to pattern counting in degenerate graphs.

\pagebreak

\tableofcontents

\pagebreak

\input{prelims.tex}
\input{gadgets.tex}

\input{subcount.tex}
\input{indcount.tex}

\section{Towards a Classification for $\#\homsprob_{\text{D}}(C)$}
In this section, we provide results on both the hardness and the tractability of $\#\homsprob_{\text{D}}(C)$. These results do not form a complete classification of the fixed-parameter tractability of $\#\homsprob_{\text{D}}(C)$, as they do not give conditions that are both necessary \emph{and} sufficient for the class $C$ to be ``easy''. However, they highlight some interesting directions. Moreover, the hardness part is the technical heart of the hardness results for $\#\subsprob_{\text{D}}(C)$ and $\#\indsubsprob_{\text{D}}(C)$. The tractability part instead provides some intuition on how one could prove tractability through bounds on the treewidth of F-gadgets, or bounds on the cliquewidth of the pattern graph.
\input{homcount.tex}

\input{hom_ub.tex}

\input{approx.tex}

\bibliography{conference}
\newpage

\appendix
\section{Proof of Lemma~\ref{lem:striking_fact}}
\begin{proof}
We consider the three cases separately.
\begin{itemize}
    \item $P=\#\homsprob$.
    
    The claim holds for $P(C)$ by the results of~\cite{DalmauJ04} and the Excluded Grid Theorem~\cite{RobertsonS86-ExGrid}. For $P_{\textsc{d}}$, it holds in one direction by our Theorem~\ref{thm:intro_homs_hard}.
    \item $P=\#\subsprob$. 
    
    For $P(C)$,~\cite{CurticapeanM14}, says that $P(C)$ is $\#\W{1}$-hard if and only if $\vc(C)$ is unbounded. Since the vertex cover number is within a factor of $2$ of the matching number, this implies that $P(C)$ is $\#\W{1}$-hard if and only if $C$ has unbounded matching number. Now let $M_{\ell}$ be a matching on $2\ell$ vertices. Now we observe that a graph $H$ has $M_{\ell}$ as a minor if and only if it has $M_{\ell}$ as a subgraph. The ``if'' is trivial since a subgraph is a minor. For the ``only if'', look at the model of $M_{\ell}$ in $H$, and note that if $\{u,v\} \in E(M_{\ell})$, then there is an edge between the blocks $S_u,S_v$ of the model of $M_{\ell}$ in $H$. We conclude that $\vc(C)$ is unbounded if and only if $C$ has unbounded matching minors.
    
    For $P_{\textsc{d}}(C)$, Theorem~\ref{thm:intro_subs_param} says that $P_{\textsc{d}}(C)$ is $\#\W{1}$-hard if and only if $\imn(C)$ is unbounded. Now we observe that a graph $H$ has $M_{\ell}$ as an induced minor if and only if it has $M_{\ell}$ as an induced subgraph. The ``if'' direction is trivial since and induced subgraph is an induced minor. The ``only'' if direction follows the same argument as above, but now $\{u,v\} \in E(M_{\ell})$ if and only if there is an edge between $S_u,S_v$ in $H$.
    
    \item $P=\#\indsubsprob$.
    
    From~\cite{ChenM15}, we know that $P(C)$ is $\#\W{1}$-hard if and only if $C$ is infinite. Now observe that $C$ is infinite if and only if $C$ contains arbitrarily large independent set minors, since any $k$-vertex graph $H$ has an independent set minor of size $k$ (just delete all edges of $H$).
    
    For $P_{\textsc{d}}(C)$, Theorem~\ref{thm:intro_indsubs_param} says that $P_{\textsc{d}}(C)$ is $\#\W{1}$-hard if and only if $\imn(C)$ is unbounded. Now let $I_{\ell}$ be the independent set on $\ell$ vertices. We observe that a graph $H$ has $I_{\ell}$ as induced minor if and only if it has $I_{\ell}$ as induced subgraph. The ``if'' is trivial since an induced subgraph is an induced minor. For the ``only if'', take one vertex from each block $S_v$ in the witness structure of $I_{\ell}$ in $H$ and note that no edge can exist between two such vertices since $I_{\ell}$ is an induced minor.
\end{itemize}
\end{proof}

\end{document}

%% file: intro.tex
\section{Extended Abstract}
We study the following problems. Given a pattern graph $H$ and a host graph $G$,
\begin{enumerate}[leftmargin=1em,itemsep=2pt,parsep=0pt,topsep=2pt]
\item[]$\#\subsprob$: compute the number of subgraphs of $G$ isomorphic to $H$
\item[]$\#\indsubsprob$: compute the number of induced subgraphs of $G$ isomorphic to $H$
\item[]$\#\homsprob$: compute the number of homomorphisms from $H$ to $G$
\end{enumerate}
These three problems arise in a variety of disciplines such as statistical physics~\cite{TemperleyF61,Kasteleyn63}, database theory~\cite{DurandM15,ChenM16,Arenasetal20}, constraint satisfaction problems~\cite{DalmauJ04}, bioinformatics~\cite{Nogaetat08}, and network analysis~\cite{Miloetal02,Schilleretal15}.
Unfortunately, they are believed to be intractable: loosely speaking, any algorithm for solving them has running time $|G|^{\Theta(k)}$, where $k=|V(H)|$, unless standard conjectures fail.
To circumvent this obstacle, it is common to allow the complexity to depend not only on~$|G|$ and~$k$, but on other structural parameters as well. For instance, it is well known that $\#\homsprob$ can be solved in time $f(k) \cdot |G|^{\tw(H)}$, for some~$f$, where $k=|V(H)|$ and $\tw(H)$ is the treewidth of $H$~\cite{DalmauJ04}. Since $H$ is typically much smaller than~$G$, a running time of $f(k)\cdot |G|^{O(1)}$ is better than $|G|^{\Theta(k)}$ even if $f$ grows quickly. Therefore, we can consider $\#\homsprob$ as tractable when restricted to a class $C$ of patterns with bounded treewidth. In this case one says that $\#\homsprob(C)$ is \emph{fixed-parameter tractable}, or \ccFPT\ for short, where the parameterization is given by $k$; in other words, it is solvable in time $f(k) \cdot |G|^{O(1)}$.
It is also known~\cite{DalmauJ04} that, under the Exponential Time Hypothesis (ETH)~\cite{ImpagliazzoP01}, no algorithm for $\#\homsprob(C)$ runs in time $f(k) \cdot |G|^{o(\tw(H) / \log\tw(H))}$ for any function $f$, whenever $C$ has unbounded treewidth.\footnote{In this paper, when we say that a problem is not \ccFPT, or hard in any other sense, we always assume that one works under standard hardness conjectures (often, ETH or one of its variants). For the sake of readability, from now on we take this assumption for granted, and avoid repeating it.} Thus, the treewidth of $C$ characterises the fixed-parameter tractability of $\#\homsprob(C)$. Similar characterizations are known for the other two problems as well: Curticapean and Marx~\cite{CurticapeanM14} have shown that $\#\subsprob(C)$ is \ccFPT\ if and only if the vertex cover number $\vc(C)$ of $C$ is bounded, while Chen, Thurley and Weyer~\cite{ChenTW08} have shown that $\#\indsubsprob(C)$ is \ccFPT\ if and only if $C$ is finite. 

In this work, we study $\#\subsprob$, $\#\indsubsprob$ and $\#\homsprob$ by assuming as additional parameter the \emph{degeneracy} of $G$, denoted by $d(G)$ or simply $d$. We seek to understand when those problems are efficiently solvable given that $d$ is small, or more precisely, when they are solvable in time $f(k,d) \cdot |G|^{O(1)}$ for some function~$f$. More formally, let us denote by $\#\subsprobd$, $\#\homsprobd$, $\#\indsubsprobd$ the three problems above, but with the parameterization given by $k+d$. Our goal is to find, for each problem, an explicit criterion on the class of allowed patterns~$C$ such that the restriction of the problem to $C$ is \ccFPT\ if and only if $C$ satisfies the criterion. In addition, we aim at establishing fine-grained upper and lower bounds on the complexity of those problems. We observe that degeneracy is a natural and widely-used notion of sparsity which, unlike (say) the treewidth or the maximum degree, captures well the structure of many real-world graphs~\cite{Eppstein2011-ListClique}. For this reason, pattern counting in graphs with small degeneracy has been intensely studied, both in the exact and the approximate variants, with many exciting results~\cite{Bera-ITCS20,Bera-SODA21,Bressan19,Bressan21,Eden20-CliqueGap,Eden20-CliqueArboricity,Gishboliner20}. Still, no \ccFPT\ classification like the one seeked here has been found yet.

The remainder of this extended abstract discusses the related work (Section~\ref{sub:related}), presents our contributions (Section~\ref{sub:results}), discusses our techniques for the upper bounds and the lower bounds (Section~\ref{sub:techniques_ub} and Section~\ref{sub:techniques_lb}), before concluding with some open problems (Section~\ref{sec:conclusion}).

\subsection{Related work}\label{sub:related}
For what concerns upper bounds, the seminal work of Chiba and Nishizeki~\cite{Chiba&1985} showed that, when~$H$ is the complete graph, 
the copies of $H$ in $G$ can be counted in time $f(k,d) \cdot |V(G)|$. The key idea is that, if $G$ is $d$-degenerate then we can orient it acyclically so that every node has out-degree at most~$d$. Therefore, in time $O(f(k,d)\cdot |V(G)|)$ we can perform an exhaustive search of depth $k-1$ from each vertex of $G$. Recently, one of the authors extended this approach to arbitrary patterns via \emph{DAG tree decompositions}~\cite{Bressan19,Bressan21}, which are similar to standard tree decompositions of undirected graphs, but are designed to exploit the degeneracy orientation of $G$. The resulting running time bounds are in the form $f(k,d) \cdot |V(G)|^{\tau(H)} \cdot \log |V(G)|$, where $\tau(H)$ is called the \emph{dag treewidth} of $H$ (for instance, $\tau(H)=1$ for the clique). Our upper bounds for exact counting rely on these bounds (for approximate counting, instead, we use novel algorithms). In a line of work on near-linear-time algorithms, Gishboliner et al.~\cite{Gishboliner20} and Bera et al.~\cite{Bera-ITCS20,Bera-SODA21} have shown that, if the longest induced cycle of $H$ has size at most $5$, then $\#\subsprobd$ can be solved in time $f(d,k)\cdot |V(G)| \cdot \log |V(G)|$.

For what concerns hardness results,~\cite{Bressan19,Bressan21} showed that, unless ETH fails, no algorithm can count the copies of \emph{every} $H$ in time $f(k,d) \cdot |V(G)|^{o(\tau(H)/\log \tau(H))}$, for any $f$. However, this bound does not apply to \emph{every class} of patterns with unbounded dag-treewidth $\tau$, and so~\cite{Bressan19,Bressan21} do not establish an \ccFPT\ classification. Even more recently,~\cite{Gishboliner20,Bera-SODA21} showed that, under the Triangle Detection Conjecture, counting the copies of $H$ in $G$ takes time $\Omega(|V(G)|^{1+\gamma})$ for some $\gamma > 0$ whenever $H$ has an induced cycle of length at least $6$. Again, this does not establish an \ccFPT\ classification,\footnote{However, we wish to point out that an implicit consequence of~\cite{Gishboliner20} is that $\#\subsprobd$ is $\#\W{1}$-hard when $H$ is a matching, and $\#\indsubsprobd$ is $\#\W{1}$-hard when $H$ is the independent set. Here $\#\W{1}$ should be considered the parameterized equivalent of $\#\mathsf{P}$ (see Section~\ref{sec:prelims} for a formal definition).  In fact, we will show that counting $k$-matchings and counting $k$-independent sets are the minimal $\#\W{1}$-hard cases for $\#\subsprobd$ and $\#\indsubsprobd$ respectively. } and the Triangle Detection Conjecture does not seem sufficiently powerful for that matter. In addition, the techniques used in these works are not immediately extendable to obtain an \ccFPT\ classification; and although we use a technique of~\cite{Gishboliner20} (the extension of complexity monotonicity to degenerate graphs, see below), that technique alone is far from sufficient, and indeed our results are based on substantial novel ingredients.

\subsection{Our Results}\label{sub:results}
We provide several results of the fixed-parameter tractability of $\#\homsprobd$, $\#\subsprobd$, and $\#\indsubsprobd$. Our main result is a set of upper and lower bounds that establishes a complete and explicit \ccFPT\ classification for both $\#\subsprobd$ and $\#\indsubsprobd$, therefore giving necessary and sufficient conditions for any class of patterns $C$ to be easily countable (in the \ccFPT\ sense). In addition, we give novel hardness results for $\#\homsprobd$, as well as for $\#\indsubsprobd(\Phi)$, the problem of counting the induced copies of all $k$-vertex patterns that satisfy property $\Phi$. We also give novel (in)tractability results for the \emph{approximate} versions of these problems, where the goal is to compute a multiplicative $\varepsilon$-approximation of the solution. In particular, we provide novel algorithms and upper bounds for approximately counting copies and induced copies. Table~\ref{tab:results} and Table~\ref{tab:results_approx} summarize all these results and compares them to the non-degenerate case.

The remainder of this section discusses all the results in detail (Sections~\ref{sec:intro_subs}-\ref{sub:striking_fact} for exact counting, and Section~\ref{sec:intro_approx} for approximate counting). Our main claims will be stated both in the language of fine-grained complexity (that is, as fine-grained running time bounds) and in the language of parameterized complexity (that is, as dichotomies between classes of patterns). 
In what follows, we denote by $C$ a generic computable class of graphs. For every problem, say $\#\homsprobd$, we denote by $\#\homsprobd(C)$ its restriction to $H \in C$. 

\begin{table}[p]\small
    \begin{tabularx}{\textwidth}{lll}
        \toprule
        \begin{tabular}{l}
            \textbf{Problem}\\(Exact Counting)
        \end{tabular}
          &\begin{tabular}{l}
             \textbf{Our results}  \\
             (parameter: $k+d$)
        \end{tabular}& \begin{tabular}{l}
             \textbf{Known results}  \\
             (parameter: $k$)
        \end{tabular}\\
        \midrule
        \begin{tabular}{l}
             $\#\subsprob$
        \end{tabular}
         & \begin{tabular}{l}
         $f(k,d)\cdot n^{\max(1,\imn(H))}\cdot \log n$\\
         not in $f(k,d)\cdot n^{o(\imn(H)/\log \imn(H))}$\\[7pt]
         Theorem~\ref{thm:intro_subs_finegrained} and Theorem~\ref{thm:intro_subs_param}
         \end{tabular} &
         \begin{tabular}{l}
         $f(k)\cdot n^{\mathsf{vc}(H)+O(1)}$\\
         not in $f(k)\cdot n^{o(\mathsf{vc}(H)/\log \mathsf{vc}(H))}$\\[7pt]
         see~\cite{CurticapeanM14}
        \end{tabular}\\
        \midrule
        \begin{tabular}{l}
             $\#\indsubsprob$
        \end{tabular} &
        \begin{tabular}{l}
        $f(k,d)\cdot n^{\alpha(H)}\cdot \log n$\\
        not in $f(k,d)\cdot n^{o(\alpha(H)/\log \alpha(H))}$\\[7pt]
        Theorem~\ref{thm:intro_indsubs_finegrained} and Theorem~\ref{thm:intro_indsubs_param}
        \end{tabular}  &
        \begin{tabular}{l}
        $f(k)\cdot n^{k}$\\
        not in $f(k)\cdot n^{o(k/\log k)}$ $^\dagger$\\[7pt]
        see~\cite{ChenTW08}
        \end{tabular}\\
        \midrule
        \begin{tabular}{l}
             $\#\indsubsprob(\Phi)$
        \end{tabular} & \begin{tabular}{l} $\#\W{1}$-hard if $\Phi$ = connectedness\\ or if $\Phi$ is minor-closed, unless\\ trivial or with bounded $\alpha(H)$ \\[7pt]
        Theorem~\ref{thm:intro_graphlets} and Theorem~\ref{thm:intro_indsubs_minor_closed}
        \end{tabular}
        & \begin{tabular}{l}
        see~\cite{JerrumM15,Meeks16,RothSW20} for an overview
        \end{tabular}\\
        \midrule
        \begin{tabular}{l}
             $\#\homsprob$
        \end{tabular} & \begin{tabular}{l} $f(k,d)\cdot n^{\tau_1(H)}\cdot \log n$\\ $\#\W{1}$-hard if $\igm(H)$ unbounded\\[7pt]
        Theorem~\ref{thm:intro_homs_hard} and Theorem~\ref{thm:intro_homs_tract}
        \end{tabular}  & \begin{tabular}{l}$f(k)\cdot n^{\tw(H)+1}$\\
        not in $f(k)\cdot n^{o(\tw(H)/\log \tw(H))}$\\[7pt]
        see~\cite{DalmauJ04,Marx10}
        \end{tabular}\\
        \bottomrule
     \end{tabularx}
     \small \caption[caption]{Our results for the case of \emph{exact} counting. Here $\imn$, $\mathsf{vc}$, $\alpha$, $\tau_1$, $\tw$, $\igm$ denote respectively: induced matching number, vertex-cover number, independence number, dag treewidth, treewidth, and size of the largest induced grid minor.
     All fine-grained lower bounds assume the ETH. All the hardness results hold for any computable class $C$ of patterns whose relevant parameter (e.g., $\imn$) is unbounded. The fine-grained lower bound for counting homomorphisms was shown in~\cite{Marx10} and transfers to subgraphs and induced subgraphs via complexity monotonicity~\cite{CurticapeanDM17}.\\
     
     $^\dagger$\begin{footnotesize} For dense classes of graphs $C$, the conditional lower bound is known to be tight (cf.\ \cite{Chenetal05,Chenetal06,RothS18}).\end{footnotesize}}
     \label{tab:results}
\end{table}
\begin{table}[p]\small
\vspace{1cm}
    \begin{tabularx}{\textwidth}{lll}
        \toprule
        \begin{tabular}{l}
            \textbf{Problem}\\(Approx. Counting)
        \end{tabular}
          &\begin{tabular}{l}
             \textbf{Our results}  \\
             (parameter: $k+d$)
        \end{tabular}& \begin{tabular}{l}
             \textbf{Known results}  \\
             (parameter: $k$)
        \end{tabular}\\
        \midrule
        \begin{tabular}{l}
             $\#\subsprob$
        \end{tabular}
         & \begin{tabular}{l}
         $f(k,d)\cdot(1/\varepsilon)^{2} \cdot n^{\tau_1(H)+o(1)}$\\[7pt]
         Theorem~\ref{thm:intro_FPTRAS_subcount} 
         \end{tabular} &
         \begin{tabular}{l}
         $k^{O(k)}\cdot(1/\varepsilon)^{O(1)}\cdot n^{\mathsf{tw}(H)+O(1)}$\\[7pt]
         see~\cite{ArvindR02}
        \end{tabular}\\
        \midrule
        \begin{tabular}{l}
             $\#\indsubsprob$
        \end{tabular} &
        \begin{tabular}{l}
        $(kd)^{O(k)}\cdot (1/\varepsilon)^{2}\cdot n^{\imn(H)+1+o(1)}$\\[7pt]
        Theorem~\ref{thm:intro_FPTRAS_indsubcount} 
        \end{tabular}  &
        \begin{tabular}{l}
        $\W{1}$-hard for any infinite\\ class of patterns$^\ddagger$\\[7pt]
        see~\cite{ChenTW08}
        \end{tabular}\\
        \midrule
        \begin{tabular}{l}
             $\#\indsubsprob(\Phi)$
        \end{tabular} &
        \begin{tabular}{l}
        \ccFPT\ if $\Phi$ is minor-closed\\[7pt]
        Theorem~\ref{thm:intro_approx_minor_closed} 
        \end{tabular}  &
        \begin{tabular}{l}
        see~\cite{Meeks16} for an overview
        \end{tabular}\\
        \bottomrule
     \end{tabularx}
     \small\caption[caption]{Our results for the case of \emph{approximate} counting.
     The symbols $\tw$, $\tau_1$, and $\imn$ are as in Table~\ref{tab:results}. Note that $\#\homsprob$ is left out since we cannot prove bounds better than those for exact counting; in fact, even in the non-degenerate case it is not known whether approximate counting is easier than exact counting, see~\cite{BulatovZ20}.\\
     
     $^\ddagger$\begin{footnotesize} The problem of detecting an induced copy of a pattern graph $H$ in $G$ is $\W{1}$-hard (when parameterized by $k=|V(H)|$), whenever the class of allowed pattern graphs is infinite~\cite{ChenTW08}. Since an $\varepsilon$-approximating reveals whether the solution is zero, the computation of the latter is $\W{1}$-hard under randomised reductions. \end{footnotesize}}
     \label{tab:results_approx}
\end{table}

\subsubsection{Counting Subgraphs: The Induced Matching Number}\label{sec:intro_subs}
Let $\imn(H)$ be the size of the largest induced matching of $H$. A class of graphs $C$ has unbounded induced matching number if for every $c \in \N$, there exists a graph $H \in C$ with $\imn(H) > c$.  
We prove:\pagebreak
  
\begin{theorem}\label{thm:intro_subs_finegrained}
    There exists a computable function $f$ such that $\#\subsprob$ can be solved in time\footnote{The factor $\log |V(G)|$ can be reduced to $O(1)$ if we accept a a bound in expectation, by replacing the dictionary with deterministic logarithmic access time of~\cite{Bressan21} with a dictionary with expected access time $O(1)$.}
    $f(|H|,d(G)) \cdot |V(G)|^{\max(1,\imn(H))} \cdot \log |V(G)|$.
    Moreover, for any computable class of graphs $C$ of unbounded induced matching number, there is no function $f$ such that $\#\subsprob(C)$ can be solved in time
    $f(|H|,d(G)) \cdot |V(G)|^{o(\imn(H) / \log \imn(H))}$
    unless ETH fails.
\end{theorem}
\noindent In the language of parameterized complexity, Theorem~\ref{thm:intro_subs_finegrained} yields:
\begin{theorem}\label{thm:intro_subs_param}
    Let $C$ be a computable class of graphs. If $C$ has bounded induced matching number, then $\#\subsprobd(C)$ is \ccFPT. Otherwise, $\#\subsprobd(C)$ is $\#\W{1}$-hard.
\end{theorem}
Our upper bounds subsume several known algorithms for subgraph counting in degenerate graphs. For instance, we subsume the clique-counting algorithm of Chiba and Nishizeki~\cite{Chiba&1985} 
or the $(k,k)$-biclique counting algorithm of~\cite{Eppstein99} --- indeed, both cliques and bicliques have induced matching number $1$. Our negative results, on the other hand, imply novel conditional lower bounds for a variety of subgraph patterns:
\begin{corollary}\label{cor:intro_subs_cor}
	The following problems are $\#\W{1}$-hard when parameterized by $k$ and $d(G)$ and cannot be solved in time $f(k,d(G))\cdot |G|^{o(k/\log k)}$ for any function $f$, unless ETH fails:
	\begin{enumerate}[itemsep=2pt,parsep=0pt,topsep=2pt]
		\item compute the number of $k$-cycles in $G$.
		\item compute the number of $k$-paths in $G$.
		\item compute the number of $k$-matchings in $G$.
	\end{enumerate}
\end{corollary}
Before, $\#\W{1}$-hardness and fine-grained lower bounds for these patterns were known only for non-degenerate host graphs~\cite{FlumG04,Curticapean13,CurticapeanDM17}.

\subsubsection{Counting Induced Subgraphs: The Independence Number}\label{sec:intro_indsubs}
Let $\alpha(H)$ be the size of the largest independent set of $H$. A class of graphs $C$ has unbounded independence number if for every $c \in \N$, there exists a graph $H \in C$ with $\alpha(H) > c$. We prove:
\begin{theorem}\label{thm:intro_indsubs_finegrained}
    There exists a computable function $f$ such that $\#\indsubsprob$ can be solved in time\footnote{Again, the $\log |V(G)|$ factor can be removed if we want only an \emph{expected} running time bound.}
    $f(|H|,d(G)) \cdot |V(G)|^{\alpha(H)} \cdot \log |V(G)|$.
    Moreover, for any computable class of graphs $C$ of unbounded independence number, there is no function $f$ such that $\#\indsubsprob(C)$ can be solved in time
    $f(|H|,d(G)) \cdot |V(G)|^{o(\alpha(H) / \log \alpha(H))}$
    unless ETH fails. 
\end{theorem}
\noindent In the language of parameterized complexity, Theorem~\ref{thm:intro_indsubs_finegrained} yields:
\begin{theorem}\label{thm:intro_indsubs_param}
    Let $C$ be a computable class of graphs. If $C$ has bounded independence number, then $\#\indsubsprobd(C)$ is \ccFPT. Otherwise, $\#\indsubsprobd(C)$ is $\#\W{1}$-hard.
\end{theorem}

\noindent As a consequence of our results, we obtain novel conditional lower bounds for several patterns:
\begin{corollary}\label{cor:intro_indsubs_cor}
	The following problems are $\#\W{1}$-hard when parameterized by $k$ and $d(G)$ and cannot be solved in time $f(k,d(G))\cdot |G|^{o(k/\log k)}$ for any function $f$, unless ETH fails. 
	\begin{enumerate}[itemsep=2pt,parsep=0pt,topsep=2pt]
		\item compute the number of induced copies of the $(k,k)$-biclique in $G$.
		\item compute the number of $k$-independent sets in $G$.
		\item compute the number of induced $k$-cycles in $G$.
		\item compute the number of induced $k$-paths in $G$.
		\item compute the number of induced $k$-matchings in $G$.
	\end{enumerate}
\end{corollary}

\noindent \textbf{Remark.}
Between our upper and lower bounds, there is a gap in the exponent due to the factors $(\log\imn(H))^{-1}$, $(\log\alpha(H))^{-1}$, and $(\log k)^{-1}$. This is not an artifact of our analysis, but arises from the well-known ``can you beat treewidth'' open problem (see Conjecture~1.3 in~\cite{Marx13}). If ``you cannot beat treewidth'', as recent results seem to suggest~\cite{Komarath20,Bringmann21}, then all those factors can be dropped and the exponents of our upper bounds are asymptotically tight.

\subsubsection{Generalised Induced Subgraph Counting}\label{sec:intro_indsubs_gen}
Following Jerrum and Meeks~\cite{JerrumM15}, we study the following problem. Let $\Phi$ be a fixed graph property. The problem $\#\indsubsprob(\Phi)$\footnote{In~\cite{JerrumM15}, the problem is called $p\text{-}\#\textsc{UnlabelledInducedSubgraphsWithProperty}(\Phi)$.} expects as input a graph $G$ and a positive integer $k$, and the goal is to compute the number of induced subgraphs of size $k$ in $G$ that satisfy~$\Phi$. The parameterization is given by $k$. Initially, Jerrum and Meeks established $\#\indsubsprob(\Phi)$ to be ($\#\W{1}$-)hard if~$\Phi$ is the property of being connected. Today, it is conjectured that $\#\indsubsprob(\Phi)$ is intractable whenever~$\Phi$ is non-trivial, and recent results go in this direction~\cite{RothSW20}. 

In principle, one can hope that $\#\indsubsprob(\Phi)$ becomes easier when $d$ is small. However we prove that, even in this case, the problem remains mostly intractable. As usual, let $\#\indsubsprobd(\Phi)$ be the version of $\#\indsubsprob(\Phi)$ with $d(G)$ as additional parameter. As a first result, we prove:
\begin{theorem}\label{thm:intro_graphlets}
    Let $\Phi$ be the property of being connected. Then $\#\indsubsprobd(\Phi)$ is $\#\W{1}$-hard.
\end{theorem}
\noindent This is only a special case, and understanding the complexity of $\#\indsubsprobd(\Phi)$ for any $\Phi$ remains elusive, even for the not necessarily degenerate case. However, we provide a complete picture in case of \emph{minor-closed properties}. This is a well-studied class that includes, for example, planarity and acyclicity. 
\begin{theorem}\label{thm:intro_indsubs_minor_closed}
    Let $\Phi$ be a minor-closed property. If $\Phi$ is trivial (i.e., constant) or of bounded independence number, then $\#\indsubsprobd(\Phi)$ is \ccFPT. Otherwise $\#\indsubsprobd(\Phi)$ is $\#\W{1}$-hard.
\end{theorem}

\subsubsection{Counting Homomorphisms: The Source of Hardness}\label{sec:intro_homs}
Recall that $\#\homsprobd$ is the problem of counting homomorphisms, parameterized by $k+d(G)$. A central part of our work consists in novel hardness results for $\#\homsprobd$. This is indeed the starting point for all the hardness results for $\#\subsprobd$ and $\#\indsubsprobd$ that we described above. From a technical standpoint, our main finding is that induced grid minors\footnote{A graph $F$ is an induced minor of a graph $H$ if $F$ can be obtained from $H$ by deleting vertices and contracting edges. In contrast to not necessarily induced minors, edge-deletions are not allowed.} are obstructions:
\begin{theorem}\label{thm:intro_homs_hard}
    Let $C$ be a computable class of graphs. If $C$ has induced grid minors of unbounded size, then $\#\homsprobd(C)$ is $\#\W{1}$-hard.
\end{theorem}

\noindent Notice the parallel with $\#\homsprob(C)$,
which is $\#\W{1}$-hard if $C$ has grid minors of unbounded size, \emph{where the minors are not necessarily induced} (this follows by the results of~\cite{DalmauJ04} and by the Excluded Grid Theorem~\cite{RobertsonS86-ExGrid}). However, while $\#\homsprob(C)$ is known to be hard \emph{if and only if} $C$ has grid minors of unbounded size, for $\#\homsprobd(C)$ we give only the ``if'' direction. Hence, we do not know if $\#\homsprobd(C)$ is \ccFPT\ when $C$ has bounded induced grid minors. We leave closing this gap as an open problem, see Section~\ref{sec:conclusion}.

Even though we do not prove that $\#\homsprobd$ is \ccFPT\ when $C$ has bounded induced grid minors, we give some positive results for $\#\homsprobd$. In particular, we prove upper bounds based on the cliquewidth of the \emph{skeleton graph} of the acyclic orientations of $H$. Given an acyclic orientation~$\orient{H}$ of~$H$, its skeleton graph $\skel(\orient{H})$ is the bipartite directed graph having on the left side the \emph{sources} of $\orient H$ (vertices with no incoming arcs), and on the right side the \emph{joints} of $\orient H$ (vertices reachable from two or more sources), and where the arc $(s,v)$ exists if and only if $v$ is reachable from~$s$.\footnote{The skeleton graph was introduced in~\cite{Bressan19} to prove that the dag-treewidth of any pattern $H$ is at most $\lfloor\frac{k}{4}\rfloor+2$, and captures the structure of $\orient H$ that is relevant for the DAG tree decomposition.} The skeleton clique-width $\skelcw(H)$ of $H$ is the maximum clique-width of $\skel(\orient{H})$ over all acyclic orientations of $H$. Our result is:
\begin{theorem}\label{thm:intro_homs_tract}
    There exists a computable function $f$ such that $\#\homsprob$ can be solved in time $f(|H|,d(G))\cdot |V(G)|^{\skelcw(H)}\cdot \log |V(G)|$.
\end{theorem}

\subsubsection{A Remarkable Parallel}\label{sub:striking_fact}
The hardness results described above reveal a remarkable parallel between pattern counting in general graphs and pattern counting in bounded-degeneracy graphs: in both cases, a class of patterns $C$ is hard to count if $C$ is not free from some specific minor, with the only difference that, in the degenerate case, the minor must be induced. This minor-free condition actually gives a dichotomy into $\#\W{1}$-hard and $\ccFPT$ cases, with the possible exception of $\#\homsprobd$, for which we know only the hardness direction.
\begin{lemma}\label{lem:striking_fact}
Let $P \in \{\#\homsprob,\#\subsprob,\#\indsubsprob\}$, and let $P_{\textsc{d}}$ be the version of $P$ parameterized by $k+d(G)$. Then $P(C)$ and $P_{\textsc{d}}(C)$ are $\#\W{1}$-hard if $C$ has, respectively, minors and induced minors of unbounded size, where the minors are:
\begin{itemize}
    \item for $P=\#\homsprob$,  grids
    \item for $P=\#\subsprob$, matchings
    \item for $P=\#\indsubsprob$, independent sets
\end{itemize}
With the possible exception of $\#\homsprobd$, the converse holds too, with the problem becoming $\ccFPT$.
\end{lemma}
\noindent This parallel suggests, of course, that induced grid minors are the right obstructions to $\#\homsprobd$. This is the main problem that our work leaves open, see Section~\ref{sec:conclusion}.

\subsubsection{Approximate Counting}\label{sec:intro_approx}
Since for exact pattern counting our results are mostly negative, we also consider \emph{approximate} pattern counting. In particular, we study fixed-parameter tractable randomised approximation schemes (FPTRASes,~\cite{ArvindR02}), which are the FPT equivalent of fully polynomial-time randomised approximation schemes. Let $P$ be a counting problem with parameterization $\kappa$. For any instance~$I$ of~$P$, let $P(I)$ be the corresponding solution. An $\varepsilon$\emph{-approximation} of $P(I)$ is any rational number $\hat{c}$ such that $(1-\varepsilon)\cdot P(I) \leq \hat{c} \leq (1+\varepsilon)\cdot P(I)$. An FPTRAS for $P$ and $\kappa$ is a (randomised) algorithm $\mathbb{A}$ that, on input $I$ and $\varepsilon$, returns an $\varepsilon$-approximation of $P(I)$ with probability at least $2/3$, and has running time bounded by $f(\kappa(I)) \cdot \poly(|I|,\varepsilon^{-1})$. The error probability $2/3$ can be improved by standard probability amplification to any $\delta > 0$, at the cost of an additional factor of $O(\log (1/\delta))$ in the running time.

We begin with the approximate counting of copies. In the general case, Arvind and Raman~\cite{ArvindR02} show that $\#\subsprob(C)$ admits an FPTRAS if $C$ has bounded treewidth.
Here we show that, in the bounded-degeneracy case, the treewidth can be replaced by the dag-treewidth:
\begin{theorem}\label{thm:intro_FPTRAS_subcount}
There exists a computable function $f$ such that $\#\subsprob$ can be $\varepsilon$-approximated with probability $2/3$ in time $\varepsilon^{-2} \cdot f(k,d)\cdot n^{\tau_1(H)+o(1)}$.
Hence, $\#\subsprobd(C)$ has an FPTRAS whenever $\tau_1(C)$ is bounded.
\end{theorem}
\noindent Combined with our hardness results, this implies that there are classes of patterns for which, in degenerate graphs, exact counting is hard but approximate counting is easy. For example, let~$C$ be the class of all $k$-wreath graphs.\footnote{For $k\geq 3$, the graph~$W_k$ has vertices $V_0~\dot\cup\dots \dot\cup ~V_{k-1}$ and, for all $i\in\{0,\dots,k-1\}$ it contains all edges between $V_i$ and $V_{i+1(\mathsf{mod}~k)}$.} One can see that $C$ has unbounded induced matching number, and therefore by Theorem~\ref{thm:intro_subs_param} $\#\subsprobd(C)$ is $\#\W{1}$-hard; but its dag treewidth turns out to be at most $2$, hence by Theorem~\ref{thm:intro_FPTRAS_subcount} $\#\subsprobd(C)$ has an FPTRAS. Moreover, $C$ has unbounded treewidth, and thus Theorem~\ref{thm:intro_FPTRAS_subcount} is not subsumed by~\cite{ArvindR02}. Further classes with bounded dag treewidth but unbounded induced matching number and unbounded treewidth include, for example, exploded paths and unions of cliques.

Next, we consider approximate counting of induced copies. In the general case, efficient approximation algorithms are unlikely to exist, as even \emph{detecting} an induced copy is $\W{1}$-hard~\cite{ChenTW08}.
In sharp contrast we show that, in the degenerate case, FPTRASes exist whenever the patterns have bounded induced matching number:
\begin{theorem}\label{thm:intro_FPTRAS_indsubcount}
The problem $\#\indsubsprob$ can be $\varepsilon$-approximated with probability $2/3$ in time $O(\varepsilon^{-2} k^{O(k)} d^{k+1} \cdot |V(G)|^{\imn(H)+1+o(1)})$.
Hence, $\#\indsubsprobd(C)$ has an FPTRAS whenever $\imn(C)$ is bounded.
\end{theorem}
\noindent Also in this case, there are classes of patterns for which exact counting is hard but approximate counting is tractable.
For instance, consider the class $C$ of all bicliques, for which $\alpha(C)$ is unbounded but $\imn(C)=1$: exact counting is $\#\W{1}$-hard by Theorem~\ref{thm:intro_indsubs_param}, yet by Theorem~\ref{thm:intro_FPTRAS_indsubcount} we have an FPTRAS with running time proportional to $n^{2+o(1)}$.

Finally, we consider the generalised subgraph counting problem, $\#\indsubsprobd(\Phi)$, see Section~\ref{sec:intro_indsubs_gen}. By Theorem~\ref{thm:intro_indsubs_minor_closed}, $\#\indsubsprobd(\Phi)$ is $\#\W{1}$-hard for any non-trivial minor-closed property of unbounded independence number. Here we complement that result by showing that, for minor-closed properties, the approximate version is easy:
\begin{theorem}\label{thm:intro_approx_minor_closed}
Let $\Phi$ be a minor-closed graph property. Then $\#\indsubsprobd(\Phi)$ has an FPTRAS.
\end{theorem}

We conclude with a note about hardness of approximate counting in bounded-degeneracy graphs. In Section~\ref{sec:approx}, we show hardness for some classes of patterns (for example, the subdivisions of a clique). Unfortunately, those hardness results do not match the upper bounds above, so we do not show a dichotomy like we did for exact counting. However we point out that, for approximate counting problems, complexity dichotomies are rare, and sometimes unlikely to exist~\cite{DyerGJ10}. In fact, even the case of arbitrary host graphs, which has been studied for a long time, is not resolved yet: While the algorithm of Arvind and Raman~\cite{ArvindR02} establishes tractability for $H$ of bounded treewidth, it is open whether the remaining instances yield hardness. On the positive side, we provide a possible starting point for future research in Section~\ref{sec:conclusion}.

\subsection{Overview of our Techniques: Upper Bounds}\label{sub:techniques_ub}
For exact counting, we use the dynamic programming algorithm of~\cite{Bressan19,Bressan21} based on \emph{dag tree decompositions}. The algorithm has running time in the form $f(d,k) \cdot n^{\tau_i(H)} \log n$, where $\tau_1(H) \le \tau_2(H) \le \tau_3(H)$ are three variants of the \emph{dag treewidth} of $H$; $\tau_1$ is for counting homomorphisms, $\tau_2$ is for counting copies, and $\tau_3$ is for counting induced copies. Our contribution here is to prove upper bounds on $\tau_1,\tau_2,\tau_3$ as a function of $\imn$ and $\alpha$. We also bound $\tau_1,\tau_2,\tau_3$ as a function of more complex parameters, see Section~\ref{sec:hom_ub}. The techniques here are mostly constructive: we show how to build a dag tree decomposition of $H$ starting from (say) the tree decomposition of certain induced minors of $H$, see below.

For approximate counting, we make use of a recent result~\cite{DLM20} that reduces parameterized approximate counting problems to their \emph{colourful} decision version. In the colourful decision version,~$G$ is given with a (not necessarily proper) $k$-colouring of its vertices, and we are asked whether it contains or not a colourful copy of $H$ (one that spans all $k$ colours). The result of~\cite{DLM20} says that, if one can solve this problem in time $T(G,k,d)$, then in time $\varepsilon^{-2} k^{2k} n^{o(1)} T(G,k,d)$ one can probabilistically count the (uncoloured) copies of $H$ in $G$ within a multiplicative error of~$\varepsilon$. This holds for both subgraphs and induced subgraphs. Hence, our contribution here is to give FPT algorithms for the colourful decision version of our counting problems. For example, to decide if $G$ contains a colourful independent set, we proceed as follows: If every colour class $V_i$ of~$G$ is large, say $|V_i| > k d$, then $G$ necessarily contains a colourful $k$-independent set (this can be seen by sorting $G$ in degeneracy ordering and greedily removing the first vertex of each colour). Otherwise, if there are $\ell \ge 1$ colour classes of size $\le kd$, then we can explicitly enumerate their colourful $\ell$-independent sets in time $O((kd)^{\ell})$ and recurse on the remaining classes. Our algorithms use this routine and other FPT-style arguments.

\subsection{Overview of our Techniques: Lower Bounds}\label{sub:techniques_lb}
Most of our work is devoted to the lower bounds. We prove those bounds by (i) developing hardness results for $\#\homsprobd$, and (ii) lifting them to $\#\subsprobd$ and $\#\indsubsprobd$. 

\noindent These two steps are carried out with different techniques:

\textbf{A novel obstruction: F-gadgets}. In the first step we introduce F-gadgets, a novel kind of obstructions designed to capture the hardness of counting homomorphisms in bounded-degeneracy graphs. Informally, a graph $H$ has an $F$-gadget if $H$ contains an induced subgraph that can be obtained from $F$ by substituting each vertex $v$ by a connected component $S_v$ and each edge $e$ by a path $P_e$ of length at least $2$ (see Figure~\ref{fig:F_gadget_intro} for an example, and Definition~\ref{def:gadgets} for a formalization). Therefore, F-gadgets can be seen as induced minors under constraints, and indeed, if $F$ is an F-gadget of $H$, then $F$ is an induced minor of $H$, but the converse is not true (for instance, the complete graph is an induced minor of itself, but not an F-gadget of itself).
\begin{figure}[ht]
    \centering
    \input{f-gadget-3.tikz}
    \caption{Left to right: a graph $F$, a graph $H$, and an $F$-gadget of $H$.}
    \label{fig:F_gadget_intro}
\end{figure}
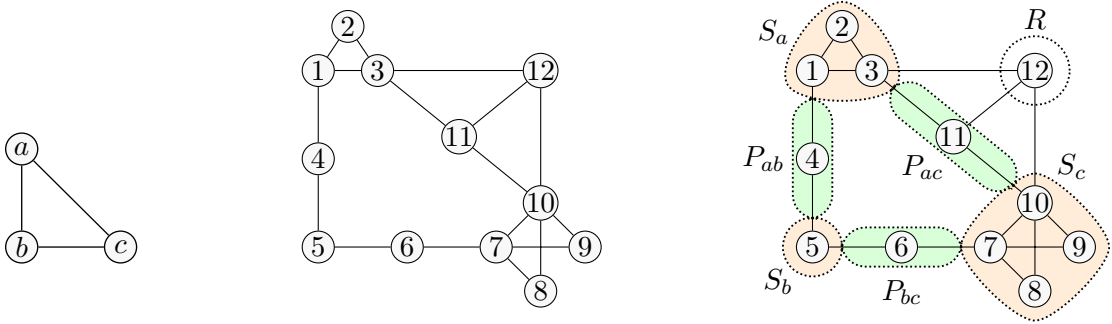

\noindent Exploiting F-gadgets, we prove the following result: if a computable class of graphs $C$ has $F$-gadgets of unbounded treewidth, then $\#\homsprobd(C)$ is $\#\W{1}$-hard. To this end, we show that the problem of counting the homomorphisms from a graph $F$ to an arbitrary graph $G$ can be FPT-reduced to the problem of counting the homomorphisms from a graph $H$ to a $|H|$-degenerate graph $G'$, where $H$ has an $F$-gadget and $|G'|=\poly(|G|)$. More precisely, let $\gadgets(C)$ be the class of all graphs $F$ such that some $H \in C$ has an $F$-gadget. We show the following paramerized Turing reduction:
\begin{linenomath*}
\begin{equation}
\#\homsprob(\gadgets(C)) \fptred
\#\homsprobd(C)    
\end{equation}
\end{linenomath*}
Since $\#\homsprob(\gadgets(C))$ is $\#\W{1}$-hard whenever $\gadgets(C)$ has unbounded treewidth~\cite{DalmauJ04}, it follows that $\#\homsprobd(C)$ is $\#\W{1}$-hard whenever $C$ has F-gadgets of unbounded treewidth, as we claimed. For technical reasons, we perform the reduction using the \emph{colour-prescribed} variants of the problems, where the $G$ is ``coloured'' with the vertices of the pattern\footnote{More precisely, the colouring is a homomorphism from $G$ to the pattern, see Section~\ref{sec:prelims}.}, and we must count only the homomorphisms that respect that colouring. The link with the original (i.e., uncoloured) problems are given by standard results, see~\cite{DalmauJ04,Grohe07,DorflerRSW19} as well as~\cite{Curticapean15,Roth19}.

Although F-gagdets are the technical machinery used in the proof, our hardness result can be restated in more natural terms using \emph{induced grid minors}, see Theorem~\ref{thm:intro_homs_hard}. To this end, we establish the following connection:
\begin{lemma}\label{lem:intro_induced_minor_char}
A class of graphs $C$ has F-gadgets of unbounded treewidth if and only if it has unbounded induced grid minors.
\end{lemma}
\noindent Therefore, our hardness result can be rephrased as follows: $\#\homsprobd(C)$ is $\#\W{1}$-hard for any computable class of graphs $C$ having unbounded induced grid minors.

\textbf{Complexity monotonicity and expanders.} In the second step, we use \emph{complexity monotonicity}, a principle discovered independently by Curticapean, Dell and Marx~\cite{CurticapeanDM17} and by Chen and Mengel~\cite{ChenM16}. It is well-known (cf.\ \cite[5.2.3]{Lovasz12}) that subgraph counts can be written as finite linear combinations of homomorphism counts, that is:
\begin{linenomath*}
\begin{equation}\label{eq:lincomb_intro}
    |\subs{H}{G}| = \sum_{H'} a_H(H')\cdot|\homs{H'}{G}|\,,
\end{equation}
\end{linenomath*}
where $H'$ ranges over all possible graphs, but $a_H$ is non-zero only for a finite number of them. Therefore, computing $|\subs{H}{G}|$ is \emph{at most as hard} as computing the terms $|\homs{H'}{G}|$. Now, complexity monotonicity says that computing $|\subs{H}{G}|$ is \emph{precisely as hard as} computing the single hardest term $|\homs{H'}{G}|$ for which $a_H(H') \ne 0$. This allows one to prove hardness results for $\#\subsprob$ starting from hardness results for $\#\homsprob$. This principle was originally shown for general graphs $G$, but Gishboliner et al.~\cite{Gishboliner20} showed that it holds when $G$ is degenerate as well, providing a way to link the hardness of $\#\subsprobd$ to that of $\#\homsprobd$. The same principle also holds for computing \emph{induced} subgraph counts.

To exploit this principle, we show the following result. For any graph $H$, there exists some~$H'$ such that $a_H(H') \ne 0$ and that $H'$ has an F-gadget that is a regular expander of treewidth $\Omega(\imn(H))$. Since $\#\homsprobd(C)$ is hard when $\tw(\gadgets(C))$ is unbounded, by complexity monotonicity we infer that $\#\subsprobd(C)$ is hard when $\imn(C)$ is unbounded. Besides mere hardness, we also make sure to preserve the parameters along the reductions, which yields us the fine-grained lower bounds of Theorem~\ref{thm:intro_subs_finegrained}. The argument for $\#\indsubsprobd(C)$ is the same, but in this case the F-gadget turns out to be a regular expander of treewidth $\Omega(\alpha(H))$, which yields the fine-grained lower bounds of Theorem~\ref{thm:intro_indsubs_finegrained}. 

We use the same strategy for $\#\indsubsprobd(\Phi)$. However, this requires significantly more work: given only $\Phi$, or some property enjoyed by $\Phi$ (for instance, given only that $\Phi$ is minor-closed), it is much harder to determine which coefficients vanish in the linear combination of a specific pattern $H$ that satisfies $\Phi$. To do this, we rely on the ``algebraic approach to hardness''~\cite{DorflerRSW19}. Let $|\indsubs{\Phi,k}{\star}|$ be the function that maps $G$ to the number of $k$-vertex induced subgraphs of $G$ that satisfy $\Phi$. First, we can again write $|\indsubs{\Phi,k}{\star}|$ as a linear combination of homomorphism counts:
\begin{linenomath*}
\begin{equation}\label{eq:lincomb_intro_indsubphi}
   |\indsubs{\Phi,k}{G}| = \sum_{H'} a_{\Phi,k}(H')\cdot|\homs{H'}{G}|\,.
\end{equation}
\end{linenomath*}
Note that the sum is over all possible graphs $H'$. Now, the algebraic approach to hardness states that $a_{\Phi,k}(H') \ne 0$ whenever $H'$ is an edge-transitive graph with a prime-power number of edges, and $\Phi$ distinguishes $H'$ from an independent set of size $|V(H')|$.\footnote{More precisely, the algebraic approach applies to an intermediate, vertex-coloured version of the problem. We will carefully introduce the vertex-coloured version when needed, and show that it is interreducible with the uncoloured version, even in case of degenerate graphs $G$.} Therefore, if one can find some~$H',\Phi$ that satisfy these requirements, one can apply again complexity monotonicity. We show that these requirements hold when $H'$ is the subdivision of an $(\ell,\ell)$-biclique with $\ell$ a power of $2$, and $\Phi$ is one of several interesting properties, including for example all non-trivial minor-closed properties of unbounded independence number. Then we apply complexity monotonicity; since the class of all bicliques has unbounded treewidth, and a subdivision of a biclique has the biclique as an F-gadget, we obtain hardness of $\#\indsubsprobd(\Phi)$ from our results on $\#\homsprobd(C)$, yielding Theorems~\ref{thm:intro_graphlets} and~\ref{thm:intro_indsubs_minor_closed}.

\subsection{Conclusions and Open Questions}\label{sec:conclusion}
In this work, we have given exhaustive and explicit classifications for the problems of counting subgraphs and induced subgraphs in degenerate host graphs, from the viewpoints of both fine-grained and parameterized complexity theory. The most important problem which we leave open is the question of whether a similar classification can be shown for counting homomorphisms, too. While complexity dichotomies for parameterized homomorphism counting problems in FPT and $\#\W{1}$-hard cases are not always possible~\cite{RothW20}, we conjecture that one exists in the current setting. Formally, we state the question as follows:
\begin{question*}
Find an explicit criterion on computable graph classes $C$ such that $\#\homsprobd(C)$ is \ccFPT\ if the criterion is satisfied, and $\#\homsprobd(C)$ is $\#\W{1}$-hard otherwise. 
\end{question*}
\noindent Assuming such a criterion exists, our work shows that it must lie between dag treewidth (which, if bounded, induces fixed-parameter tractability) and induced grid minors size (which, if unbounded, induces $\#\W{1}$-hardness). However, as far as we know, it might in fact be the case that dag treewidth and induced grid minor size are equivalent measures in the sense that a class of graphs $C$ has bounded dag treewidth if and only if it has bounded induced grid minors. Let us point out the similarity to treewidth and (not necessarily induced) grid minors: By the Excluded-Grid-Theorem~\cite{RobertsonS86-ExGrid}, a class of graphs has bounded treewidth if and only if it has bounded grid minors. If a similar statement is true for dag treewidth and induced grid minors, we expect the proof to be of significant difficulty, requiring multiple novel structural insights on dag treewidth. The reason for the latter is the fact that induced grid minor size is a measure that seems, a priori, unrelated to orientations of edges, while the dag treewidth originates from DAGs and only extends to undirected graphs by considering all acyclic orientations.

A second direction for future research is provided by the absence of hardness results for approximate pattern counting in degenerate graphs. As mentioned earlier, complexity dichotomies for approximate counting problems are not always possible~\cite{DyerGJ10}, and even for not necessarily degenerate host graphs, the complexity of approximating subgraph counts is only partially resolved~\cite{ArvindR02}. Consequently, a complete picture of the complexity of approximate pattern counting in degenerate graphs seems to be elusive at this point. Nevertheless, we suggest the induced matching problem as a concrete starting point for further research on that matter: 
\begin{question*}
Let $k$ be a positive integer and let $G$ be an $n$-vertex graph of degeneracy $d$. Is it possible to compute (with high probability) an $\varepsilon$-approximation of the number of induced $k$-matchings in $G$ in time $f(k,d)\cdot \mathsf{poly}(|V(G)|,\varepsilon^{-1})$ for some computable function $f$?  
\end{question*}
Note that we chose the problem of counting induced $k$-matchings in degenerate graphs, since exact counting is hard by Theorem~\ref{thm:intro_indsubs_param}, but its decision version is known to be fixed-parameter tractable~\cite{ErmanKKW10}. Furthermore, $k$-matchings constitute the minimal class of graphs for which our algorithm for approximating induced subgraph counts (see Theorem~\ref{thm:intro_FPTRAS_indsubcount}) does not yield fixed-parameter tractability.

%% file: f-gadget-3.tikz
\begin{tikzpicture}[
  scale=1.3,
  every node/.style={circle, inner sep=0pt, minimum size=12},
  halfspace/.style={left color=white, right color=white, path fading=west, fill opacity=1}
]
\pgfdeclarelayer{foreground}
\pgfsetlayers{background,main,foreground}

\begin{scope}[shift={(-4,1)}]
\node[graph] (a) at (0,0) {$a$};
\node[graph] (b) at (0,-1) {$b$};
\node[graph] (c) at (1,-1) {$c$};
\draw (a) -- (b) -- (c) -- (a);
\end{scope}

\begin{scope}[shift={(-1,0)},scale=1.5]
\begin{pgfonlayer}{foreground}
\node[graph] (Ha1) at (0,1.2) {$1$};
\node[graph] (Ha2) at (.4,1.2) {$3$};
\node[graph] (Ha3) at ($(Ha1)!.5!(Ha2)+(90:.3)$) {$2$};
\node[graph] (Hb1) at (0,0) {$5$};
\node[graph] (Hc1) at (1.2,0) {$7$};
\node[graph] (Hc2) at (1.5,.3) {$10$};
\node[graph] (Hc3) at (1.5,-.3) {$8$};
\node[graph] (Hc4) at (1.8,0) {$9$};
\node[graph] (Hab) at ($(Ha1)!.5!(Hb1)$) {$4$};
\node[graph] (Hbc) at ($(Hb1)!.5!(Hc1)$) {$6$};
\node[graph] (Hac) at ($(Ha2)!.5!(Hc2)$) {$11$};
\node[graph] (Hr) at (1.5,1.2) {$12$};
\draw (Ha1) -- (Ha2) -- (Ha3) -- (Ha1);
\draw (Ha1) -- (Hab) -- (Hb1) -- (Hbc) -- (Hc1) -- (Hc2) -- (Hc3) -- (Hc1) -- (Hc4) -- (Hc2) -- (Hac) -- (Ha2) -- (Hr) -- (Hc2);
\draw (Hr) -- (Hac);
\end{pgfonlayer}
\node (H) at ($(Hac)+(1,.8)$) {};
\end{scope}

\begin{scope}[shift={(4,0)},scale=1.5]
\begin{pgfonlayer}{foreground}
\node[graph] (Ha1) at (0,1.2) {$1$};
\node[graph] (Ha2) at (.4,1.2) {$3$};
\node[graph] (Ha3) at ($(Ha1)!.5!(Ha2)+(90:.3)$) {$2$};
\node[graph] (Hb1) at (0,0) {$5$};
\node[graph] (Hc1) at (1.2,0) {$7$};
\node[graph] (Hc2) at (1.5,.3) {$10$};
\node[graph] (Hc3) at (1.5,-.3) {$8$};
\node[graph] (Hc4) at (1.8,0) {$9$};
\node[graph] (Hab) at ($(Ha1)!.5!(Hb1)$) {$4$};
\node[graph] (Hbc) at ($(Hb1)!.5!(Hc1)$) {$6$};
\node[graph] (Hac) at ($(Ha2)!.5!(Hc2)$) {$11$};
\node[graph] (Hr) at (1.5,1.2) {$12$};
\draw (Ha1) -- (Ha2) -- (Ha3) -- (Ha1);
\draw (Ha1) -- (Hab) -- (Hb1) -- (Hbc) -- (Hc1) -- (Hc2) -- (Hc3) -- (Hc1) -- (Hc4) -- (Hc2) -- (Hac) -- (Ha2) -- (Hr) -- (Hc2);
\draw (Hr) -- (Hac);
\end{pgfonlayer}
\end{scope}

\pgfmathanglebetweenpoints{\pgfpointanchor{Ha2}{center}}{\pgfpointanchor{Hc2}{center}}
\edef\angleAC{\pgfmathresult}
\path[stainEdge,rounded corners=9] ($(Ha1)!.3!(Hb1)+(-.2,.25)$) rectangle ($(Ha1)!.7!(Hb1)+(.2,-.25)$) {};
\path[stainEdge,rounded corners=9] ($(Hb1)!.33!(Hc1)+(-.3,.2)$) rectangle ($(Hb1)!.67!(Hc1)+(.3,-.2)$) {};
\path[stainEdge,rounded corners=9,rotate=\angleAC] ($(Ha2)!.27!(Hc2)+(-.3,.2)$) rectangle ($(Ha2)!.73!(Hc2)+(.3,-.2)$) {};
\draw[stain] plot[smooth cycle,tension=.9] coordinates {($(Ha1)+(-.25,-.15)$) ($(Ha2)+(.25,-.15)$) ($(Ha3)+(0,.25)$)};
\draw[stain] plot[smooth cycle,tension=.5] coordinates {($(Hc1)+(-.3,0)$) ($(Hc2)+(0,.3)$) ($(Hc4)+(.3,0)$) ($(Hc3)+(0,-.3)$)};
\draw[stain] (Hb1) circle (.3);
\draw[stain,fill=white] (Hr) circle (.35);
\node[above left=13pt] (Sa) at (Ha1) {$S_a$};
\node[below left=11pt] (Sb) at (Hb1) {$S_b$};
\node[above right=12pt] (Sc) at (Hc2) {$S_c$};
\node[left=9pt] (Pab) at (Hab) {$P_{ab}$};
\node[below=8pt] (Pbc) at (Hbc) {$P_{bc}$};
\node[below left=7pt] (Pac) at (Hac) {$P_{ac}$};
\node[above=13pt] (R) at (Hr) {$R$};

\end{tikzpicture}

%% file: prelims.tex
\section{Preliminaries}\label{sec:prelims}

\subparagraph*{Graph Theory, Homomorphisms, etc.}
Graphs are simple and without self-loops, unless stated otherwise. Given a graph $G=(V,E)$, a \emph{subgraph} of $G$ is a graph obtained by deleting edges and vertices, and an \emph{induced subgraph} of $G$ is a graph obtained by deleting vertices only. Given a graph $H$, we write $\subs{H}{G}$ ($\indsubs{H}{G}$) for the set of all subgraphs (induced subgraphs) of $G$ isomorphic to $H$.
Given a vertex subset $S\subseteq V$, we write $G[S]$ for the graph induced by $S$, that is $V(G[S])=S$ and $E(G[S])= E(G) \cap S^2$. Given an edge subset $S \subseteq E$, we write $G[S]$ for the \emph{edge-subgraph} $(V,S)$ of $G$.

A \emph{homomorphism} from a graph $H$ to a graph $G$ is a mapping $\varphi:V(H) \rightarrow V(G)$ that preserves the edges, i.e., for each edge $\{u,v\}\in E(H)$ we have $\{\varphi(u),\varphi(v)\}\in E(G)$. An injective homomorphism is called an \emph{embedding}, and an embedding is called a \emph{strong embedding} if $\{u,v\}\in E(H)$ is equivalent to $\{\varphi(u),\varphi(v)\}\in E(G)$. We write $\homs{H}{G}$, $\embs{H}{G}$, and $\strembs{H}{G}$ for the sets of homomorphisms, embeddings, and strong embeddings, respectively, from $H$ to $G$. Finally, we write $\auts{H}$ for the set of all automorphisms of $H$.
We will use the following well-known relationships:
\begin{fact}
For every pair of graphs $H$ and $G$, we have $|\embs{H}{G}|=|\auts{H}|\cdot|\subs{H}{G}|$ and  $|\strembs{H}{G}|=|\auts{H}|\cdot|\indsubs{H}{G}|$.
\end{fact}

We will also need the following transformations.
Given a partition $\rho$ of $V(G)$, the \emph{quotient graph} (or just the \emph{quotient}) $G/\rho$ is obtained from $G$ by identifying all vertices that are in the same block of $\rho$ and deleting multiple edges. We point out that $G/\rho$ might, a priori, have self-loops, but will usually only consider those quotients without self-loops.\footnote{A quotient without self-loops is called a \emph{spasm} by some authors, see for instance~\cite{CurticapeanDM17}.}
Given a function $f:A\times B \rightarrow C$ and an element $a\in A$, we write $f(a,\star):B \rightarrow C$ for the function that maps $b\in B$ to $f(a,b)$.
\begin{fact}[See Chapt.\ 5.2.3 in~\cite{Lovasz12}]\label{fact:sub_hom_basis}
Let $H$ be a graph. We have
\begin{linenomath*}
\begin{equation}
|\embs{H}{\star}| =  \sum_{\rho}\mu(\bot,\rho) \cdot  |\homs{H/\rho}{\star}|\,,
\end{equation}
\end{linenomath*}
where the sum is taken over all partitions of $V(H)$, $\bot$ is the finest partition, and
\begin{linenomath*}
\begin{equation}
\mu(\bot,\rho) = (-1)^{|V(H)|-|V(H/\rho)|} \cdot \prod_{B\in \rho}(|B|-1)!
\end{equation}
\end{linenomath*}
is the M\"obius function of the partition lattice of $V(H)$.
\end{fact}
\begin{fact}[See Chapt.\ 5.2.3 in~\cite{Lovasz12}]\label{fact:indsub_sub_basis}
Let $H$ be a graph. We have
\begin{linenomath*}
\begin{equation}
|\strembs{H}{\star}| = \sum_{H'}(-1)^{|E(H')|-|E(H)|} \cdot |\{H' \supseteq H\}|\cdot  |\embs{H'}{\star}|\,,
\end{equation}
\end{linenomath*}
where the sum is taken over all (isomorphism classes of) graphs, and $\{H' \supseteq H\}$ contains all sets $S\in V(H)^2\setminus E(H)$ such that $(V(H),E(H)\cup S)$ is isomorphic to $H'$.
\end{fact}

A graph $G$ is $H$\emph{-coloured} if it comes with a homomorphism $c\in \homs{G}{H}$, called \emph{colouring}. A homomorphism $\varphi \in \homs{H}{G}$ is \emph{colour-prescribed} if $c(\varphi(v))=v$ for all $v\in V(H)$. 
This means that each $v \in V(H)$ must be mapped to a vertex of $G$ that is coloured with $v$. We write $\cphoms{H}{G}$ for the set of all colour-prescribed homomorphisms from $H$ to $G$. Moreover, $\varphi$ is called \emph{colourful} if $c(\varphi(V(H)))=V(H)$, that is, each colour is met precisely once. We write $\cfhoms{H}{G}$ for the set of all colourful homomorphisms from $H$ to $G$. We always assume that colourings are surjective, since otherwise $\cphoms{H}{G}=\cfhoms{H}{G}=\emptyset$.
For any $S\subseteq V(H)$, let $G-S$ be the graph obtained from $G$ by deleting all vertices coloured with some vertex in $S$.
We will use the following two well-known facts (see for instance~\cite[Section~2.5]{Roth19}).

\begin{fact}\label{fact:cphoms_to_cfhoms}
$|\cphoms{H}{G}| = |\auts{H}|^{-1} \cdot |\cfhoms{H}{G}|$.
\end{fact}

\begin{fact}\label{fact:cfhoms_to_homs} $|\cfhoms{H}{G}| = \sum_{S\subseteq V(H)}(-1)^{|S|} \cdot |\homs{H}{G-S}|$.
\end{fact}

A graph $F$ is a \emph{minor} of a graph $H$ if it can be obtained from $H$ by a sequence of edge contractions, edge deletions, and vertex deletions. It is well known that minors can be equivalently defined via \emph{models}: a model of $F$ in $H$ is a mapping from $V(F)$ to nonempty disjoint subsets of $V(H)$, that we call \emph{blocks} and denote by $\{B_u : u \in V(F)\}$, such that $H[B_u]$ is connected for all $u \in V(F)$, and that $H$ has an edge between $B_u$ and $B_v$ if $\{u,v\}\in E(F)$. We emphasise that $B_u \ne \emptyset$ for all $u \in V(F)$.
A graph $F$ is an \emph{induced minor} of a graph $H$ if $F$ can be obtained from $H$ by a sequence of edge contractions and vertex deletions. Induced minors can be expressed via \emph{witness structures} (cf.\ \cite{Hofetal12}), which are identical to models save for the fact that $H$ has an edge between $B_u$ and $B_v$ if \emph{and only if} $\{u,v\}\in E(F)$.

\subparagraph*{Parameterized and Fine-grained Complexity Theory.}
We give a concise introduction to parameterized counting problems and fine-grained complexity theory. For more details see, e.g.,~\cite[Chapter 14]{FlumG06}.

A \emph{parameterized counting problem} is a pair $(P,\kappa)$, where $P:\{0,1\}^\ast\rightarrow \mathbb{N}$ is a counting problem and $\kappa:\{0,1\}^\ast\rightarrow \mathbb{N}$ is a parameterization. For example, the problem $\#\clique$ expects as input a graph $G$ and a positive integer $k$, and the task is to compute the number of $k$-cliques in~$G$, that is, $P(G,k):=|\subs{K_k}{G}|$. The parameterization is given by $k$, that is, $\kappa(G,k):=k$.
A parameterized (counting) problem is called \emph{fixed-parameter tractable} (FPT) if there exists a computable function $f$ such that the problem can be solved in time $f(\kappa(x))\cdot |x|^{O(1)}$, where $x\in \{0,1\}^\ast$ is the input instance. An algorithm running in this time is called an \emph{FPT algorithm}.
A \emph{parameterized Turing-reduction} from $(P,\kappa)$ to $(P',\kappa')$ is an FPT algorithm $\algo$ equipped with oracle access to $P'$ that computes $P(x)$. Additionally, there must be a computable function~$f$ such that for any input $x$, the parameter $\kappa'(y)$ of each oracle query $y$ must be bounded by $f(\kappa(x))$. We write $(P,\kappa)\fptred (P',\kappa')$ if a parameterized Turing-reduction exists.
The notion of parameterized intractability is given by $\#\W{1}$-hardness:\footnote{We use the definition via a complete problem and refer the reader to~\cite[Chapter 14]{FlumG06} for a more structural and machine-based definition of $\#\W{1}$.} A parameterized counting problem $(P,\kappa)$ is $\#\W{1}$\emph{-hard} if $\#\clique \fptred (P,\kappa)$. It is known that $\#\W{1}$-hard problems are not fixed-parameter tractable unless ETH, defined below, fails.

\begin{conjecture}[ETH~\cite{ImpagliazzoP01}]
The \emph{Exponential Time Hypothesis} (ETH) asserts that the problem $3\textsc{-SAT}$ cannot be solved in time $\exp(o(n))$ where $n$ is the number of variables of the input formula.
\end{conjecture}

Our hardness results will rely on the following theorem, which is an immediate consequence of the works of Dalmau and Jonsson~\cite{DalmauJ04} and Marx~\cite{Marx10}.
We include a proof for the readers' convenience.
\begin{theorem}\label{thm:hardness_bottleneck}
    Let $C$ be a computable class of graphs. If $C$ has unbounded treewidth\footnote{We will exclusively rely on treewidth in a black-box manner. Thus we avoid stating the definition and refer the reader e.g.\ to Chapter~7 in~\cite{CyganFKLMPPS15}.} then $\#\cphomsprob(C)$ is $\#\W{1}$-hard and cannot be solved in time
    $f(|H|)\cdot |G|^{o(\mathsf{tw}(H)/\log \tw(H))}$
    for any function $f$, unless ETH fails.
\end{theorem}
\begin{proof}
    The $\#\W{1}$-hardness follows by the $\#\W{1}$-hardness of $\#\homsprob(C)$, see~\cite{DalmauJ04}, and by a reduction from $\#\homsprob(C)$ to $\#\cphomsprob(C)$, see~\cite{DorflerRSW19}.
    The lower bound follows from a result on the partitioned subgraph isomorphism problem, see~\cite{Marx10}, which tightly reduces to $\#\cphomsprob(C)$, see~\cite[Section~2]{RothSW20}).
\end{proof}
We will apply Theorem~\ref{thm:hardness_bottleneck} to the following family $\mathcal{F}$ of regular expanders:
\begin{theorem}[see~\cite{Grohe&2009expansion}, Theorem 8]\label{thm:nice_expanders}
There exists a family of graphs $\mathcal{F}=\{F_k\}_{k\geq 1}$ and constants $k_0,c>0$ such that for all $k\geq k_0$ we have $\mathsf{tw}(F_k)\geq k$ and $|V(F_k)|\leq ck$ and $|E(F_k)|\leq ck$. 
\end{theorem}

\subparagraph*{Degeneracy and dag tree decompositions.}
The \emph{degeneracy} $d(G)$ of a graph $G$ can be defined as the smallest integer $d$ such that there exists an acyclic orientation of $G$ with maximum out-degree bounded by $d$. We can assume that such an orientation is given, since it can be found in time $O(|E(G)|)$~\cite{NesetrildM12}.

Our running time upper bounds rely on the dag tree decomposition of~\cite{Bressan21} and the dynamic program that it yields. Here, we adopt a generalisation of the original definition of dag tree decomposition. We need this generalisation for some of our results in Section~\ref{sec:hom_ub}. Let $\orient{H}=(V,A)$ be any directed acyclic graph. We denote by $S=S(\orient{H})$ the set of sources of $\orient{H}$, that is, the vertices of indegree zero. For any $B \subseteq V$, we let $\orient{H}(B)$ be the subgraph of $\orient{H}$ induced by all vertices reachable from some vertex in $B$. For any tree $\scT$ and any two vertices $B',B''$ of $\scT$, we denote by $\scT(B',B'')$ the unique simple path between $B'$ and $B''$ in $\scT$.

\begin{definition}
\label{def:gdtd}
A (generalised) dag tree decomposition (d.t.d.) of $\orient{H}$ is a rooted tree $\scT=(\bags,\et)$ such that: 
\begin{enumerate}[itemsep=4pt,parsep=0pt,topsep=2pt]
\item $B \subseteq V(H)$ for all $B \in \bags$
\item $\bigcup_{B \in \bags} V(\orient{H}(B)) = V(H)$
\item \label{pr:joint_path} for all $B,B_1,B_2 \in \bags$, if $B \in \scT(B_1,B_2)$ then $V(\orient{H}(B_1)) \cap V(\orient{H}(B_2)) \subseteq V(\orient{H}(B))$
\end{enumerate}
The width of $\scT$ is $\dtw(\scT)=\max_{B \in \scT}|B|$. The (generalised) dag treewidth of $\orient{H}$ is $\dtw(\orient{H}) = \min_{\scT} \dtw(\scT)$ where the minimum is over all possible (generalised) d.t.d.'s of $\orient{H}$.
\end{definition}

It is immediate to verify that if $\scT$ is a d.t.d.\ of $\orient{H}$ according to~\cite{Bressan21} then $\scT$ is a generalised d.t.d.\ according to Definition~\ref{def:gdtd} as well. It is easy to verify that the dynamic program of~\cite{Bressan21} and all subsequent bounds continue to hold when $\scT$ is a generalised d.t.d., too. Hence for our upper bounds we can use $\dtw(\orient{H})$ with $\dtw$ as in Definition~\ref{def:gdtd}.

Now let $H$ be any simple graph. We let $\Sigma(H)$ be the set of all possible acyclic orientations $\orient{H}$ of $H$. We let $R(H)$ be the set of all partitions $\rho$ of $V(H)$ such that the quotient $H/\rho$ is without self-loops. We let $D(H)$ be the set of all supergraphs of $H$ on the same vertex set $V(H)$. Following~\cite{Bressan21}, we define:
\begin{linenomath*}
\begin{align}
    \dtw_1(H) &= \max\left\{ \dtw(\orient{H}) : \orient{H} \in \Sigma(H) \right\}
    \\
    \dtw_2(H) &= \max\left\{ \dtw_1(H/\rho) : \rho \in R(H) \right\}
    \\
    \dtw_3(H) &= \max\left\{ \dtw_2(H') : H' \in D(H) \right\}
\end{align}
\end{linenomath*}
By the discussion above, the parameters $\dtw_1,\dtw_2,\dtw_3$ defined above in terms of Definition~\ref{def:gdtd} are bounded from above by their homologues defined in~\cite{Bressan21}. Therefore from now on we omit the adjective ``generalised'' unless necessary.

%% file: gadgets.tex
\section{F-gadgets}
\label{sec:gadgets}
In this section we introduce our main technical tool: F-gadgets. Informally, a graph $H$ has an $F$-gadget if, by exploding the vertices and subdividing the edges of $F$, we can obtain an induced subgraph of $H$. We will show that $\#\homsprobd(C)$ is hard if $C$ contains graphs that have F-gadgets of unbounded treewidth. As the reader will have noticed, we write ``F-gadget'' to denote the concept and $F$-gadget to specify a graph $F$.

Before giving a formal definition of F-gadgets, we use a toy example to provide some intuition. Suppose that we want to count the copies of $H=C_6$, the cycle on $6$ vertices, in graphs of bounded degeneracy. As shown in~\cite{Bera-ITCS20,Bera-SODA21}, this problem is at least as hard as counting $F=C_3$ in arbitrary graphs. To see why this is the case, consider the following reduction. Given $G$, we compute its $1$-subdivision $G'$. It is immediate to see that $d(G') \le 2$, so $G'$ has bounded degeneracy, and that the copies of $H$ in $G$ are in one-to-one correspondence with the copies of $F$ in $G'$. Thus, counting the copies of $F$ in $G$ reduces to counting the copies of $H$ in $G'$.

The purpose of F-gadgets is to be an obstruction that captures this kind of hardness and applies to counting homomorphisms. That is, the hardness of counting, in a bounded-degeneracy graph $G'$, the homomorphisms from a graph $H$ which looks like a subdivision $F$. However, for technical reasons, the actual definition of $H$ having an $F$-gadget is fairly more complex than just $H$ being the $1$-subdivision of $F$.
\begin{definition}
\label{def:gadgets}
Given two graphs $F$ and $H$, we say that $H$ has an $F$-gadget if $V(H)$ can be partitioned as:\footnote{Formally, this is not a partition since $R$ might be empty, but we decided to abuse notation here for the sake of readability.}
\begin{linenomath*}
\begin{equation}
    V(H) = \dot{\bigcup_{v\in V(F)}} S_v ~~~\dot\cup~~~ \dot{\bigcup_{e\in E(F)}} P_e ~~~\dot\cup~~~ R
    \label{eq:Fgadget_partition}
\end{equation} 
\end{linenomath*}
under the following constraints:
\begin{enumerate}[itemsep=2pt,parsep=0pt,topsep=2pt]
	\item $\forall\, v\in V(F)$: $S_v \ne \emptyset$ and $H[S_v]$ is connected.
	\item $\forall\, e\in E(F)$: $P_e \ne \emptyset$, and if $e=\{u,v\}$, then there are two vertices $s_{u,e} \in S_u$ and $s_{v,e} \in S_v$ such that the subgraph $p_e := H[\{s_{u,e}\} \cup P_e \cup \{s_{v,e}\}]$ is a simple path with endpoints $s_{u,e}$ and $s_{v,e}$.
	\item $\forall\, e_h \in E(H)$, either $e_h \in H[S_v]$ for some $v \in V(F)$, or $e_h \in p_e$ for some $e \in E(F)$, or $e_h$ is incident to a vertex of $R$.
\end{enumerate}
If this is the case, then we also say that the triple $(\mathcal{S},\mathcal{P},R)$, where $\mathcal{S}=\{S_v\}_{v \in V(F)}$ and $\mathcal{P}=\{P_e\}_{e \in E(F)}$, is an $F$-gadget of $H$.
\end{definition}

\noindent Recalling the example above, if $H$ is the $1$-subdivision of $F$, then it is easy to see that $H$ has an $F$-gadget: for each $v\in V(F)$, set $S_v:=\{v\}$, and for each $e\in E(F)$, set $P_e:=\{v_e\}$, where $v_e$ is the vertex of $H$ corresponding to the edge $e$ of $F$, and set $R:=\emptyset$. 
This holds more in general if~$H$ is obtained from $F$ by subdividing each edge at least once.
For another example of F-gadget, see Figure~\ref{fig:F_gadget_intro}.

We need one last piece of notation. Let $C$ be any class of graphs.
\begin{definition}
The set of F-gadgets of $C$ is $\gadgets(C) = \left\{ F \,:\, \exists H \in C \,:\, H \text{ has an } F\text{-gadget} \right\}$.
\end{definition}
For instance, if $C = \{K_n : n \in \mathbb{N}^+ \}$, then $\gadgets(C)$ contains only the singleton graph.

\subsection{Hardness via F-gadgets}
Using F-gadgets, we show that the problem of counting colour-prescribed homomorphisms from~$F$ to~$G$ can be reduced to counting colour-prescribed homomorphisms from $H$ to $G'$, where $H$ has an $F$-gadget and $G'$ is a graph of degeneracy roughly $|V(H)|$.
\begin{lemma}[F-gadget reduction lemma]\label{lem:FGadget_main}
    For some computable function $f$ there exists a deterministic algorithm $\algo$ that, given in input a triple $(F,H,G)$ such that $H$ has an $F$-gadget and $G$ is $F$-coloured,
    in time $f(|H|)\cdot |G|^{O(1)}$ computes an $H$-coloured graph $G'$ such that $d(G')\leq |V(H)|+2$, $|G'| = O(|G|\cdot |H|)$, and $|\cphoms{F}{G}| = |\cphoms{H}{G'}|$.
\end{lemma}
\begin{proof}
    The algorithm $\algo$ works as follows. First, we compute an $F$-gadget $(\mathcal{S},\mathcal{P},R)$ of $H$, see Definition~\ref{def:gadgets}. Note that $(\mathcal{S},\mathcal{P},R)$ exists since by hypothesis $H$ has an $F$-gadget, and the time needed to compute it is a function of $H$. Next, we construct an $H$-coloured graph $G'$, whose colouring is denoted by $c_H$, as follows. We start with $G'$ being the empty graph, and:
	\begin{enumerate}
		\item For each $g\in V(G)$, we add to $G'$ a copy $H[S_v]^g$ of $H[S_v]$, where $v = c_F(g)$.
		For each vertex $h_v^g \in H[S_v]^g$, we set $c_H(h_v^g) := h_v$, where $h_v \in S_v$ is the vertex of which $h_v^g$ is a copy.
		\item For each $\{g,g'\}\in E(G)$ let $(u,v):=(c_F(g),c_F(g'))$. Note that $e:=\{u,v\}\in E(F)$, since $c_F$ is a homomorphism. We add to $G'$ a copy $H[P_e]^{\{g,g'\}}$ of $H[P_e]$. Furthermore, for each vertex $h_e^{\{g,g'\}} \in H[P_e]^{\{g,g'\}}$, we set $c_H(h_e^{\{g,g'\}}):=h_e$, where $h_e \in P_e$ is the vertex of which $h_e^{\{g,g'\}}$ is a copy. Now let $p_e := H[\{s_{u,e}\} \cup P_e \cup \{s_{v,e}\}]$. We add an edge between $s_{u,e}^g \in H[S_u]^g$ and $h_{e}^{\{g,g'\}}$, where $h_e \in P_e$ is the vertex of $P_e$ connected to $s_{u,e}$ in $H$. We do the same with $s_{v,e}^{g'}$.
		\item We add to $G'$ a copy $H[R]'$ of $H[R]$, and for each vertex $r' \in H[R]'$, we set $c_H(r')=r'$. Finally, for each edge $\{r,h\} \in E(H)$ with $r\in R$ and $h \notin R$, we add to $G'$ the edge $\{r',h'\}$ between the copy $r'$ of $r$ and any copy $h'$ of $h$.
	\end{enumerate}
	A pictorial example is given in Figure~\ref{fig:cpred}
	\begin{figure}
    \centering
    \input{cpred.tikz}
    \caption{the construction of $G'$. Top: a graph $F$ and (a part of) a graph $G$ that is $F$-coloured. Bottom: a graph $H$ that has an $F$-gadget, and (a part of) the $H$-coloured graph $G'$ computed by the reduction.}
    \label{fig:cpred}
    \end{figure}
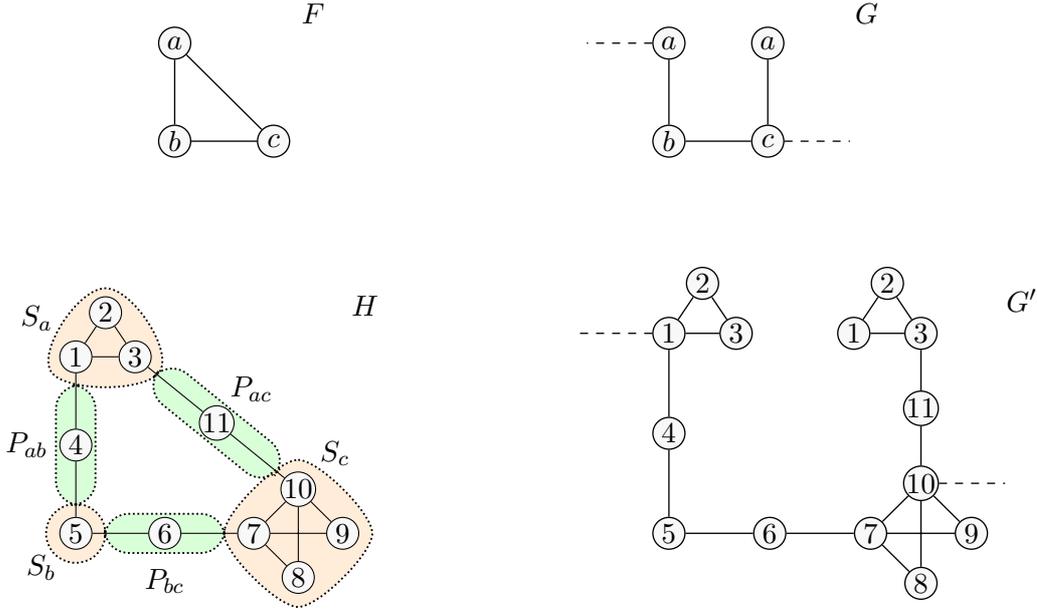
	Observe that $|G'| = O(|G| \cdot |H|)$, since for every vertex and/or edge of $G$ we add to $G'$ a subgraph of $H$. This proves the running time bound.
	It remains to prove (i) that $c_H$ is a valid $H$-colouring of $G$, (ii) that $d(G') \le |V(H)|+2$, (iii) that $|\cphoms{F}{G}| = |\cphoms{H}{G'}|$.
\begin{claim}
	$c_H$ is an $H$-colouring of $G'$.
\end{claim}
\begin{claimproof}
    By definition of $H$-colouring, we are claiming that $c_H\in \homs{G'}{H}$. By construction, $c_H$ is a mapping from $V(G')$ to $V(H)$. Thus, it remains to show that $c_H$ is edge-preserving. Let $e'\in E(G')$. According to the step of the construction of $G'$ where $e'$ is created:
	\begin{enumerate}
		\item $e' \in H[S_v]^g$ for some $v\in V(F)$ and $g\in V(G)$.
		Since $H[S_v]^g$ is a copy of $H[S_v]$, then $e'=\{h_v^g,\hat{h}_v^g\}$ where $\{h_v,\hat{h}_v\}$ is an edge of $H[S_v]$.
		Since by construction we set precisely $c_H(h_v^g)=h_v$ and $c_H(\hat{h}_v^g)=\hat{h}_v$, in this case $c_H$ preserves $e'$.
		\item $e' \in H[P_e]^{\{g,g'\}}$, or $e'$ connects $s_{u,e}^g \in H[S_u]^g$ to $H[P_e]^{\{g,g'\}}$, for some $u \in V(F)$ and $e \in E(F)$, and for some $g,g' \in V(G')$.
		\begin{enumerate}
		    \item if $e' \in H[P_e]^{\{g,g'\}}$, then $e'=\{h_e^{\{g,g'\}}, \hat{h}_e^{\{g,g'\}}\}$ where $\{h_e, \hat{h}_e\}$ is an edge of $H[P_e]$.
    		Since we set $c_G(h_e^{\{g,g'\}})=h_e$ and $c_G(\hat{h}_e^{\{g,g'\}})=\hat{h}_e$, in this case $c_H$ preserves $e'$.
		    \item if $e'$ connects $s_{u,e}^g \in H[S_u]^g$ to $H[P_e]^{\{g,g'\}}$, then $e'=\{s_{u,e}^g,h_e^{\{g,g'\}}\}$.
		    In this case, again by construction, $\{s_{u,e},h_e\} \in E(H)$, and we set precisely $c_H(s_{u,e}^g) = s_{u,e}$ and $c_H(h_e^{\{g,g'\}})=h_e$. Thus, in this case $c_H$ preserves $e'$.
		\end{enumerate}
        \item $e' \in H[R]'$, or $e'$ connects some $r' \in H[R]'$ to a copy $h^g$ of some vertex $h \in V(H)$.
        \begin{enumerate}
            \item if $e' \in H[R]'$ then we are done since $c_H$ is the identity on $V(H[R]')$.
            \item if $e'$ connects some $r' \in H[R]'$ to $h^g$, then, by construction $(r,h) \in E(H)$, where $r \in R$ is the vertex of which $r'$ is a copy. So in this case $c_H$ preserves $e'$ as well.
        \end{enumerate}
	\end{enumerate} 
This concludes the proof.
\end{claimproof}
\begin{claim}
	$d(G')\leq |V(H)|+ 2$
\end{claim}
\begin{claimproof}
    Order the vertices of $G'$ as follows. First, any vertex $u$ such that $c_H(u) \in R$.
	Next, any vertex $u$ such that $c_H(u) \in S_v$ for some $v\in V(F)$.
	Finally, any vertex $u$ such that $c_H(u)\in P_e$ for some $e\in E(F)$.
	We denote the three sets as  $\mathcal{R}'$, $\mathcal{S}'$, and $\mathcal{P}'$.
	We claim that, in this ordering, any vertex $u$ has at most $|V(H)|+2$ neighbors preceding it.
	This is straightforward to see if $u \in \mathcal{R}'$. Suppose instead that $u \in \mathcal{S}'$.
	Then, by the construction of $G'$, we have $u = h_v^g$ and $h_v^g$ has at most $|S_v^g|-1 + |R| = |S_v|-1+|R| \leq |V(H)|$ neighbours in $\mathcal{R}'\cup \mathcal{S}'$.
	Finally, suppose that $u \in \mathcal{P}'$.
	Then $u$ has at most $|R|+2\leq |V(H)|+2$ neighbors in total --- the two neighbours in the copy of a path $p_e$, and vertices in $R$.
\end{claimproof}
\begin{claim}
	$|\cphoms{F}{G}| = |\cphoms{H}{G'}|$.
\end{claim}
\begin{claimproof}
First, we show that $|\cphoms{F}{G}| \leq |\cphoms{H}{G'}|$.
To this end, for each $\varphi \in \cphoms{F}{G}$ we show some $\hat{\varphi} \in \cphoms{H}{G'}$ so that the resulting map $\mu : \varphi \mapsto \hat\varphi$ is injective.
Let $h \in V(H)$.
If $h \in S_v$ for some $v \in V(F)$, then $\hat{\varphi}(h)=h^{g}_v$, where $g=\varphi(v)$.
If $h \in P_e$ for some $e \in E(F)$, then $\hat{\varphi}(h) = h^{\{g,g'\}}_e$ where $\{g,g'\}=\varphi(e)$.
If $h \in R$, then $\hat{\varphi}(h)=h$.
Using the construction of $G'$ one can check that $\hat{\varphi}\in \cphoms{H}{G'}$, and that $\mu : \varphi \mapsto \hat{\varphi}$ is injective.

    Second, we show that $|\cphoms{H}{G'}|\leq |\cphoms{F}{G}|$.
	As above, for each $\hat{\varphi} \in \cphoms{H}{G'}$ we show some $\varphi \in \cphoms{F}{G}$ so that the resulting map $\mu': \hat\varphi \mapsto \varphi$ is injective.
	In fact, we will use $\mu' = \mu^{-1}$.
	Let then $\hat{\varphi} \in \cphoms{H}{G'}$.
	The crucial observation is that for any $v \in V(F)$ we must have $\hat{\varphi}(S_v)=V(H[S_v]^g)$ for some $g$ such that $c_F(g)=v$.
	Similarly, for any $e=\{u,v\} \in E(F)$ we must have $\hat{\varphi}(P_e) = H[P_e]^{\{g,g'\}}$ where $c_F(\{g,g'\})=e$.
	This holds because $\hat\varphi \in \cphoms{H}{G'}$ and by the construction of $G'$.
	Therefore, for any $v \in V(F)$ we simply let $\varphi(v) := g$ if $\hat{\varphi}(S_v)=S^g_v$. The argument above implies that $\varphi$ is colour-prescribed and  edge-preserving.
    Finally, let us show that $\mu' : \hat\varphi \mapsto \varphi$ is injective. Assume indeed that $\hat{\varphi} \neq \hat{\varphi}'$ for some $\hat{\varphi},\hat{\varphi}'\in \cphoms{H}{G'}$, and let $\varphi:=\mu'(\hat{\varphi})$ and $\varphi':=\mu'(\hat{\varphi}')$. Note that $\hat{\varphi}$ and $\hat{\varphi}'$ cannot differ on $R$ as they are colour-prescribed and thus the identity on $R$. If they differ on some $h \in P_e$ for some $e$, they must also differ on some $h \in S_v$ for some $v\in V(F)$ to which $e$ is incident, as otherwise one of $\hat{\varphi}$ and $\hat{\varphi}'$ would not violated its own colour prescription. Consequently, we have $\varphi \neq \varphi'$, so $\mu'$ is injective.
	
	For completeness we point out, as anticipated, that $\mu'=\mu^{-1}$, as one can check. 
\end{claimproof}
\noindent This concludes the proof of Lemma~\ref{lem:FGadget_main}.
\end{proof}

%% file: cpred.tikz
\begin{tikzpicture}[
  scale=1.3,
  every node/.style={circle, inner sep=0pt, minimum size=12},
  halfspace/.style={left color=white, right color=white, path fading=west, fill opacity=1}
]
\pgfdeclarelayer{foreground}
\pgfsetlayers{background,main,foreground}

\begin{scope}[shift={(-5,0)}]
\node[graph] (a) at (0,0) {$a$};
\node[graph] (b) at (0,-1) {$b$};
\node[graph] (c) at (1,-1) {$c$};
\draw (a) -- (b) -- (c) -- (a);
\node (F) at (1.4,.3) {$F$};
\end{scope}

\begin{scope}
\node[graph] (ga1) at (0,0) {$a$};
\node[graph] (gb1) at (0,-1) {$b$};
\node[graph] (gc1) at (1,-1) {$c$};
\node[graph] (ga2) at (1,0) {$a$};
\draw (ga1) -- (gb1) -- (gc1) -- (ga2);
\node (gb2) at (-1,0) {};
\node (ga3) at (2,-1) {};
\draw[dashed] (ga1) -- (gb2);
\draw[dashed] (gc1) -- (ga3);
\node (G) at (2,.3) {$G$};
\end{scope}

\begin{scope}[shift={(-6,-5)},scale=1.5]
\begin{pgfonlayer}{foreground}
\node[graph] (Ha1) at (0,1.2) {$1$};
\node[graph] (Ha2) at (.4,1.2) {$3$};
\node[graph] (Ha3) at ($(Ha1)!.5!(Ha2)+(90:.3)$) {$2$};
\node[graph] (Hb1) at (0,0) {$5$};
\node[graph] (Hc1) at (1.2,0) {$7$};
\node[graph] (Hc2) at (1.5,.3) {$10$};
\node[graph] (Hc3) at (1.5,-.3) {$8$};
\node[graph] (Hc4) at (1.8,0) {$9$};
\node[graph] (Hab) at ($(Ha1)!.5!(Hb1)$) {$4$};
\node[graph] (Hbc) at ($(Hb1)!.5!(Hc1)$) {$6$};
\node[graph] (Hac) at ($(Ha2)!.5!(Hc2)$) {$11$};
\draw (Ha1) -- (Ha2) -- (Ha3) -- (Ha1);
\draw (Ha1) -- (Hab) -- (Hb1) -- (Hbc) -- (Hc1) -- (Hc2) -- (Hc3) -- (Hc1) -- (Hc4) -- (Hc2) -- (Hac) -- (Ha2);
\end{pgfonlayer}
\node (H) at ($(Hac)+(1,.8)$) {$H$};
\end{scope}

\pgfmathanglebetweenpoints{\pgfpointanchor{Ha2}{center}}{\pgfpointanchor{Hc2}{center}}
\edef\angleAC{\pgfmathresult}
\path[stainEdge,rounded corners=9] ($(Ha1)!.3!(Hb1)+(-.2,.25)$) rectangle ($(Ha1)!.7!(Hb1)+(.2,-.25)$) {};
\path[stainEdge,rounded corners=9] ($(Hb1)!.33!(Hc1)+(-.3,.2)$) rectangle ($(Hb1)!.67!(Hc1)+(.3,-.2)$) {};
\path[stainEdge,rounded corners=9,rotate=\angleAC] ($(Ha2)!.27!(Hc2)+(-.3,.2)$) rectangle ($(Ha2)!.73!(Hc2)+(.3,-.2)$) {};
\draw[stain] plot[smooth cycle,tension=.9] coordinates {($(Ha1)+(-.25,-.15)$) ($(Ha2)+(.25,-.15)$) ($(Ha3)+(0,.25)$)};
\draw[stain] plot[smooth cycle,tension=.5] coordinates {($(Hc1)+(-.3,0)$) ($(Hc2)+(0,.3)$) ($(Hc4)+(.3,0)$) ($(Hc3)+(0,-.3)$)};
\draw[stain] (Hb1) circle (.3);
\node[above left=13pt] (Sa) at (Ha1) {$S_a$};
\node[below left=11pt] (Sb) at (Hb1) {$S_b$};
\node[above right=12pt] (Sc) at (Hc2) {$S_c$};
\node[left=9pt] (Pab) at (Hab) {$P_{ab}$};
\node[below=9pt] (Pbc) at (Hbc) {$P_{bc}$};
\node[above right=9pt] (Pac) at (Hac) {$P_{ac}$};

\begin{scope}[shift={(0,-5)},scale=1.7]
\node[graph] (GHa1) at (0,1.2) {$1$};
\node[graph] (GHa2) at (.4,1.2) {$3$};
\node[graph] (GHa3) at ($(GHa1)!.5!(GHa2)+(90:.3)$) {$2$};
\begin{scope}[shift={(1.1,0)}]
\node[graph] (GHa12) at (0,1.2) {$1$};
\node[graph] (GHa22) at (.4,1.2) {$3$};
\node[graph] (GHa32) at ($(GHa12)!.5!(GHa22)+(90:.3)$) {$2$};
\draw (GHa12) -- (GHa22) -- (GHa32) -- (GHa12);
\end{scope}
\node[graph] (GHb1) at (0,0) {$5$};
\node[graph] (GHc1) at (1.2,0) {$7$};
\node[graph] (GHc2) at (1.5,.3) {$10$};
\node[graph] (GHc3) at (1.8,0) {$9$};
\node[graph] (GHc4) at (1.5,-.3) {$8$};
\node[graph] (GHab) at ($(GHa1)!.5!(GHb1)$) {$4$};
\node[graph] (GHbc) at ($(GHb1)!.5!(GHc1)$) {$6$};
\node (GHb2) at ($(GHa1)-(.63,0)$) {};
\draw (GHa1) -- (GHa2) -- (GHa3) -- (GHa1);
\draw (GHa1) -- (GHab) -- (GHb1) -- (GHbc) -- (GHc1) -- (GHc2) -- (GHc3) -- (GHc1) -- (GHc4) -- (GHc2);
\draw[dashed] (GHa1) -- (GHb2);
\begin{scope}[shift={(2.8,-1.2)}]
\node (GHa31) at ($(GHc2)+(1.2,0)$) {};
\end{scope}
\node[graph] (GHc2a22) at ($(GHc2)!.5!(GHa22)$) {$11$};
\draw (GHc2) -- (GHc2a22) -- (GHa22);

\node (GHc2a31) at ($(GHc2)!.5!(GHa31)$) {};
\draw[dashed] (GHc2) -- (GHc2a31);
\node (G1) at ($(GHa22)+(.6,.2)$) {$G'$};
\end{scope}
\end{tikzpicture}

%% file: subcount.tex
\section{Classification for $\#\subsprob_{\text{D}}(C)$}\label{sec:subgraphs}
In this section we give our complete and almost-tight classification for $\#\subsprobd(C)$.
We recall that $\#\subsprobd(C)$ is the problem that expects as input a graph $H\in C$ and an arbitrary graph $G$, and whose goal is to compute the number of subgraphs of $G$ that are isomorphic to $H$, with the parameterization given by $|H|+d(G)$.
We show:
\begin{theorem}\label{thm:main_subgraphs}
    If $C$ has bounded induced matching number then $\#\subsprobd(C)$ is fixed-parameter tractable, and can be solved in time
    \begin{linenomath*}
    \begin{equation}
    f(|H|,d(G))\cdot |V(G)|^{\max(1,\imn(H))}\cdot \log |V(G)|
    \end{equation}
    \end{linenomath*}
    for some computable function $f$.
    Otherwise, $\#\subsprobd(C)$ is $\#\W{1}$-hard and cannot be solved in time
    \begin{linenomath*}
    \begin{equation}
    f(|H|,d(G))\cdot |V(G)|^{o\left(\frac{\imn(H)}{\log (\imn(H))}\right)}
    \end{equation}
    \end{linenomath*}
    for any function $f$, unless ETH fails.
\end{theorem}
The upper bounds and the lower bounds are proven separately in the next two subsections.

\subsection{Upper Bounds}
For the first part of Theorem~\ref{thm:main_indsubgraphs}, we rely on a result by one of the authors~\cite[Theorem 9]{Bressan21}.
The result states that $\#\subsprobd(C)$ can be solved in time
\begin{linenomath*}
\begin{equation}
    f(|H|,d(G))\cdot |V(G)|^{\tau_2(H)} \cdot \log |V(G)|
\end{equation}
\end{linenomath*}
for some computable function $f$, where $\tau_2(H)$ is the maximum dag treewidth of any quotient $\widehat{H}$ of $H$ (with self-loops deleted), see Section~\ref{sec:prelims}. Therefore, we only need to show that $\tau_2(H) \le \max(1,\imn(H))$.
Let $\orient{H}$ be a generic acyclic orientation of $H$, and let $S=S(\orient{H})$ be the set of sources of $\orient{H}$.
For any subset $U \subseteq V(H)$, we write $\orient{H}(U)$ for the set all vertices of $H$ that, in $\orient H$, are reachable from some vertex of $U$.
For any $t \ge 1$, a $t$\emph{-kernel} of $\orient{H}$ is a subset $K \subseteq S$ with $|K|=t$ such that $V(H) \setminus S \subseteq \orient{H}(K)$; that is, any non-source node of $\orient{H}$ is reachable from some node of $K$.
Note that by definition $t$-kernels are nonempty, since $t \ge 1$. We prove two intermediate results and, then, we prove that $\tau_2(H) \le \max(1,\imn(H))$.

\begin{lemma}\label{lem:kernel_and_tau}
	If $\orient{H}$ has a $t$-kernel, then $\dtw(\orient{H}) \le t$.
\end{lemma} 
\begin{proof}
    Let $K$ be a $t$-kernel of $\orient H$. We construct a d.t.d.\ $\scT=(\bags,\et)$ of $\orient H$ as follows. The root of $\scT$ is $K$, and for each source $s \in S\setminus K$, the bag $B_s = \{s\}$ is a child of $K$ in $\scT$. One can immediately check that $\scT$ satisfies all three properties of Definition~\ref{def:gdtd} and therefore is a valid d.t.d.\ for $\orient H$. Finally, note that $|B| \le t$ for all $B \in \bags$, so $\dtw(\scT) \le t$, which proves that $\dtw(\orient H) \le t$.
\end{proof}

\begin{lemma}\label{lem:match_and_kernel}
	$\orient{H}$ has a $t$-kernel with $t \le \max(1,\imn(H))$.
\end{lemma}
\begin{proof}
    For any subset $U\subseteq S$ we let $\orient{H}^{+}(U)= \orient{H}(U) \setminus U$. In words, $\orient{H}^+(U)$ is the set of non-sources reachable from $U$. Now consider the following iterative procedure. We start by setting $U=S$. Note that this implies that $\orient{H}^+(U) = V(H) \setminus S$. Then, while there is some $s \in U$ such that $\orient H^+(U\setminus\{s\})= H^+(U)$, we set $U \leftarrow U \setminus\{s\}$.
    Note that, since we start with $\orient{H}^+(U) = V(H) \setminus S$ and at each step we ensure that $\orient{H}^+(U\setminus\{s\}) = \orient{H}^+(U)$, then $U$ is always a $|U|$-kernel, unless $U \ne \emptyset$.
	
	Now consider the final value of $U$. Observe that, if $U=\emptyset$, then $H$ is necessarily the independent set. In this case, $\imn(H)=0$. We choose any $s \in S$ and set $K=\{s\}$, and $K$ will be a $t$-kernel of $H$, with $t \le 1 = \max(1,\imn(H))$.
	
	Suppose instead that $U=\{s_1,\dots,s_k\}\ne \emptyset$. In this case, we set $K=U$. Note that, by construction, $K$ is a $k$-kernel of $\orient H$. Now, we claim that $k = |K| \le \imn(H)$.
	First, observe what follows. For each $s_i \in U$, there must exist an arc $(s_i,u_i) \in \orient{H}$ such that for all $s_j \ne s_i$ we have $u_i \notin \orient{H}(s_j)$. Indeed, if this was not the case, then $\orient{H}^+(\{s_i\}) \subseteq \orient{H}^+(U \setminus \{s_i\})$.
	This implies that $\orient{H}^+(U \setminus \{s_i\}) = \orient{H}^+(U)$, so the procedure should have removed $s_i$ from $U$, a contradiction.
	Now, we claim that $M=\{\{s_1,u_1\},\dots,\{s_k,u_k\}\}$ is an induced matching of $H$. To see that $M$ is a matching, simply note that by the above discussion the pairs $\{s_1,u_1\}$ are pairwise disjoint edges of $H$. To see that $M$ is induced, observe what follows:
	\begin{enumerate}
		\item $E(H)$ cannot contain $\{s_i,s_j\}$, as both $s_i$ and $s_j$ are sources.
		\item $E(H)$ cannot contain $\{s_i,u_j\}$ for $i \neq j$, as otherwise $u_j$ would be reachable from $s_i$, contradicting the fact that $u_j$ is reachable only from $s_j$.
		\item $E(H)$ cannot contain $\{u_i,u_j\}$, since otherwise $u_i$ would be reachable from $s_j$, or vice versa, which is again a contradiction.
	\end{enumerate} 
	Thus $M$ is an induced matching of $H$, and therefore $|M| \le \imn(H)$. Moreover, $\imn(H) \ge 1$ as $|M| \ge 1$. Hence again $|M|\le \max(1,\imn(H))$, as claimed.
\end{proof}

\begin{lemma}\label{lem:tau2_leq_imn}
    Any graph $H$ satisfies $\tau_2(H) \leq \max(1,\imn(H))$.
\end{lemma}
\begin{proof}
    Let $H' = H/\rho$ be any quotient of $H$, and $\orient H'$ be any acyclic orientation of it.
    By Lemma~\ref{lem:kernel_and_tau} and Lemma~\ref{lem:match_and_kernel}, we have $\tau(\orient H') \le \imn(H')+1$.
	Now, observe that $\imn(H') \leq \imn(H)$, since taking the quotient of a graph cannot increase the induced matching number.
	By definition of $\tau_2(H)$, this implies that $\tau_2(H) \le \imn(H)$.
\end{proof}

\subsection{Lower Bounds}
We first show that graphs with induced matching number $\ell$ admit quotients with treewidth $\Omega(\ell)$. For the remainder of this section, we let $\mathcal{F}=\{F_k\}_{k\geq 1}$ be the class of expander graphs defined in Theorem~\ref{thm:nice_expanders}, and $c,k_0$ be the constants specified by the same theorem.

\begin{lemma}\label{lem:F-quotient}
    For some positive constants $c$ and $k_0$, we have what follows.
    Let $H$ be a graph with $\imn(H)=\ell$ and set $k:= \lfloor\frac{\ell}{2c} \rfloor$. If $k\geq k_0$, then there exists a partition $\rho$ of the vertices of $H$ such that $H/\rho$ is without self-loops, and has an $F_k$-gadget, that is, an F-gadget of treewidth at least $k$.
\end{lemma}
\begin{proof}
    The intuition is that, if $H$ has a large induced matching, then we can identify the endpoints of such a matching according to some appropriate quotient $\rho$ in such a way to obtain the $1$-subdivision of $F_k$. This implies that $H$ has an $F_k$-gadget, and by Theorem~\ref{thm:nice_expanders} we have $\tw(F_k)\ge k$.
    
    Formally, since $k\geq k_0$, by Theorem~\ref{thm:nice_expanders} we have $|E(F_k)| \le ck$. Let $\hat{F}$ be the $1$-subdivision of $F_k$, and let $m=|E(\hat{F})|$. Clearly $m \le 2ck$, which implies $m \leq \ell$ since $k \le \frac{\ell}{2c}$. Now look at $H$. Since $m \le \ell = \imn(H)$, then $H$ has an induced $m$-matching $H[M]$. And since $H[M]$ is an $m$-matching and $m = |E(\hat{F})|$, there exists a partition $\sigma$ such that $H[M]/\sigma$ is isomorphic to $\hat{F}$. Now, consider the partition $\rho:=\sigma\cup \{\{v\}~|~v\in V(H)\setminus M \}$ of $V(H)$. As we have just shown, $H/\rho$ contains $\hat{F}$ as an induced subgraph, and moreover $H/\rho$ is without self-loops, for otherwise $H[M]$ would not be induced. But $\hat{F}$ is the $1$-subdivision of $F_k$. This implies that $H/\rho$ has an $F_k$-gadget as remarked in Section~\ref{sec:gadgets}.
\end{proof}
Next, we show that if $C$ has F-gadgets of unbounded induced matching number then $\#\subsprobd(C)$ is hard.
\begin{lemma}\label{lem:subs_hard}
    If $C$ has unbounded induced matching number, then $\#\subsprobd(C)$ is $\#\W{1}$-hard and cannot be solved in time
    \begin{linenomath*}
    \begin{equation}
    f(|H|,d(G))\cdot |V(G)|^{o\left(\frac{\imn(H)}{\log (\imn(H))}\right)} \,,
    \end{equation}
    \end{linenomath*}
    for any function $f$, unless ETH fails.
\end{lemma}
\begin{proof}
    Let
    \begin{linenomath*}
    \begin{equation}
    \scF' = \{F_k \in \scF ~|~ k\geq k_0 ~\wedge~ \exists \ell \in \mathbb{N}, H \in C ~:~ k=\lfloor\ell/2c\rfloor \wedge\imn(H)=\ell \} \,.
    \end{equation}
    \end{linenomath*}
    Note that $\scF'$ has unbounded treewidth, since for all sufficiently large $\ell$ we have $F_k \in \scF'$ with $k=\Omega(\ell)$. Thus, by Theorem~\ref{thm:hardness_bottleneck}, unless ETH fails no algorithm exists that for some function $\hat{f}$ solves
    $\#\homsprob(\scF')$ in time $\hat{f}(|F|)\cdot |V(G)|^{o(\tw(F) / \log \tw(F))}$.
    We will construct a tight parameterized Turing-reduction from $\#\cphomsprob(\scF')$ to $\#\subsprobd(C)$, which will imply the thesis.
    
    Let $(F,G)$ be an instance of $\#\cphomsprob(\scF')$, that is, $F\in \scF'$ and $G$ is an $F$-coloured graph and the goal is to compute $|\cphoms{F}{G}|$. By definition of $\scF'$, we have $F=F_k$ for some $k\geq k_0$. Hence, our reduction first searches a graph $H\in C$ such that $\imn(H)=\ell$, $k=\lfloor\ell/2c\rfloor$, and $F=F_k$. Since $C$ is computable, finding $H$ takes time at most $f_1(|F|)$ for some computable function $f_1$.
    
    Now, by Lemma~\ref{lem:F-quotient} $V(H)$ admits a quotient graph $\widehat{H}=H/\rho$ that is simple and has an $F_k$-gadget. Note that we can find $\widehat{H}=H/\rho$ in time only depending on $H$. By Lemma~\ref{lem:FGadget_main}, for some computable function $f_2$ in time $f_2(|\widehat{H}|) \cdot |G|^{O(1)}$ we can compute an $\widehat{H}$-coloured graph $G'$ such that
    \begin{linenomath*}
    \begin{equation}
    |\cphoms{F_k}{G}|=|\cphoms{\widehat{H}}{G'}|
    \end{equation}
    \end{linenomath*}
    Moreover, still by Lemma~\ref{lem:FGadget_main}, we have $d(G') \le |V(\widehat{H})|+2$ and $|G'| \le (|\widehat{H}|\cdot|G|)^{O(1)}$.
    
    Now, we show how to compute $|\cphoms{\widehat{H}}{G'}|$ by using any algorithm for $\#\subsprobd(C)$ as an oracle. First, observe that by Facts~\ref{fact:cphoms_to_cfhoms} and~\ref{fact:cfhoms_to_homs}, we have that
    \begin{linenomath*}
    \begin{equation}\label{eq:cphoms_to_homs}
        |\cphoms{\widehat{H}}{G'}| = |\auts{\widehat{H}}^{-1}| \cdot \sum_{S\subseteq V(\widehat{H})}(-1)^{|S|}\cdot |\homs{\widehat{H}}{G' - S}|\,, 
    \end{equation}
    \end{linenomath*}
    where $G'-S$ is the graph obtained from $G'$ by deleting all vertices that are coloured with a vertex in~$S$. Thus, we focus on computing the term $|\homs{\widehat{H}}{G' - S}|$ for each possible $S\subseteq V(\widehat{H})$; clearly, there are $2^{|V(\widehat{H})|}$ such terms. Note also that $d(G'-S) \le d(G')$, since $G'-S \subseteq G'$. 
    
    Let then $S\subseteq V(\widehat{H})$ be fixed. Since $|\subs{H}{\star}| = |\auts{H}^{-1}| \cdot |\embs{H}{\star}|$, and by Fact~\ref{fact:sub_hom_basis}, we have:
    \begin{linenomath*}
    \begin{equation}
    |\subs{H}{\star}| = |\auts{H}^{-1}| \cdot \sum_\sigma \mu(\bot,\sigma) \cdot |\homs{H/\sigma}{\star}| \,.
    \end{equation}
    \end{linenomath*}
    Collecting for isomorphic graphs we obtain a function $a$ of finite support such that
    \begin{linenomath*}
    \begin{equation}
    |\subs{H}{\star}| =  \sum_{H'} a(H') \cdot |\homs{H'}{\star}| \,,
    \end{equation}
    \end{linenomath*}
    where the sum is taken over all (isomorphism classes of) graphs $H'$. It is by now well-known that, for each partition $\sigma$ such that $H/\sigma$ is without self-loops, we have $a(H')\neq 0$, where $H'$ is the representant of the isomorphism class of $H/\sigma$~\cite{CurticapeanDM17}. Thus, without loss of generality, we can assume that $a(\widehat{H}) \ne 0$.
    
    This allows us to invoke complexity monotonicity for bounded degeneracy graphs as described by Gishboliner et al.~\cite{Gishboliner20}. That is, we use the oracle for $\#\subsprobd(C)$ to obtain 
    \begin{linenomath*}
    \begin{equation}
    |\subs{H}{(G'-S) \otimes H_i}| =  \sum_{H'} a(H') \cdot |\homs{H'}{(G'-S)\otimes H_i}|
    \end{equation}
    \end{linenomath*}
    for a sequence of graphs $H_i$ only depending on $H$, where $\otimes$ is the tensor product of graphs\footnote{The adjacency matrix of $G_1\otimes G_2$ is the Kronecker product of the adjacency matrices of $G_1$ and $G_2$.}. Using that 
    \begin{linenomath*}
    \begin{equation}
    |\homs{H'}{(G'-S) \otimes H_i}| = |\homs{H'}{(G'-S)}| \cdot |\homs{H'}{H_i}| \,,
    \end{equation}
    \end{linenomath*}
    we obtain a system of linear equations which has a unique solution for an appropriate choice of the $H_i$~\cite[Lemma 3.6]{CurticapeanDM17}. Hence, solving the systems of linear equations for each $S$ reveals $|\homs{\widehat{H}}{G-S}|$, which in turn reveals $|\cphoms{\widehat{H}}{G}|$ by Equation~(\ref{eq:cphoms_to_homs}). This concludes the reduction.
    
    Now let us bound the running time. The crucial observation of Gishboliner et al.\ \cite[Observation A.1]{Gishboliner20} is now that the degeneracy of $(G'-S)\otimes H_i$ satisfies
    \begin{linenomath*}
    \begin{equation}
        d((G'-S)\otimes H_i) \le d(G'-S)\cdot |V(H_i)|\leq (|V(\widehat{H})|+2)\cdot |V(H_i)| \,.
    \end{equation}
    \end{linenomath*}
    Since $H$, and thus $\widehat{H}$ and the $H_i$, are a function of $F$, we observe the parameter of each oracle query $(H, (G'-S)\otimes H_i)$ with respect to the parameterization of $\#\subsprobd(C)$, that is, $|H|+d((G'-S)\otimes H_i)$, is bounded by a function of $|F|$.
    Now, suppose that there exists an algorithm that solves $\#\subsprobd(C)$ in time:
    \begin{linenomath*}
\begin{equation}
        f(|H|,d(G'))\cdot |V(G')|^{o\left(\frac{\imn(H)}{\log (\imn(H))}\right)} \,.
\end{equation}
\end{linenomath*}
    Now recall that, by construction, we have $\imn(H)=\ell \in \scO(k)$, and that $|G'|\leq (|\widehat{H}|\cdot |G|)^{O(1)}$.
    Hence, every call to the oracle for $\#\subsprobd(C)$ takes time at most:
    \begin{linenomath*}
\begin{equation}
        \hat{f}(|F|)\cdot |V(G)|^{o\left(\frac{\tw(F)}{\log (\tw(F))}\right)} \,.
\end{equation}
\end{linenomath*}
    for some function $\hat{f}$.
    As the number of oracle calls is bounded by a function of $F$ as well, by Theorem~\ref{thm:hardness_bottleneck} we would refute ETH, concluding the proof.
\end{proof}

We emphasise a variety of lower bounds that follow from the previous lemma; note that Corollary~\ref{cor:intro_subs_cor} is subsumed by the following result.
\begin{corollary}
	The following problems are $\#\W{1}$-hard when parameterized by $k$ and $d(G)$, and cannot be solved in time $f(k,d(G))\cdot |G|^{o(k/\log k)}$ for any function $f$, unless ETH fails. Given $G$ and $k$, 
	\begin{enumerate}
		\item compute the number of $k$-cycles in $G$.
		\item compute the number of $k$-paths in $G$.
		\item compute the number of $k$-matchings in $G$.
		\item compute the number of $k$-trees (i.e., trees with $k$ vertices) in $G$.
	\end{enumerate}
\end{corollary}
\begin{proof}
    The claim for (1.-3.) follows immediately by the previous theorem since $k$-cycles, $k$-paths, and $k$-matchings all have induced matching number $\Omega(k)$.
    
    For (4.), note that the number of $k$-trees in a graph $G'$ equals $\sum_{T} |\subs{T}{G'}|$, where the sum is over all trees with $k$ vertices. Curticapean, Dell and Marx~\cite[Corollary~5.3]{CurticapeanDM17} observed that the linear combination of homomorphism counts associated with the previous sum also satisfies that no quotient (without self-loops) of \emph{any} $k$-tree vanishes. Thus the proof of the previous theorem applies to counting all $k$-trees as well, and the lower bounds are obtained by observing that the set of all $k$-trees has induced matching number $\Omega(k)$, by taking for example $T=P_k$.
\end{proof}

Let us finally prove the full classification for counting subgraphs in degenerate graphs; Theorems~\ref{thm:intro_subs_finegrained} and~\ref{thm:intro_subs_param} are immediate consequences:
\begin{proof}[Proof of Theorem~\ref{thm:main_subgraphs}]
    The upper bound holds by Lemma~\ref{lem:tau2_leq_imn} and~\cite[Theorem~9]{Bressan21}, and the lower bound holds by Lemma~\ref{lem:subs_hard}.
\end{proof}

%% file: indcount.tex
\section{Classification for $\#\indsubsprob_{\text{D}}(C)$}\label{sec:indsubgraphs}
In this section we give our complete and almost-tight classification for $\#\indsubsprobd(C)$.
We recall that $\#\indsubsprobd(C)$ is the problem that expects as input a graph $H\in C$ and an arbitrary graph~$G$, and whose goal is to compute the number of induced subgraphs $|\indsubs{H}{G}|$ of $G$ that are isomorphic to $H$, with the parameterization is given by $|H|+d(G)$. For what follows, recall that $\alpha(H)$ is the independence number of $H$.
We prove:
\begin{theorem}\label{thm:main_indsubgraphs}
    If $C$ has bounded independence number, then $\#\indsubsprobd(C)$ is fixed-parameter tractable and can be solved in time
    \begin{linenomath*}
    \begin{equation}
    f(|H|,d(G)) \cdot |V(G)|^{\alpha(H)}\cdot \log |V(G)| \,,
    \end{equation}
    \end{linenomath*}
    for some computable function $f$.
    Otherwise, $\#\indsubsprobd(C)$ is $\#\W{1}$-hard and cannot be solved in time
    \begin{linenomath*}
    \[f(|H|,d(G))\cdot |V(G)|^{o\left(\frac{\alpha(H)}{\log (\alpha(H))}\right)} \,,\]
    \end{linenomath*}
    for any function $f$, unless ETH fails.
\end{theorem}
The upper bounds of Theorem~\ref{thm:main_indsubgraphs} follow directly from Theorem 9 and Lemma 14 of~\cite{Bressan21}.
Thus, we focus on the lower bounds. Similarly to the hardness proof for $\#\subsprobd(C)$, see Section~\ref{sec:subgraphs}, we will use a family of expanders to identify graphs with $F$-gadgets of high treewidth, which occur with a non-zero coefficient in the associated expression as linear combination of homomorphism counts.
We let $\mathcal{F}=\{F_k\}_{k\geq 1}$ be the family of expander graphs given by Theorem~\ref{thm:nice_expanders}. Recall that an edge-supergraph of a graph $H$ is a graph obtained from $H$ by adding edges.

\begin{lemma}\label{lem:F-supergraph}
    There exist positive constants $c$ and $k_0$ such that the following is true:
    Let $H$ be a graph with $\alpha(H)=\ell$ and set $k:= \lfloor\frac{\ell}{2c} \rfloor$. If $k\geq k_0$, then there exists an edge-supergraph $\widehat{H}$ of $H$ such that $\widehat{H}$ has an $F_k$-gadget, that is, an F-gadget of treewidth at least $k$.
\end{lemma}
\begin{proof}
     If $k\geq k_0$, then by Theorem~\ref{thm:nice_expanders} we have $|E(F_k)| \le ck$ and $|V(F_k)|\le ck$. Let $\hat{F}$ be the $1$-subdivision of $F_k$, and let $m=|V(\hat{F})|$. Clearly, $|V(\hat{F})| \le |V(F_k)|+|E(F_k)|\leq 2ck\leq \ell$.
     
     Now let us look at $H$. Since $\alpha(H)=\ell \ge m$, then $H$ has an independent set $H[I]$ of size~$m$, that is, of size $|V(\hat{F})|$. Therefore, $H[I]$ has an edge-supergraph isomorphic to $\hat{F}$. Let us write $A$    for a set of edges (between vertices in $I$) that, when added to $H[I]$, yield $\hat{F}$. Clearly the edge-supergraph $\widehat{H}:=(V(H),E(H)\cup A)$ contains $\hat{F}$ as induced subgraph. 
     Since $\hat{F}$ is the $1$-subdivision of $F_k$, we conclude that $\widehat{H}$ has an $F_k$-gadget.
\end{proof}

Now we can show that, if $C$ has unbounded independence number, then $\#\indsubsprobd(C)$ is hard.
\begin{lemma}\label{lem:indsubs_hard}
    If $C$ has unbounded independence number, then $\#\indsubsprobd(C)$ is $\#\W{1}$-hard and cannot be solved in time
    \begin{linenomath*}
    \begin{equation}
    f(|H|,d(G))\cdot |V(G)|^{o\left(\frac{\alpha(H)}{\log (\alpha(H))}\right)} \,,
    \end{equation}
    \end{linenomath*}
    for any function $f$, unless ETH fails.
\end{lemma}
\begin{proof}
    The proof is almost identical to the proof of Lemma~\ref{lem:subs_hard}, with only a few differences.
    First, we let:
    \begin{linenomath*}
    \begin{equation}
    \scF'=\{F_k~|~k\geq k_0 ~\wedge~ \exists \ell \in \mathbb{N}, H\in C ~:~ k=\lfloor\ell/2c\rfloor \wedge\alpha(H)=\ell \} \,.
    \end{equation}
    \end{linenomath*}
    Note that this is the same set $\scF'$ defined in the proof of Lemma~\ref{lem:subs_hard}, but with $\alpha(H)$ in place of $\imn(H)$.
    
    We construct a tight parameterized Turing-reduction from $\#\cphomsprob(\scF')$ to $\#\indsubsprobd(C)$.
    The reduction is identical to the one in the proof of Lemma~\ref{lem:subs_hard},
    but now we let $\widehat{H}$ be an edge-supergraph of $H$ that has an $F_k$-gadget. Note that $\widehat{H}$ exists, by Lemma~\ref{lem:F-supergraph}, and again, it can be computed in time depending only on $H$.
    Now, we need to express the number of \emph{induced} subgraphs as a linear combination of homomorphism counts (and to show that the coefficient of~$\widehat{H}$ does not vanish).
    To this end, since $|\indsubs{H}{\star}| = |\auts{H}^{-1}| \cdot |\strembs{H}{\star}|$, we can combine Fact~\ref{fact:sub_hom_basis} and Fact~\ref{fact:indsub_sub_basis} to obtain:
    \begin{linenomath*}
    \begin{equation}
    |\indsubs{H}{\star}| = |\auts{H}^{-1}| \cdot \sum_{H'}(-1)^{|E(H')|-|E(H)|} \cdot |\{H' \supseteq H\}| \cdot \sum_\sigma \mu(\bot,\sigma) \cdot |\homs{H'/\sigma}{\star}| \,
    \end{equation}
    \end{linenomath*}
    where the first sum is over all isomorphism classes of graphs, and the second sum is over all partitions of $V(H')$.
    Collecting for isomorphic graphs we obtain a function $a$ of finite support such that
    \begin{linenomath*}
    \begin{equation}
    |\indsubs{H}{\star}| =  \sum_{H'} a(H') \cdot |\homs{H'}{\star}| \,
    \end{equation}
    \end{linenomath*}
    where the sum is again taken over all isomorphism classes of graphs $H'$. Curticapean, Dell and Marx~\cite{CurticapeanDM17} have shown that $a(H')\neq 0$ for each edge-supergraph $H'$ of $H$ --- see also~\cite{Gishboliner20}. In particular, the coefficient $a(\widehat{H})$ of  $|\homs{\widehat{H}}{\star}|$ is non-zero. 
    This yields the desired reduction from $\#\cphomsprobd(\scF')$.
    The lower bound follows along the lines of Lemma~\ref{lem:subs_hard}, only with $\alpha(H)$ in place of $\imn(H)$.
\end{proof}

\begin{corollary}[Corollary~\ref{cor:intro_indsubs_cor} restated]
	The following problems are $\#\W{1}$-hard when parameterized by $k$ and $d(G)$, and cannot be solved in time $f(k,d(G))\cdot |G|^{o(k/\log k)}$ for any function $f$, unless ETH fails. Given $G$ and $k$, 
	\begin{enumerate}
		\item compute the number of induced copies of the $(k,k)$-biclique in $G$.
		\item compute the number of $k$-independent sets in $G$.
		\item compute the number of induced $k$-cycles in $G$.
		\item compute the number of induced $k$-paths in $G$.
		\item compute the number of induced $k$-matchings in $G$.
	\end{enumerate}
\end{corollary}
\begin{proof}
Use Lemma~\ref{lem:indsubs_hard}, noting that every mentioned graph class has independence number lower bounded by $\Omega(k)$.
\end{proof}

Let us finally prove the full classification for counting subgraphs in degenerate graphs; Theorems~\ref{thm:intro_indsubs_finegrained} and~\ref{thm:intro_indsubs_param} are immediate consequences:
\begin{proof}[Proof of Theorem~\ref{thm:main_indsubgraphs}]
    The upper bound follows diectly from Theorem 9 and Lemma 14 of~\cite{Bressan21}, and the lower bound holds by Lemma~\ref{lem:indsubs_hard}.
\end{proof}

\subsection{Generalised Induced Subgraph Patterns and k-Graphlets}
Having classified the complexity of counting the number of induced copies of a particular graph~$H$, we will now follow the approach of Jerrum and Meeks~\cite{JerrumM15} and consider the problem of counting all induced subgraphs that satisfy a fixed graph property. More precisely, given a computable graph property $\Phi$, we consider the problem $\#\indsubsprobd(\Phi)$ which asks, on input a graph $G$ and a positive integer $k$, to compute the number of induced subgraphs of size $k$ in $G$ that satisfy $\Phi$. The parameterization is given by~$k$ and~$d(G)$.

\noindent Observe that $\#\indsubsprobd(\Phi)$ subsumes the problems of counting induced copies of $k$-cliques, $k$-cycles or $k$-paths in bounded degeneracy graphs by setting $\Phi(H):=1$ if and only if $H$ is a clique, a cycle, or a path, respectively. However, since $\Phi$ can be any (computable) graph property, we can also express problems of counting more general induced subgraph patterns, such as the number of $k$-graphlets; here, a $k$-graphlet is a connected induced subgraph of size $k$.

The goal is to determine the complexity of $\#\indsubsprobd(\Phi)$ for any property $\Phi$. However, even without the condition of $G$ having bounded degeneracy (i.e., even if $d(G)$ is not a parameter), the complexity of this problem is not fully understood (cf.\ \cite{RothSW20}). For this reason, a complete classification in the spirit of Theorem~\ref{thm:main_indsubgraphs} along $\Phi$ is currently elusive. Therefore, we will focus on properties that are hereditary or minor-closed, as well as on selected natural properties such as connectivity; we emphasise that the latter was the first one for which the case of unbounded degeneracy was understood~\cite{JerrumM15}.

Our hope is that including the degeneracy of $G$ in the parameter (i.e. assuming that the degeneracy of $G$ is small), will allow us to identify new tractability results. However, the results of the current section indicate that bounded degeneracy does not help for generalised induced subgraph counting; note that Theorem~\ref{thm:intro_graphlets} is subsumed by what follows:
\begin{theorem}\label{thm:indsubsphi_hard_cor}
    The problem $\#\indsubsprobd(\Phi)$ is $\#\W{1}$-hard if any of the following is true:
    \begin{enumerate}
        \item $\Phi$ is the property of being connected.
        \item $\Phi$ is the property of being claw-free.
        \item $\Phi$ is minor-closed, non-trivial and of unbounded independence number (e.g., planarity).
    \end{enumerate}
\end{theorem}
We emphasise that, in contrast to Theorem~\ref{thm:main_indsubgraphs}, we do not obtain (almost) tight conditional lower bounds, but only $\#\W{1}$-hardness.\footnote{Our $\#\W{1}$-hardness proof yields an implicit conditional lower bound stating that a running time of $f(k,d(G)) \cdot |G|^{o(\sqrt{k})}$ is not possible for any function $f$, unless ETH fails.}

The proof of Theorem~\ref{thm:indsubsphi_hard_cor} relies on an algebraic approach of analysing the coefficients in the associated linear combinations of homomorphism counts, as introduced in~\cite{DorflerRSW19}. However, here we need a result that is stronger (and, it turns out, also more technical).

\begin{lemma}\label{lem:phi_IS_biclique}
    Let $\Phi$ be a computable graph property. Suppose that for all but finitely many positive integers~$\ell$ we have
    \begin{linenomath*}
    \begin{equation}\label{eq:is_biclique_dif}
        \Phi(\mathsf{IS}_{2\ell+\ell^2}) \neq \Phi(K^2_{\ell,\ell})\,,
    \end{equation}
    \end{linenomath*}
    where $K^2_{\ell,\ell}$ is the subdivision of the $(\ell,\ell)$-biclique, and $\mathsf{IS}_{2\ell+\ell^2}$ is the independent set of size $2\ell+\ell^2$. Then, $\#\indsubsprobd(\Phi)$ is $\#\W{1}$-hard.
\end{lemma}
\begin{proof}
    Let $\widehat{C}$ be the set of all bicliques $K_{\ell,\ell}$ such that $\ell$ is a power of $2$ and satisfies Equation~(\ref{eq:is_biclique_dif}). The hypotheses of the lemma imply that $\widehat{C}$ is infinite and has unbounded treewidth. Thus, $\#\cphomsprob(\widehat{C})$ is $\#\W{1}$-hard by Theorem~\ref{thm:hardness_bottleneck}. We will use $\#\cphomsprob(\widehat{C})$ to infer the hardness of $\#\indsubsprobd(\Phi)$. 
    
    Now let $C$ be the set of all subdivisions $K^2_{\ell,\ell}$ such that $\ell$ satisfies Equation~(\ref{eq:is_biclique_dif}). 
    Moreover, let  $\#\cpindsubsprobd(\Phi)$ be the following problem: we are given an $H$-coloured graph $G$ for some~$H$, and the goal is to compute $|\cpindsubs{\Phi}{H}{G}|$, where:
    \begin{linenomath*}
    \begin{equation}
    \cpindsubs{\Phi}{H}{G}:= \{ S \subseteq V(G)~:~\Phi(G[S]) \,\wedge\, |S|=|V(H)|\,\wedge\, c_H(S)=V(H)\} \,,
    \end{equation}
    \end{linenomath*}
    that is, the number of induced subgraphs of order $|V(H)|$ in $G$ that satisfy $\Phi$ and are colourful with respect to $c_H$. The parameterization is given by $|V(H)|+d(G)$.
    
    \noindent We will prove the following chain of reductions:
    \begin{linenomath*}
    \begin{equation}
        \#\cphomsprob(\widehat{C})
        \fptred \#\cphomsprobd(C) \fptred \#\cpindsubsprobd(\Phi) \fptred \#\indsubsprobd(\Phi)
    \end{equation}
    \end{linenomath*}
    which will prove the claim since, as said, $\#\cphomsprob(\widehat{C})$ is $\#W[1]$-hard.

    \begin{claim}
    $\#\cphomsprob(\widehat{C}) \fptred \#\cphomsprobd(C)$
    \end{claim}
    \begin{claimproof}
    This reduction holds by Lemma~\ref{lem:FGadget_main}, since $K^2_{\ell,\ell}$ by construction has a $K_{\ell,\ell}$-gadget.
    \end{claimproof}

    \begin{claim}
    $\#\cphomsprobd(C) \fptred \#\cpindsubsprobd(\Phi)$
    \end{claim}
    \begin{claimproof}
    We apply the following result from~\cite{DorflerRSW19}.
    Let $G$ be an $H$-coloured graph. There exists a computable function $a$ of finite support such that
    \begin{linenomath*}
    \begin{equation}\label{eq:cpindsubs_Phi_H_G}
    |\cpindsubs{\Phi}{H}{G}| = \sum_{H'} a(H') \cdot |\cphoms{H'}{G}|  \,,
    \end{equation}
    \end{linenomath*}
    where the sum is over all (isomorphism classes of) edge-subgraphs $H'$ of $H$, that is, graphs obtained from~$H$ by deleting edges, and where $\cphoms{H'}{G}$ has the obvious meaning, since $V(H')=V(H)$.
    
    Now, complexity monotonicity holds as in the uncoloured case: querying the oracle to obtain the counts $|\cpindsubs{\Phi}{H}{G\otimes H_i}|$ for a sequence of $H$-coloured graphs $H_i$ yields a system of linear equations which has a unique solution from which we can obtain $|\cphoms{H'}{G}|$ for each $H'$ with $a(H')\neq 0$ (see~\cite[Lemma~14]{DorflerRSW19}). Since the $H_i$ only depend on $H$, the degeneracy $d(G\otimes H_i)$ is bounded by $f(|H|)\cdot d(G)$ for some computable function $f$. 
    Consequently, it remains to show that $a(H')$ does not vanish in Equation~\ref{eq:cpindsubs_Phi_H_G} for $H'=K^2_{\ell,\ell}$. This is where the algebraic approach applies: It was shown in~\cite{DorflerRSW19} that $a(H)\neq 0$ whenever the following two criteria are satisfied:
    \begin{enumerate}
        \item $\Phi(\mathsf{IS}_{|V(H)|})\neq \Phi(H)$
        \item $H$ is a $p$-edge-transitive graph for some prime $p$.
    \end{enumerate}
    Here, a graph is called $p$-edge-transitive, if it is edge-transitive, that is, the automorphism group acts transitively on the set of edges, and if the number of edges is a power of $p$. 
    
    Now, by the premise of the lemma, we have $\Phi(\mathsf{IS}_{|V(K^2_{\ell,\ell})|}\neq \Phi(H)$. Note further that $|E(K^2_{\ell,\ell})|=2\ell^2$ is a power of $2$ since $\ell$ is. Finally, we use the fact that a graph $H$ is edge-transitive if and only if $H\setminus e$ and $H\setminus e'$ are isomorphic for every pair of edges~$e$ and~$e'$~\cite{AndersenDSV92}. Using this criterion, it is easy to verify that $K^2_{\ell,\ell}$ is indeed edge-transitive.
    
    We can hence conclude that the coefficient of the subdivision of the biclique does not vanish, hence $\#\cphomsprobd(C) \fptred \#\cpindsubsprobd(\Phi)$ as claimed.
    \end{claimproof}

    \begin{claim}
    $\#\cpindsubsprobd(\Phi) \fptred \#\indsubsprobd(\Phi)$
    \end{claim}
    \begin{claimproof}
    It is known that computing $|\cpindsubs{\Phi}{H}{G}|$ reduces to computing the number of induced subgraphs of size $|V(H)|$ that satisfy $\Phi$ in (uncoloured) induced subgraphs of $G$ via the inclusion-exclusion principle (see, for instance, the full version of~\cite{DorflerRSW19}). Since the degeneracy of an induced subgraph of $G$ is at most $d(G)$, the latter reduction remains valid if the degeneracy of the graphs is an additional parameter. This proves the claim and the lemma.
    \end{claimproof}
    This concludes the proof of Lemma~\ref{lem:phi_IS_biclique}.
\end{proof}

\noindent We can finally prove Theorem~\ref{thm:indsubsphi_hard_cor}.
\begin{proof}[Proof of Theorem~\ref{thm:indsubsphi_hard_cor}]
    For cases 1.\ and 2.,  is note that $\Phi$ distinguishes independent sets from subdivisions of bicliques. Thus, Lemma~\ref{lem:phi_IS_biclique} applies. For case 3., the Robertson-Seymour Theorem~\cite{RobertsonS04} implies that any non-trivial minor-closed property is characterised by a \emph{finite} set of forbidden minors. Therefore $\Phi$ is computable. Since $\Phi$ is non-trivial, the set of forbidden minors is non-empty. Since every graph is the minor of the subdivision of some large enough biclique, there exists an $\ell_0$ such that $\Phi$ is false for $K^2_{\ell,\ell}$ for all $\ell\geq \ell_0$. Finally, since $\Phi$ has unbounded independence number, we have that for all $\ell\geq \ell_0$:
    \begin{linenomath*}
    \begin{equation}
    \Phi(\mathsf{IS}_{2\ell+\ell^2}) \neq \Phi(K^2_{\ell,\ell})\,,
    \end{equation}
    \end{linenomath*}
    concluding the proof.
\end{proof}

\noindent\textbf{Remark.} As a consequence of our results above, for minor-closed properties we obtain an exhaustive classification into fixed-parameter and $\#\W{1}$-hard cases.
Formally:
\begin{corollary}[Theorem~\ref{thm:intro_indsubs_minor_closed} restated]
    Let $\Phi$ be a minor-closed graph property. If $\Phi$ is trivial or has bounded independence number, then $\#\indsubsprobd(\Phi)$ is fixed-parameter tractable. Otherwise, $\#\indsubsprobd(\Phi)$ is $\#\W{1}$-hard.
\end{corollary}
\begin{proof}
    Recall that every minor-closed property is computable, by the Robertson-Seymour theorem~\cite{RobertsonS04}. If $\Phi$ is trivial, then on input $G$ and $k$ we output either $0$ or $\binom{|V(G)|}{k}$, depending on whether $\Phi=0$ or $\Phi=1$.
    If $\Phi$ has bounded independence number, we can, on input $G$ and $k$, invoke the algorithm from Theorem~\ref{thm:main_indsubgraphs} on every graph $H$ on $k$ vertices and for which $\Phi(H)=1$. Afterwards, we output the sum of those counts.
    The $\#\W{1}$-hardness of the remaining cases holds by Theorem~\ref{thm:indsubsphi_hard_cor}.
\end{proof}

%% file: homcount.tex
\subsection{Hardness Results}
Let $\boxplus_k$ be the $k$-by-$k$ grid graph.\footnote{Formally, $\boxplus_k$ has vertex set $\{(i,j)~|~0\leq i,j \leq k-1\}$, and has an edge between $(i,j)$ and $(i',j')$ if and only if $|i-i'|+|j-j'|=1$.}
We say that a class of graphs $C$ has induced grid minors of unbounded size if, for every positive integer $k$, there exists $H\in C$ such that $\boxplus_k$ is an induced minor of $H$.
By leveraging the $F$-gadgets results of Section~\ref{sec:gadgets}, we give the following hardness result:
\begin{theorem}[Theorem~\ref{thm:intro_homs_hard} restated]\label{thm:igm_homsprobd_hardness}
Let $C$ be a computable class of graphs. If $C$ has unbounded induced grid minors, then $\#\homsprobd(C)$ is $\#\W{1}$-hard.
\end{theorem}
Theorem~\ref{thm:igm_homsprobd_hardness} follows immediately by Lemma~\ref{lem:induced_minor_char} and Lemma~\ref{lem:hardness_F_gadgets}, that we prove in the rest of the section.

\begin{lemma}\label{lem:induced_minor_char}
$C$ has unbounded induced grid minors if and only if $\gadgets(C)$ has unbounded treewidth.
\end{lemma}
\begin{proof}
    We prove the two directions separately. 
    \begin{claim}
        If $C$ has unbounded induced grid minors, then $\gadgets(C)$ has unbounded treewidth.
    \end{claim}
    \begin{claimproof}
    Since $C$ has unbounded induced grid minors, for any $k>0$ there exists $H\in C$ such that $\boxplus_{2k}$ is an induced minor of $H$. We show that this implies that $H$ has an $F$-gadget where $F=\boxplus_k$.
    This implies the claim, since $\boxplus_k$ has treewidth $k$ (cf.\ Chapter~7.7.1 in~\cite{CyganFKLMPPS15}).

    Let $W$ be the witness structure for $\boxplus_{2k}$ in $H$. For every $(i,j)\in V(\boxplus_{2k})$, let $B(i,j)$ be the subset of $V(W)$ corresponding to $(i,j)$. An $F$-gadget of $H$ can be constructed as follows:
    \begin{itemize}
        \item For each $(i,j)\in V(F)$, define $S_{(i,j)} = B(2i,2j)$.
        \item For each $e\in E(F)$, define $P_e$ as follows. Without loss of generality assume $e=\{(i,j),(i+1,j)\}$; the case where $e=\{(i,j),(i,j+1)\}$ is symmetric.
        First, note that $H$ must contain an induced path $p_e$ between $S_{(i,j)}=B(2i,2j)$ and  $S_{(i+1,j)}=B(2(i+1),2j)$ such that all intermediate vertices are in $B(2i+1,2j)$.
        This is true since, by definition of the minor, there are edges between $B(2i,2j)$ and $B(2i+1,2j)$ and between $B(2i+1,2j)$ and $B(2(i+1),2j)$, but no edges between $B(2i,2j)$ and $B(2(i+1),2j)$.
        We set $P_e$ to be the set of intermediate vertices of $p_e$.
        \item We let $R$ contain all remaining vertices of $V(H)$.
    \end{itemize}
    We argue that the resulting partition $(\mathcal{S},\mathcal{P},R)$ of $V(H)$ is an $F$-gadget, by checking that it satisfies the three properties of Definition~\ref{def:gadgets}.
    \begin{enumerate}
        \item $\forall\, u\in V(F)$, $S_{u} \ne \emptyset$ and $H[S_{u}]$ is connected. This holds true since in our case $S_{u}=B_{(2i,2j)}$, and by the properties of $B$.
        \item $\forall\, e\in E(F)$: $P_e \ne \emptyset$, and if $e=\{u,v\}$, then there are two vertices $s_{u,e} \in S_u$ and $s_{v,e} \in S_v$ such that the subgraph $p_e := H[\{s_{u,e}\} \cup P_e \cup \{s_{v,e}\}]$ is a simple path with endpoints $s_{u,e}$ and $s_{v,e}$.
        All these conditions hold by construction, see above.
        \item $\forall\, e_h \in E(H)$, either $e_h \in H[S_v]$ for some $v \in V(F)$, or $e_h \in p_e$ for some $e \in E(F)$, or $e_h$ is incident to a vertex of $R$.
        To see this, note that each edge of $H$ which is not incident to a vertex of $R$ must be fully contained in some $H[S_{(i,j)}]$ or in some $p_e$. The latter holds true, as the paths $p_e$ are induced and there is no edge between $S_{(i,j)}=B_{(2i,2j)}$ and $S_{(i',j')}=B_{(2i',2j')}$ for any $(i,j)\neq(i',j')$.
    \end{enumerate}
    Hence, $H$ has an $F$-gadget for $F=\boxplus_k$.
\end{claimproof}
    \begin{claim}
        If $\gadgets(C)$ has unbounded treewidth, then $C$ has unbounded induced grid minors.
    \end{claim}
    \begin{claimproof}
    Since $\gadgets(C)$ has unbounded treewidth, by the Excluded Grid Theorem~\cite{RobertsonS86-ExGrid}, for any integer $k>0$ there exists some $F \in \gadgets(C)$ such that $\boxplus_k$ is a minor of $F$. Fix such an $F$, and consider any $H \in C$ such that $H$ has an $F$-gadget.
    We show that $H$ contains $\boxplus_k$ as induced minor.
    
    Let $M=(M_{(i,j)})_{0\leq i,j\leq k-1}$ be the model of $\boxplus_k$ in $F$.
    Using $M$, we define a witness structure $W=(W_{(i,j)})_{0\leq i,j\leq k-1}$ of $\boxplus_k$ in $H$, as follows. For every vertex $(i,j)\in V(\boxplus_k)$, we let $E(i,j)$ be the set of edges $e=\{u,v\}$ of $F$ such that:
    \begin{itemize}
        \item[(a)] $\{u,v\} \subseteq M_{(i,j)}$, or
        \item[(b)] $u\in M_{(i,j)}$ and $v\in M_{(i+1,j)}$ (if $i<k-1$), or
        \item[(c)] $u\in M_{(i,j)}$ and $v \in M_{(i,j+1)}$ (if $j<k-1$).
    \end{itemize}
    Now, by hypothesis, $H$ has an $F$-gadget $(\mathcal{S},\mathcal{P},R)$, where $\mathcal{S}=\{S_v\}_{v \in V(F)}$ and $\mathcal{P}=\{P_e\}_{e \in E(F)}$.
    We set:
    \begin{linenomath*}
    \begin{equation}
    W_{(i,j)} = \dot{\bigcup_{v\in M_{(i,j)}}} S_v ~~~\dot\cup~~~ \dot{\bigcup_{e\in E(i,j)}}P_e\,,
    \end{equation}
    \end{linenomath*}
    We observe that the sets $W_{(i,j)}$ are pairwise disjoint by construction. This holds because (i) the sets $M_{(i,j)}$ are pairwise disjoint by definition, and the sets $S_v$ are pairwise disjoint by definition, and (ii) the sets $E(i,j)$ are pairwise disjoint by construction, and the sets $P_e$ are pairwise disjoint by definition. It remains to show that $W=(W_{(i,j)})_{0\leq i,j\leq k-1}$ is a witness structure of $\boxplus_k$ in $H$.
   
   First, let us argue that $H[W_{(i,j)}]$ is connected for each $(i,j)\in V(\boxplus_k)$. By construction, each vertex $a$ in $W_{(i,j)}$ is either contained in $S_v$ for some $v\in M_{(i,j)}$, or in $P_e$ for some $e\in E(i,j)$. In the latter case, by definition of $E(i,j)$ and of $F$-gadgets, $a$ must be connected to a vertex $b\in S_u$ for some $u\in M_{(i,j)}$ via a path containing only vertices in $P_e$. Since $P_e \subseteq W_{(i,j)}$ it remains to show that any pair $a\in S_v$ and $b\in S_u$ for $u,u'\in M_{(i,j)}$ is connected in $H[W_{(i,j)}]$. By definition of $F$-gadgets, each $S_u$ induces a connected graph in $H$, and by definition of models, $u$ and $u'$ are connected in $F[M_{(i,j)}]$. Since $P_e$ is a subset of $W_{(i,j)}$ for each $e\in E(F)$ fully contained in $M_{(i,j)}$ (see (a)), we conclude that $a$ and $b$ are indeed connected via a path in $H[W_{(i,j)}]$.
   
   Finally, we have to show that there is an edge between $W_{(i,j)}$ and $W_{(i',j')}$ in $H$ if and only if $(i,j)$ and $(i',j')$ are adjacent in $\boxplus_k$.
   For the ``if'' direction, assume $(i,j)$ and $(i',j')$ are adjacent in $\boxplus_k$; we only consider the case $i'=i+1$ and $j'=j$; since the case $i'=i$ and $j'=j+1$ is symmetric. By definition of model, there exist vertices $u\in M_{(i,j)}$ and $v\in M_{(i+1,j)}$ such that $e:=\{u,v\}\in E(F)$. By construction of $W$, we have $P_e \subseteq W_{(i,j)}$ and $S_v \subseteq W_{(i+1,j)}$, and by definition of $F$-gadget, one vertex of $P_e$ (the last vertex of the path $p_e$) is adjacent to a vertex in $S_v$, concluding the proof of this direction.
   For the ``only if'' direction, assume that $(i,j)$ and $(i',j')$ are not adjacent in $\boxplus_k$. By definition of $F$-gadget, two vertices $a,b\in \dot{\bigcup}_{v\in V(F)} S_v ~\dot\cup~ \dot{\bigcup}_{e\in E(F)}P_e$ can be adjacent in $H$ only if one of the following is satisfied:
   \begin{itemize}
       \item $\{a,b\} \subseteq S_v$ for some $v\in V(F)$, or
       \item $\{a,b\} \subseteq P_e$ for some $e\in E(F)$, or
       \item $a\in S_v$ for some $v\in V(F)$ and $b\in P_e$ for some $e\in E(F)$, with $e$ incident to $v$.
   \end{itemize}
   By construction of $W$, we included only vertices of $P_e$ in a block of $W$ if the edge $e$ is either fully contained in one block (see (a)) or if $e$ connects two adjacent blocks (see (b) and (c)).
   Consequently, there are no edges between $W_{(i,j)}$ and $W_{(i',j')}$ for non-adjacent $(i,j)$ and $(i',j')$. This concludes the second direction and thus the proof of this claim.
    \end{claimproof}
    This concludes the proof of Lemma~\ref{lem:induced_minor_char}.
\end{proof}

\begin{lemma}\label{lem:hardness_F_gadgets}
If $\gadgets(C)$ has unbounded treewidth then $\#\homsprobd(C)$ is $\#\W{1}$-hard.
\end{lemma}
\begin{proof}
First, if $\gadgets(C)$ has unbounded treewidth, then $\#\homsprob(\gadgets(C))$ is $\#\W{1}$-hard by the well-known classification of Dalmau and Jonsson~\cite{DalmauJ04}.
Now, we prove the following chain of reductions:
\begin{linenomath*}
\begin{equation}
\#\homsprob(\gadgets(C)) \fptred \#\cphomsprob(\gadgets(C)) \fptred \#\cphomsprobd(C) \fptred \#\homsprobd(C)    
\end{equation}
\end{linenomath*}
The first inequality is a well-known fact, see for instance~\cite{DorflerRSW19}.
The second inequality follows from Lemma~\ref{lem:FGadget_main}.
The third inequality can be proven by reducing $\#\cphomsprobd(C)$ to $\#\homsprobd(C)$ via the successive application of Fact~\ref{fact:cphoms_to_cfhoms} and Fact~\ref{fact:cfhoms_to_homs}; note that, for the latter, we only need to query the oracle for induced subgraphs of $G'$, which cannot increase its degeneracy.
The proof is complete.
\end{proof}

%% file: hom_ub.tex
\subsection{Some Tractability Results}
\newcommand{\dagH}{\orient{H}}
\label{sec:hom_ub}
In this section we give some tractability results for $\#\homsprobd(C)$. The first result links the dag-treewidth of $\orient{H}$ to the treewidth of an $F$-gadget of $H$. The second result connects the dag-treewidth of $\orient{H}$ to the cliquewidth of its skeleton graph, described below. In both cases, we obtain fixed-parameter tractability even when the treewidth of $H$ is unbounded.

Let $\dagH$ be any directed acyclic graph. Recall that we denote by $S$ the set of sources of $\dagH$, and for any $U \subseteq V(\dagH)$ we denote by $\dagH(U)$ the set of all vertices that are reachable from some $u \in U$. If $v \in V(\dagH)$ is such that $v \in \dagH(\{s\}) \cap \dagH(\{s'\})$ for two distinct sources $s,s'$ then we say $v$ is a \emph{joint}. We write $J(U)$ for the set of all joints reachable from some $u \in U$. When $U$ is a singleton $\{u\}$, we may simply write $J(u)$, $\dagH(u)$, and so on.

\subsubsection{Upper Bounds from the Treewidth of an F-gadget}
For any subset $U \subseteq V(H)$, the \emph{local sources} of $\orient{H}[U]$ are the sources of $\orient{H}[U]$ taken as a standalone graph. The set of local sources of $\orient{H}[U]$ is denoted by $\ls(U)$.  We prove:
\begin{theorem}\label{thm:F_tw_H_dtw}
Suppose $H$ has an $F$-gadget $(\scS,\scP,R)$, and let $\orient{H}$ be an acyclic orientation of~$H$. Then
\begin{linenomath*}
\begin{equation}
\dtw(\orient{H}) \le r + s \cdot (\tw(F)+1)^2
\end{equation}
\end{linenomath*}
where $r = |\ls(R)|$ and $s$ is the maximum number of local sources in any block of $\scS$ and $\scP$,
\begin{linenomath*}
\begin{equation}
s = \max\left(\max_{v \in V(F)} \big|\ls(S_v)\big|, \max_{e \in E(F)} \big|\ls(P_e)\big| \right)
\end{equation}
\end{linenomath*}
\end{theorem}
Before proceeding to the proof, we need some notation and a technical lemma.
Let $T$ be a tree decomposition of $F$. Consider each bag $t \in T$. We write $e \in t$ in place of $e \subseteq t$, to denote that the endpoints of $e$ are both in $t$. We define the following sets:
\begin{linenomath*}
\begin{align}
    V(t) &= \left(\bigcup_{v \in t} S_v\right) \cup \left(\bigcup_{e \in t} P_e \right)
\\    
    \ls(t) &= \left(\bigcup_{v \in t} \ls(S_v)\right) \cup \left(\bigcup_{e \in t} \ls(P_e) \right)
\end{align}
\end{linenomath*}
Note that $V(t) \subseteq V(\orient{H}(\ls(t)))$; that is, every vertex of $V(t)$ is reachable from some $s \in \ls(t)$.

\begin{lemma}\label{lem:T_path_intersect}
Let $t_1,t_2 \in T$ and let $t \in T(t_1,t_2)$. For any two vertices $u_1 \in V(t_1)$ and $u_2 \in V(t_2)$, any path in $H \setminus R$ between $u_1$ and $u_2$ intersects $V(t)$.
\end{lemma}
\begin{proof}
Let $\pi$ be any such path. Since $F$ is an induced minor of $H$, then $\pi$ is the concatenation of $\ell \ge 1$ paths $\pi_1\, \ldots \, \pi_{\ell}$, such that every $\pi_i$ lies entirely inside $S_v$ or $P_e$ for some $v \in V(F)$ or $e \in E(F)$. This gives us a sequence $x_1,\ldots,x_{\ell}$ where the elements with odd index are vertices (edges) of $F$ and those with even index are edges (vertices) of $F$, such that $x_1 \in t_1$ and $x_{\ell} \in t_2$, and that $x_i$ and $x_{i+1}$ form an incident pair for every $i=1,\ldots,\ell-1$ if $\ell \ge 2$. Thus, we have an alternating sequence of incident vertices and edges from $F[t_1]$ to $F[t_2]$. By the properties of tree decompositions, it follows that $t$ must contain at least one $x_i$ in this sequence (whether it is a vertex or an edge). By construction this implies that $\pi_i \subseteq S_v$ for some $v \in t$ or $\pi_i \subseteq P_e$ for some $e \in t$. But then $S_v, P_e \subseteq V(t)$, hence $\pi$ intersects $V(t)$, as claimed.
\end{proof}

\begin{proof}[Proof of Theorem~\ref{thm:F_tw_H_dtw}]
Let $T$ be an optimal tree decomposition of $F$. We build a d.t.d.\ $\scT$ of $H$ as follows.
For each $t \in T$, define the bag:
\begin{linenomath*}
\begin{equation}
    B_t = \ls(R) \,\cup\, \ls(t)
\end{equation}
\end{linenomath*}
Then, let $\scT$ be the tree obtained from $T$ by replacing $t$ with $B_t$ for all $t \in T$.

First, let us prove that $\scT$ is a valid dtd for $\orient{H}$ by checking Definition~\ref{def:gdtd}. Property (1) holds by construction. For property (2), consider any $u \in V(H)$. If $u \in R$, then by construction $u \in \orient{H}(B)$ for any bag $B$. If instead $u \notin R$, then $u \in S_v$ for some $v \in V(F)$ or $u \in P_e$ for some $e \in E(F)$. In any case, since $T$ is a tree decomposition of $F$, there exists $t \in T$ such that $v \in t$ or $e \in t$. Hence, $u$ is reachable from $\ls(t)$, and therefore from $(B_t)$, since $\ls(t) \subseteq B_t$.

We now prove property (3), using $B=B_t$ and $B_1=B_{t_1}$ and $B_2=B_{t_2}$. If $\orient{H}(B_{t_1}) \cap \orient{H}(B_{t_2}) = \emptyset$ then the claim is trivial. Otherwise, fix any vertex $z \in \orient{H}(B_{t_1}) \cap \orient{H}(B_{t_2})$. If $z$ is reachable from some $u \in \ls(R)$, then $z \in \orient{H}(B_t)$ as said above. We can therefore continue assuming that $z$ is not reachable from $\ls(R)$. Hence there must be vertices $u_1 \in \ls(t_1)$ such that $z \in \orient{H}(u_1)$, and $u_2 \in \ls(t_2)$ such that $z \in \orient{H}(u_2)$. Note that $u_1 \in V(t_1)$ and $u_2 \in V(t_2)$.

Now let $\pi(u_1,z)$ be any directed path from $u_1$ to $z$. Note that $\pi(u_1,z) \cap R = \emptyset$ since otherwise $z$ would be reachable from $\ls(R)$, contradicting our assumptions. Define similarly $\pi(u_2,z)$ and let $\pi(z,u_2)$ be its reverse. Now, we consider $\pi(u_1,z,u_2)$, the path obtained by concatenating $\pi(u_1,z)$ and $\pi(z,u_2)$. Since $u_1 \in V(t_1)$ and $u_2 \in V(t_2)$, by Lemma~\ref{lem:T_path_intersect} $\pi(u_1,z,u_2)$ contains a vertex $u^* \in V(t)$. But $z$ is reachable from $u^*$ by construction of $\pi(u_1,z,u_2)$, and so $z$ is reachable from $\ls(t)$ as well. Since $B_t \supseteq \ls(t)$, then $z \in V(\orient{H}(B_t))$, as claimed.

Finally, we show that $\dtw(\scT) \le r + s \cdot (\tw(T)+1)^2$. Consider any bag $B_t \in \scT$. Clearly $|B_t| \le r + |\ls(t)|$, so we prove that $|\ls(t)| \le s \cdot (\tw(T)+1)^2$. In turn, note that for any set $S_v$ or $P_e$ such that $v \in t$ or $e \in t$, by definition we add at most $s$ vertices to $B_t$. Therefore we need only to see that $t$ contains at most $(\tw(T)+1)^2$ vertices and/or edges of $F$. This follows by straightforward calculations as $t$ contains at most $\tw(T)+1$ vertices.
\end{proof}

Consider, for instance, a graph $H$ which is obtained from a graph $F$ by first subdividing every edge, yielding new vertices $v_e$ for each $e\in E(F)$, and, second, substituting every primal vertex $v\in V(F)$ by a clique $C_v$. Then $H$ can have arbitrary large treewidth (just substitute by large enough cliques). However, $H$ also has an $F$-gadget $(\mathcal{S},\mathcal{P},\emptyset)$ where $S_v=V(C_v)$ for each $v\in V(F)$, and $P_e = \{v_e\}$ for each $e\in E(F)$. Since each $S_v$ and each $P_e$ only have one local source for any acyclic orientation of $H$, we obtain by the previous result that the dag treewidth of $H$ is bounded by $(\mathsf{tw}(F)+1)^2$. In other words, if $F$ has small treewidth, then $H$ has small dag treewidth even though it can have arbitrarily large treewidth.

\subsubsection{Upper Bounds from the Clique-Width of the Skeleton Graph}
In this section, we use the original definition of dag tree decomposition by~\cite{Bressan21}. The difference with Definition~\ref{def:gdtd} is that each bag must contain only sources of $\dagH$, whereas our definition allows for arbitrary vertices. Clearly, the dag treewidth of $\dagH$ according to this definition is an upper bound to the dag treewidth according to Definition~\ref{def:gdtd}. Therefore, the upper bounds of this section are upper bounds to our generalised notion of dag treewidth.
\begin{definition}[\cite{Bressan21}, Definition 2.]
\label{def:dtd}
A dag tree decomposition (d.t.d.) of $\orient{H}$ is a rooted tree $\scT=(\bags,\et)$ such that: 
\begin{enumerate}[itemsep=4pt,parsep=0pt,topsep=2pt]
\item $B \subseteq S$ for all $B \in \bags$
\item $\bigcup_{B \in \bags} V(\orient{H}(B)) = V(H)$
\item for all $B,B_1,B_2 \in \bags$, if $B \in \scT(B_1,B_2)$ then $V(\orient{H}(B_1)) \cap V(\orient{H}(B_2)) \subseteq V(\orient{H}(B))$
\end{enumerate}
\end{definition}
\noindent The reason for using this definition is that, instead of working with $\dagH$, we can work with a reachability graph called \emph{skeleton}.
\begin{definition}[\cite{Bressan21}, Definition 7]
\noindent The \emph{skeleton} of a dag $\dagH=(V(\dagH),A(\dagH))$ is the bipartite dag $\skel(\dagH) = (V_{\sk},E_{\sk})$ where $V_{\sk} = S \cup J$ and $E_{\sk} \subseteq S \times J$, and $(u,v) \in E_{\sk}$ if and only if $v \in J(u)$.
\end{definition}
\noindent Indeed:
\begin{lemma}
$\scT$ is a DTD for $\dagH$ if and only if $\scT$ is a DTD for $\skel(\dagH)$.
\end{lemma}
\begin{proof}
By construction, $S(\skel(\dagH))=S(\dagH)$, $J(\skel(\dagH))=J(\dagH)$, and for any nodes $s \in S(\dagH)$ and $v \in J(\dagH)$, $v$ is reachable from $s$ in $\skel(\dagH)$ if and only if $v$ is reachable from $s$ in $\dagH$. Hence $\scT$ satisfies the properties of Definition~\ref{def:dtd} for $\skel(\dagH)$ if and only if $\scT$ satisfies the properties of Definition~\ref{def:dtd} for $\dagH$.
\end{proof}
Figure~\ref{fig:tree_decomp} gives an example of skeleton graph. Note that $\skel(\dagH)$ does not contain nodes that are neither sources nor joints, as they are irrelevant as far as $\scT$ is concerned. Note also that computing $\skel(\dagH)$ takes time $|H|^{O(1)}$.

\begin{figure}[h]
\centering
\resizebox{.8\textwidth}{!}{%
\input{pattern3.tex}
\hspace*{40pt}
\input{skel3.tex}
}
\caption{Left: a generic directed acyclic graph. Right: its skeleton (sources above, joints below).}
\label{fig:tree_decomp}
\end{figure}
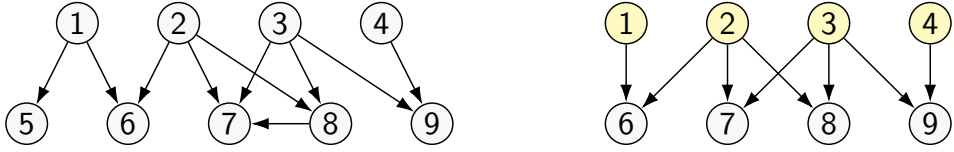

\paragraph*{Clique-width.}
\newcommand{\CliqueT}{T}
The notions of \emph{clique decomposition} and \emph{clique-width} of a graph were defined by Courcelle and Olariu in~\cite{Courcelle2000-CW}. A clique decomposition is defined by a rooted parse tree $\CliqueT$ (or, equivalently, by what is called a $k$-formula).
Each node $x \in \CliqueT$ specifies an operation $O_x$ among the ones listed below.
Let $\CliqueT_x$ be the subtree of $\CliqueT$ rooted at $x$, and $H_x$ be the subgraph of $H$ obtained by parsing $\CliqueT_x$.
We let $H=H_{\rt(\CliqueT)}$.
Each vertex $v \in H_x$ has a label $i_x(v) \in [k]$, and we denote by $\Lab_x(i)$ the set of vertices of $H_x$ with label $i$.
Each operation $O_x$ is one of:
\begin{itemize}
    \item[-] \create($i$): let $H_x=(\{v\},\emptyset)$ and $i_x(v)=i$. Allowed only if $x$ is a leaf.
    \item[-] \unite: let $H_x$ be the disjoint union of $H_{x'}$ and $H_{x''}$, where $x',x''$ are the children of $x$ in $\CliqueT$. Allowed only if $x$ has two children.
    \item[-] \link($i,j$): let $H_x$ be obtained from $H_{x'}$ by setting $V(H_x)=V(H_{x'})$ and $E(H_{x})=E(H_{x'}) \cup (\Lab_{x'}(i) \times \Lab_{x'}(j))$, where $x'$ is the child of $x$. So, we are adding edges between each vertex labelled $i$ and each vertex labelled $j$ in $H_{x'}$. Allowed only if $x$ has one child.
    \item[-] \relabel($i,j$): let $H_x$ be obtained from $H_{x'}$ by relabeling all $v \in \Lab_{x'}(i)$ as $j$, where $x'$ is the child of $x$. Allowed only if $x$ has one child.
\end{itemize}
The smallest number of distinct labels needed to build $H$ using such a tree is the clique-width of $H$, denoted by $\cw(H)$.
We will show that $\tau(\dagH) \le \cw(\skel(\dagH))$, where $\cw(\skel(\dagH))$ is meant on the undirected version of $\skel(\dagH)$.
As a consequence, we will prove the following result. First, let us introduce:
\begin{definition}
The \emph{skeleton clique-width} of $H$ is:
\begin{linenomath*}
\begin{equation}
    \skelcw(H) := \max_{\dagH} \cw(\skel(\dagH))
\end{equation}
\end{linenomath*}
where the $\max$ is over all acyclic orientations of $H$.
\end{definition}
The goal of this section is to prove Theorem~\ref{thm:intro_homs_tract}, which is an immediate consequence of combining the dynamic programming algorithm along dag treewidth~\cite[Theorem~9]{Bressan21} with the following bound on the skeleton clique-width: 
\begin{theorem}\label{thm:skelcw}
Every graph $H$ satisfies $\tau_1(H) \le \skelcw(H)$.
\end{theorem}

\paragraph*{Preliminaries.}
We will consider a clique decomposition $\CliqueT$ of the undirected version of $\skel(\dagH)$. To avoid cluttering the notation, however, we will denote $\skel(\dagH)$ simply by $H$. Bear in mind that $H$ is bipartite into sources and joints, a property that will be used in our proofs.

\noindent For any $x,y \in \CliqueT$ we denote by $\CliqueT(x,y)$ the unique simple path from $x$ to $y$ in $\CliqueT$. We denote by $\prec$ the partial order given by $\CliqueT$ over its nodes. That is, $x \preceq y$ means that $x$ is $y$ or an ancestor of $y$, so $x \in \CliqueT(y,\rt(\CliqueT))$, and $x \prec y$ means that $x$ is an ancestor of $y$, so $x \in \CliqueT(y,\rt(\CliqueT)) \setminus \{y\}$. Note that $\CliqueT$ has the following basic ``label propagation'' property: 
\begin{fact}
If $x \preceq y$ then $\Lab_y(i_y(v)) \subseteq \Lab_x(i_x(v))$.
\end{fact}
In words, the set of vertices having the same label as $v$ can only grow as we move towards the root of $\CliqueT$.

We say that a node $v$ is \emph{created} at $x \in \CliqueT$ if $x$ is a leaf and $v \in V(H_x)$. We say that an edge $(u,v)$ is \emph{created} at $x \in \CliqueT$ if $u,v \in H_v$ and $O_x = \link(i_x(u),i_x(v))$. Note that an edge can be created at multiple nodes. Given $v \in H$, we denote by \xon($v$) the leaf of $\CliqueT$ where $v$ is created, and by \xoff($v$) the node of $\CliqueT$ closest to the root of $\CliqueT$ where an edge insisting on $v$ is created. Obviously, $\xoff(v) \preceq \xon(v)$.

\begin{definition}
The \emph{active span} of $v$ is $\aspan(v)=\CliqueT(\xon(v),\xoff(v))$.
\end{definition}
If $x \in \aspan(v)$ then we say $v$ is \emph{active} at $x$. We denote by $A_x$ the vertices active at $x$, and by $A_x(i) = A_x \cap \Lab_x(i)$ the vertices with label $i$ that are active at $x$.

\paragraph*{Designated sources.}
We assign a \textit{designated source} $s_v$ to each vertex $v \in H$. If $v \in S$, then $s_v=v$. Else, we let $s_v = \min S_v$, where: 
\begin{linenomath*}
\begin{equation}
    S_v = \{s \in S: (s,v) \text{ created at } \xoff(v)\}
\end{equation}
\end{linenomath*}
Thus, to choose $s_v$, we look at all edges incident to $v$ that are created at $\xoff(v)$ and choose the smallest source appearing in those edges.

Before continuing, we give some useful properties.
\begin{claim}
\label{claim:props}
If for some $x \in \CliqueT$ we have $u \in A_x$ and $i_x(u)=i_x(v)$, then
\begin{enumerate}[leftmargin=4em,topsep=0pt,label=\roman*)]
\item $\xoff(u)=\xoff(v)$
\item $v \in A_x$
\item $\{u,v\} \subseteq J$ or $\{u, v\} \subseteq S$
\item $u \in J$ $\Longrightarrow$ $s_u=s_v$
\end{enumerate}
\end{claim}
\begin{claimproof}
\begin{enumerate}[leftmargin=2em,label=\roman*)]
\item since $i_{x}(u)=i_{x}(v)$, then, at any $y \preceq x$, an edge incident to $u$ is created if and only if an edge incident to $v$ is created.
Thus, $\xoff(u)=\xoff(v)$.

\item on the one hand, $\xon(v) \succeq x$, since $v \in H_x$. On the other hand, $\xoff(v) \preceq x$, because $\xoff(u)=\xoff(v)$ by (i), and $\xoff(u) \preceq x$ since $u \in H_x$. We conclude that $x \in \aspan(v)$, i.e., $v \in A_x$.

\item suppose by contradiction that $u \in S$ and $v \in J$.
Let $y=\xoff(u)=\xoff(v)$.
As noted above $i_y(u)=i_y(v)$, thus at $y$ we create edges between $u$ and $z$ and between $v$ and $z$, for some $z \in H$.
This contradicts the assumption that $H$ is bipartite into sources and joints.

\item recall the set $S_v$ above.
The proof of (i) implies that $S_u=S_v$, and therefore $s_u=s_v$.
\end{enumerate}
\end{claimproof}

\begin{claim}
\label{claim:active}
At any $x \in \CliqueT$, for any label $i$, either $A_x(i)=\emptyset$ or $A_x(i)=\Lab_x(i)$.
\end{claim}
\begin{claimproof}
Follows from point (ii) of Claim~\ref{claim:props}.
\end{claimproof}

\paragraph*{Representants.}
Recall that $A_x(i) = A_x \cap \Lab_x(i)$ is the set of active vertices with label $i$.
\begin{definition}
The representant of label $i$ at $x$ is defined, if $A_x(i) \ne \emptyset$, as $r_x(i) \,\defeq\, \min A_x(i) = \min \Lab_x(i)$.
The set of representants at $x$ is  $R_x = \{r_x(i) : A_x(i) \ne \emptyset\}$.
\end{definition}
Note that, if $A_x(i) \ne \emptyset$, then by Claim~\ref{claim:active} $A_x(i) = \Lab_x(i)$, so the definition above is always consistent.
Note also that $|R_x| \le k$.

\begin{claim}[Monotonicity of representant sources]
\label{claim:src_repr_cont}
Let $s \in S$. Suppose that $s \in A_y \setminus R_y$.
Then, $s \notin R_x$ for all $x \preceq y$.
\end{claim}
\begin{claimproof}
Let $i=i_y(s)$.
Since $s \in A_y$, we have $\emptyset \ne A_y(i)=L_y(i)$ by Claim~\ref{claim:active}.
However, since $s \notin R_y$, we must have $s > r_y(i) = \min L_y(i)$.
Now, since $x \preceq y$, we have $\Lab_y(i) \subseteq \Lab_x(i_x(s))$.
Therefore, if $r_x(i_x(s))$ exists, it will satisfy $r_x(i_x(s)) = \min \Lab_x(i_x(s)) \le \Lab_y(i) < s$.
Thus $s \ne r_x(i_x(s))$, i.e., $s \notin R_x$.
\end{claimproof}

\paragraph*{Constructing a DTD from a clique decomposition.}
From the parse tree $\CliqueT$, we construct a rooted dag tree decomposition $\scT$ by replacing each node $x$ of $\CliqueT$ with the bag:
\begin{linenomath*}
\begin{equation}
B(x) = \{s_{v} : v \in R_x\}
\end{equation}
\end{linenomath*}
That is, $B(x)$ contains the designated source of each representant at $x$. Note that in general we can have $B(x) \nsubseteq H_x$. Before proving that $\scT$ is a DTD of $\skel(\dagH)$, we need some ancillary notions and claims.

\begin{definition}
An edge $(s,v) \in H$ is \emph{broken at $x$} if $|\{s,v\} \cap H_x| = 1$.
\end{definition}
Clearly, if $(s,v) \in H$ is broken at $x$ then $(s,v)$ is created at some $y \preceq x$.

\begin{claim}[Active joints have their designated sources in the bag]
\label{cl:active_v_has_sv_in_bag}
\label{cl:sv_in_bag}
If $v \in J$ and $v \in A_x$, then $s_v \in B(x)$.
\end{claim}
\begin{claimproof}
Let $i=i_x(v)$. By point (iv) of Claim~\ref{claim:props}, we have $s_v=s_u$ for all $u \in A_x(i)$. In particular, $s_v=s_u$ for $u = \min A_x(i) = r_x(i)$. But, by construction of $B(x)$, we have $s_u \in B(x)$.
\end{claimproof}

\begin{claim}[Joints of broken edges are covered]
\label{claim:broken_cover}
If $(s,v)$ is broken at $x$, then $v \in J(B(x))$.
\end{claim}
\begin{claimproof}
Suppose first $v \in H_x$. Then $\xoff(v) \preceq x$, so $v \in A_x$. Moreover, $v \in J$. Then $s_v \in B(x)$ by Claim~\ref{cl:active_v_has_sv_in_bag}. Thus, $v \in J(B(x))$.

Now suppose $s \in H_x$. Then $\xoff(s) \preceq x$, so $s \in A_x$. Let $i=i_x(s)$. Then $\emptyset \ne A_x(i) = L_x(i)$, by Claim~\ref{claim:active}, and so $r=r_x(i)$ exists. Now, let $y \preceq x$ be any node where $(s,v)$ is created. Then, the edge $(r,v)$ is created at $y$, too, since $i_y(r)=i_y(s)$. Thus $(r,v) \in H$, which means $v \in J(r)$. This implies $v \in J(B(x))$, since $r \in B(x)$ by construction of $B(x)$.
\end{claimproof}

\begin{claim}
\label{cl:sv_broken}
If $s \in B(x) \setminus H_x$, then the edge $(s,v) \in H$ is broken at $x$ for some joint vertex $v \in H_x$ such that $s_v=s$. Moreover, $s \in B(y)$ for all $a \preceq y \preceq x$, for the lowest $a \preceq x$ such that $s \in H_a$.
\end{claim}
\begin{claimproof}
If $s \in B(x) \setminus H_x$, then $s=s_v$ for some $v \in J \cap A_x$, so $(s,v)$ is broken at $x$. By definition of $a$, the edge $(s,v)$ is created at $a$, and therefore $v \in A_a$ as well. Hence, for all $y$ such that $a \preceq y \preceq x$, we have $v \in A_y$, which by Claim~\ref{cl:sv_in_bag} implies $s \in B(y)$.
\end{claimproof}

\begin{claim}
\label{cl:sv_broken_2}
Suppose $s \in H_y \cap B(y)$. Then $s \in H_x \Rightarrow s \in B(x)$ for any $x \succ y$.
\end{claim}
\begin{claimproof}
We distinguish two cases. First, suppose that $s \in R_y$. Then,
\begin{linenomath*}
\begin{align}
    s \in R_y &\Rightarrow s \in A_y && R_y \subseteq A_y
    \\ &\Rightarrow s \in A_x && s \in H_x \text{ and } x \succ y
    \\ &\Rightarrow s \in R_x && s \in R_y, y \prec x, \text{Claim }\ref{claim:src_repr_cont}
    \\ &\Rightarrow s \in B(x)
\end{align}
\end{linenomath*}
Now suppose that $s \notin R_y$. Then, by construction of $R_y$,
\begin{equation}
    s \notin R_y \Rightarrow s=s_v : (s,v) \in H, v \in R_y
\end{equation}
Thus, $v \in A_y$. Since $y \prec x$, then either $v \in H_x$ and $v \in A_x$, or $v \notin H_x$. If $v \in A_x$, then $s = s_v \in B(x)$ by Claim~\ref{cl:sv_in_bag}. If instead $v \notin H_x$, then $(s,v)$ is broken at $x$ with $s \in A_x$. If $s \in R_x$, then $s \in B(x)$. Otherwise, $s \notin R_x$, which means that $\min A_x(i) = s' < s$, for some source $s'$, where $i=i_x(s)$. But then, when $(s,v)$ is created, $(s',v)$ is created too. Thus $(s',v) \in H$. But then $s_v \le s'$, so $s_v \ne s$, a contradiction.
\end{claimproof}

We now prove:
\begin{lemma}
Let $\scT$ be the tree obtained from $\CliqueT$ as described above.
Then, $\scT$ is a valid DTD of $H$ of width $\le k$.
\end{lemma}
\begin{proof}
We check the three properties of Definition~\ref{def:dtd}. First, by construction $B(x) \subseteq S(H)$ for all $x \in \CliqueT$. Second, for all $s \in S(H)$  there exists $x \in \CliqueT$ such that $s \in B(x)$. To see this, take $x=\xon(s)$; then $s$ is active at $x$, and since it is also the only vertex in $H_x$, we have $s=r_x(i_x(s))$ and therefore $B(x)=\{s\}$. Clearly, $|B(x)| \le |R_x| \le k$, so the width of $\scT$ is at most $k$. It remains to show the third property, that is:
\begin{linenomath*}
\begin{equation}
x \in \CliqueT(x',x'')
\;\Rightarrow\;
J(B(x')) \cap J(B(x'')) \subseteq J(B(x))
\end{equation}
\end{linenomath*}
Clearly, we can assume $x \notin \{x',x''\}$. This implies $x \prec x'$ or $x \prec x''$, so, without loss of generality, we assume $x \prec x'$. Now consider any vertex $z \in J(B(x')) \cap J(B(x''))$.
By construction, $z \in J(s') \cap J(s'')$ for some $s' \in B(x')$ and $s'' \in B(x'')$. Hence, $s'=s_{v'}$ and $s''=s_{v''}$ for some $v' \in R_{x'}$ and $v'' \in R_{x''}$.

Let $S_z$ be the set of sources of $H$ having an edge to $z$. We consider three cases.

\textbf{Case 1}: $(z \cup S_z) \cap H_x \ne \emptyset$ and $(z \cup S_z) \nsubseteq H_x$. Then, some edge $(s,z) \in H$ is broken at $H_x$. Claim~\ref{claim:broken_cover} then implies $z \in J(B(x))$.

\textbf{Case 2}: $(z \cup S_z) \cap H_x=\emptyset$.
Then $s' \notin H_{x'}$, since otherwise, as $x \prec x'$, we would have $s' \in H_{x}$, contradicting $S_z \cap H_x = \emptyset$. Therefore, $s' \notin R_{x'}$. Yet, we know $s' \in B(x')$. Thus, we must have $s'=s_{v'}$, where $v'$ is some joint vertex of $H$ that satisfies $v' \in R_{x'} \subseteq A_{x'}$. But since $x \prec x'$, then $v' \in H_x$, too. Since $s' \notin H_x$, then the edge $(s',v')$ is broken at $x$. This implies $v' \in A_x$, which by Claim~\ref{cl:sv_in_bag} implies $s' = s_{v'} \in B(x)$, and thus $z \in J(B(x))$.

\textbf{Case 3}: $(z \cup S_z) \subseteq H_x$.
Denote by $a \in \CliqueT$ the lowest common ancestor of $x',x''$, let $\CliqueT_a$ be the subtree of $\CliqueT$ rooted at $a$. Note that $s'' \in H_a$, since $s'' \in S_z \subseteq H_x$, and since $a \prec x$ as $a$ is the node of $\CliqueT(x',x'')$ closest to $\rt(\CliqueT)$.

Now, we claim that $s'' \in B(a)$. If $a=x''$, this holds because by hypothesis $s'' \in B(x'')$.
Suppose instead that $a \ne x''$. Then $H_x \cap H_{x''} = \emptyset$, as $x$ and $x''$ belong to disjoint subtrees of $\CliqueT_a$. But $s'' \in H_x$, since $s'' \in S_z$ and we assume $(z \cup S_z) \subseteq H_x$. Therefore, $s'' \notin H_{x''}$, and moreover, $a$ is the lowest ancestor of $x''$ such that $s'' \in H_a$. By Claim~\ref{cl:sv_broken}, $s'' \in B(a)$.

Thus, $s'' \in H_a \cap B(a)$. Since $x \succ a$ and $s'' \in H_x$, by Claim~\ref{cl:sv_broken_2} $s'' \in B(x)$, and so $z \in J(B(x))$.
\end{proof}

%% file: pattern3.tex
\begin{tikzpicture}[
scale=0.6
]
\node[graph] (r1) at (0,0) {1};
\node[graph] (r2) at (2,0) {2};
\node[graph] (r3) at (4,0) {3};
\node[graph] (r4) at (6,0) {4};
\node[graph] (v1) at ($(r1)+(-1,-2)$) {5};
\node[graph] (v2) at ($(v1)+(2,0)$) {6};
\node[graph] (v3) at ($(v2)+(2,0)$) {7};
\node[graph] (v4) at ($(v3)+(2,0)$) {8};
\node[graph] (v5) at ($(v4)+(2,0)$) {9};
\path[->] (r1) edge (v1);
\path[->] (r1) edge (v2);
\path[->] (r2) edge (v2);
\path[->] (r2) edge (v3);
\path[->] (r2) edge (v4);
\path[->] (r3) edge (v3);
\path[->] (r3) edge (v4);
\path[->] (r3) edge (v5);
\path[->] (r4) edge (v5);
\path[->] (v4) edge (v3);
\end{tikzpicture}

%% file: skel3.tex
\begin{tikzpicture}[
scale=0.6
]
\node[source] (r1) at (0,0) {1};
\node[source] (r2) at (2,0) {2};
\node[source] (r3) at (4,0) {3};
\node[source] (r4) at (6,0) {4};
\node[joint] (j6) at ($(r1)+(0,-2)$) {6};
\node[joint] (j7) at ($(j6)+(2,0)$) {7};
\node[joint] (j8) at ($(j7)+(2,0)$) {8};
\node[joint] (j9) at ($(j8)+(2,0)$) {9};
\path[->] (r1) edge (j6);
\path[->] (r2) edge (j6);
\path[->] (r2) edge (j7);
\path[->] (r2) edge (j8);
\path[->] (r3) edge (j7);
\path[->] (r3) edge (j8);
\path[->] (r3) edge (j9);
\path[->] (r4) edge (j9);
\end{tikzpicture}

%% file: approx.tex
\section{Approximate Counting}\label{sec:approx}
In this section, we consider the complexity of approximately counting subgraphs and induced subgraphs in degenerate graphs. We will identify classes of graphs $C$ for which $\#\subsprobd(C)$ and $\#\indsubsprobd(C)$ can be approximated efficiently by a fixed-parameter tractable approximation scheme:
\begin{definition}[\cite{ArvindR02,Meeks16,DLM20}] 
Let $(P,\kappa)$ be a parameterized counting problem. A \emph{fixed-parameter tractable approximation scheme (FPTRAS)} for $(P,\kappa)$ is a randomised algorithm $\mathbb{A}$ that expects as input an instance $I$ of  $(P,\kappa)$ and $\varepsilon>0$, and returns an output $\hat{c}$ such that
\begin{linenomath*}
\begin{equation}\label{eq:eps_approx}
    \Pr[(1-\varepsilon)P(I)\leq \hat{c} \leq (1+\varepsilon)P(I)]\geq 2/3\,.
\end{equation}
\end{linenomath*}
The running time of $\mathbb{A}$ must be bounded by $f(\kappa(I))\cdot {poly}(|I|,\varepsilon^{-1})$, where $f$ is computable and independent of $I$.
\end{definition}
\noindent The error probability of an FPTRAS can be brought down to any $\delta>0$, thus matching the original definition of an FPTRAS by~\cite{ArvindR02}, by standard amplification techniques, at the cost of an additional factor of $O(\log(1/\delta))$ in the running time (see~\cite{DLM20}). In what follows, we will call $\hat c$ an $\varepsilon$\emph{-approximation} of $P(I)$ if $(1-\varepsilon)P(I)\leq \hat c \leq (1+\varepsilon)P(I)$.

Let us start by noting that, for $\#\subsprobd(C)$ and $\#\indsubsprobd(C)$, an FPTRAS may not always exist. To see this, let $C$ be the set of all $1$-subdivisions of cliques. Then the $k$-clique problem, which is $\W{1}$-hard, reduces to the decision problems $\subsprobd(C)$ and $\indsubsprobd(C)$, by computing the $1$-subdivision $G'$ of $G$, which has degeneracy at most $2$, and checking if it contains the subdivision of a $k$-clique. Since an $\varepsilon$-approximation of the number of subdivided $k$-cliques in $G'$ reveals whether or not $G$ contains a $k$-clique, we conclude that $\#\subsprobd(C)$ and $\#\indsubsprobd(C)$ do not have an FPTRAS unless $\mathrm{FPT}$ coincides with $\W{1}$ under randomised reductions.

Hence, in general $\#\subsprobd(C)$ and $\#\indsubsprobd(C)$ are unlikely to admit an FPTRAS. Nonetheless, we identify large classes of patterns $C$ for which an FPTRAS exists. For the remaining classes, we are not able to prove hardness results; but we point out that, for approximate counting problems, dichotomies into easy and hard cases are quite uncommon, and often even unlikely to exist~\cite{DyerGJ10}. Even in the case of non-degenerate graphs, the complexity of approximate counting subgraphs is only partially resolved: \cite{ArvindR02} established the existence of an FPTRAS for counting copies of patterns of bounded treewidth, but it it still not known if patterns of unbounded treewidth induce hardness.

For our results, presented in the subsequent subsections, we rely on a recent reduction from ``approximate counting to colourful decision'' for self-contained $k$-witness problems due to~\cite{DLM20}. To this end, we say that a graph $G$ is $k$-\emph{vertex coloured} if it comes with a (not necessarily proper) surjective colouring $c:V(G)\rightarrow [k]$.
Now let $C$ be a class of graphs. The problem $\mcolsubsprobd(C)$ expects as input a $k$-vertex coloured graph $G$ and a $k$-vertex pattern $H$, and the goal is to decide whether $G$ contains a subgraph isomorphic to $H$ that, additionally, contains each colour precisely once. The parameterization is given by $k+d(G)$. The problem $\mcolindsubsprobd(C)$ is the same, but for \emph{induced} multicoloured copies.

The following is a direct consequence of Theorem~9 in~\cite{DLM20}:\footnote{We remark that the formal statement of Theorem~9 in~\cite{DLM20} describes the running time of the algorithm for the multicolour decision version only as a function in $n$ and $k$. However, by mimicking its proof (that is, by applying~\cite[Theorem 1]{DLM20}) for our case of approximate counting copies and induced copies of $H$ in $G$, we see that the running time can also be expressed as a function in $G$ and $k$, and thus also in $d(G)$.}

\begin{theorem}[\cite{DLM20}]\label{thm:coldecapprox}
Let $C$ be a class of graphs. If $\mcolsubsprobd(C)$ can be solved in time $T(G,k,d(G))$, then there is a randomised algorithm $\mathbb{A}$ that, on input a $k$-vertex pattern $H\in C$, an $n$-vertex graph $G$, and $\varepsilon>0$, computes in time $\varepsilon^{-2}k^{2k}n^{o(1)}T(G,k,d(G))$ and with probability at least $2/3$ an $\varepsilon$-approximation of $\#\subs{H}{G}$. In particular, $\#\subsprobd(C)$ has an FPTRAS whenever $\mcolsubsprobd(C)$ is FPT. The same holds true for $\mcolindsubsprobd(C)$ and $\#\indsubsprobd(C)$.
\end{theorem}

Before we begin with our treatment of the multicolour decision versions, let us remark that, following Theorems~2 and~9 in~\cite{DLM20}, we also obtain efficient algorithms for $\varepsilon$-approximate \emph{sampling} of subgraphs and induced subgraphs in degenerate graphs whenever the multicolour decision version is fixed-parameter tractable.

\subsection{Detecting Multicoloured Subgraphs}
As a warm-up, we will begin with $\mcolsubsprobd(C)$, which turns out to be the significantly easier case when compared to $\mcolindsubsprobd(C)$, as we will see later. The reason for that is that the detection of multicoloured subgraphs is equivalent to detecting the existence of a homomorphism whose image is multicoloured. In particular, the constraint of hitting each colour precisely once makes sure that the homomorphism is injective. 

\begin{lemma}
Let $C$ be a class of graphs. There is a computable function $\hat{f}$ such that the problem $\mcolsubsprobd(C)$ can be solved in time $\hat{f}(k,d)\cdot n^{\tau_1(H)} \log n$, where $d=d(G)$ and $n=|V(G)|$. In particular, $\mcolsubsprobd(C)$ is FPT whenever the (generalised) dag treewidth of $C$ is bounded.
\end{lemma}
\begin{proof}
Given a $k$-vertex coloured graph $G$ of degeneracy $d$ and a $k$-vertex pattern $H\in C$, the algorithm proceeds as follows: First, we guess a (not necessarily proper) surjective $k$-vertex colouring $c_H:V(H)\rightarrow [k]$ of $H$. Note that there are $k!$ many choices as $|V(H)|=k$. Let $c_G$ be the $k$-vertex colouring of $G$.

Now, for each guess of $c_H$, we can decide whether there is a homomorphism $\varphi$ from $H$ to $G$ that respects $c_H$ (i.e., for each $v\in V(H)$ we have $c_G(\varphi(v))=c_H(v)$) in time $f(k,d)\cdot n^{\tau_1(H)} \log n$. To this end, we adapt the dynamic programming algorithm of~\cite{Bressan21} in a straightforward way: when enumerating the homomorphisms $\phi$ from (the subgraphs of) $H$ to $G$, we discard $\phi$ if it does not respect the colouring. It is immediate to see that the resulting algorithm counts the number of homomorphisms from $H$ to $G$ that respect $c_H$ and thus, in particular, decides if one such homomorphism exists.

The overall running time bound is achieved by setting $\hat{f}(k,d):=k!f(k,d)$.
\end{proof}

The result on approximate counting is an immediate consequence.

\begin{theorem}[Theorem~\ref{thm:intro_FPTRAS_subcount} restated]\label{thm:FPTRAS_subcount}
There exists a computable function $f$ and a randomised algorithm $\mathbb{A}$ that, on input an $n$-vertex graph $G$, a $k$-vertex pattern $H$, and $\varepsilon>0$, computes in time
$\varepsilon^{-2} \cdot f(k,d(G))\cdot n^{\tau_1(H)+o(1)}$ and with probability at least $2/3$, an $\varepsilon$-approximation of $\#\subs{H}{G}$. In particular, $\#\subsprobd(C)$ has an FPTRAS whenever the (generalised) dag treewidth of $C$ is bounded.
\end{theorem}
\begin{proof}
Follows immediately by combining the previous lemma with Theorem~\ref{thm:coldecapprox}.
\end{proof}

\textbf{Examples.} We give two examples of classes $C$ for which approximate counting is hard when parameterized only by $k$, but is FPT when parameterized by $k+d$. First, let $C$ be the class of wreath graphs, which are the disjunctive product (also called ``OR product'') of cycles and independent sets. It is not hard to see that if $H$ is a wreath graph, then every acyclic orientation $\orient{H}$ of $H$ has essentially the same dag tree decomposition of the corresponding oriented cycle, and in particular has $\tau(H)=2$. Since this holds for every orientation and every such $H$, we have $\tau_1(C)=2$. Hence, by Theorem~\ref{thm:FPTRAS_subcount} there exists an FPTRAS for counting the copies of any $H \in C$ when the parameter is $k+d$. Second, let $C$ be the class of all graphs $H$ that can be constructed through an F-gadget as follows. Let $F$ be any fixed graph. Then, let $H$ be any graph obtained from $F$ by exploding every node into a clique, and by subdividing every original edge once. Then, $H$ has an $F$-gadget $(\mathcal{S},\mathcal{P},\emptyset)$, where $H[S_v] \in \mathcal{S}$ is a clique for every $v \in V(F)$. Now consider any acyclic orientation $\orient{H}$ of $H$. Clearly, $\orient{H}(S_v)$ as a standalone graph has exactly one node of zero indegree, and the same for $\orient{H}(P_e)$, for all $v \in V(F)$ and all $e \in E(F)$. Therefore, by Theorem~\ref{thm:F_tw_H_dtw}, we have that $\tau_1(\orient{H})$ is constant, as $\tw(F)$ is constant (being $F$ fixed). Thus, by Theorem~\ref{thm:FPTRAS_subcount} there exists an FPTRAS for counting the copies of any given $H \in C$ when the parameter is $k+d$. However, in both cases, $C$ has unbounded treewidth. Therefore the existence of an FPTRAS in the non-degenerate case (i.e., for the parameter $k$ only) is not known, and even impossible under the assumption that treewidth characterises FPT approximability of parameterized subgraph counting in the non-degenerate case (see~\cite{ArvindR02}). 


\subsection{Detecting Multicoloured Induced Subgraphs}
Next, we continue with $\mcolindsubsprobd(C)$. In contrast to $\mcolsubsprobd(C)$, we are not able to rely on dynamic programming over dtd's if we aim for any interesting results. Instead, we construct a novel algorithm for detecting colourful induced subgraphs in degenerate graphs, the base case of which is given by the independent set problem:

\begin{lemma}\label{lem:multicoloured_IS}
The problem of detecting a multicoloured $k$-independent set in a $k$-vertex coloured graph $G$, when parameterized by $k+d(G)$ is fixed-parameter tractable  and can be solved in time $O(k^{k} d(G)^k  (k^2+|G|))$.
\end{lemma}
\begin{proof}
Let $d=d(G)$ for short. Consider the following recursive algorithm. Let $G=(V,E)$, and for each $i \in [k]$, compute the set of vertices $V_i$ of $V$ having colour $i$. First, suppose that $|V_i| > d (k-1)$ for all $i$. In this case, we answer YES. To see why, consider $G$ sorted in degeneracy order. Let $v_1$ be the first vertex in $V$, and let its colour be $c_1$. Then, delete $V_{c_1} \cup N_G(v_1)$ from $V$. By the degeneracy ordering, this deletes at most $d$ vertices from any $V_j$ with $j \ne c_1$. Now let $v_2$ be the first vertex in $V$, and let its colour be $c_2$. Then, delete $V_{c_2} \cup N_G(v_2)$ from $V$. Like before, this deletes at most $d$ vertices from any colour class other than $c_2$. Continue until $i=k$. Since initially $|V_i| > d (k-1)$ for all $i \in [k]$, and at each step we delete at most $d$ vertices from any one of the remaining colour classes, then $V$ never becomes empty. This proves that there exists a set of vertices $\{v_1,\ldots,v_k\}$ spanning all $k$ colours. By construction $\{v_1,\ldots,v_k\}$ is an independent set as well. Therefore, in this case the answer is YES.

Suppose now that $|V_i| \le d (k-1)$ for $i=1,\ldots,\ell$ and $|V_i| > d (k-1)$ for $i=\ell+1,\ldots,k$. (We assume the colours are sorted in this way, without loss of generality). Then, we enumerate explicitly $V_1 \times \ldots \times V_{\ell}$, which takes time $(d(k-1))^{\ell}$. For any tuple $(v_1,\ldots,v_{\ell})$, we check whether $G[\{v_1,\ldots,v_{\ell}\}]$ is an independent set. If that is the case, then we do as follows. If $\ell=k$, then we just stop and return YES. Otherwise, we try extending $\{v_1,\ldots,v_{\ell}\}$ into a colourful $k$-independent set. Let $G'$ be the graph obtained by deleting $(V_1 \cup N_G(v_1)) \cup \ldots \cup (V_{\ell} \cup N_G(v_{\ell}))$ from $G$. Observe that $\{v_1,\ldots,v_{\ell}\}$ is part of a colourful $k$-independent set in $G$ if and only if $G'$ contains a colourful $(k-\ell)$-independent set. Therefore, we invoke the entire algorithm recursively on $G'$, and if the output of the recursive call is YES, then we return YES. If $\ell < k$, and the recursive calls return NO for all $\{v_1,\ldots,v_{\ell}\}$, then we return NO.

Straightforward arguments show that the algorithm above is correct. For the running time, consider every call (recursive or not) to the algorithm. By construction, at most $(d(k-1))^{k}$ such calls are made. To each such call we charge (i) the time taken by the call itself for computing of the sets $V_i$, which is $O(|G|)$, and (ii) the time taken by the parent call to compute the input for the call itself, which is $O(k^2 +|G|)$ since the parent call checks if $G[\{v_1,\ldots,v_{\ell}\}]$ is an independent set and then computes $G'$ in linear time. Therefore, the total running time is bounded by $O(k^{k} d^k (k^2 + |G|))$.
\end{proof}

Next, we extend the latter result to patterns of bounded induced matching number.
\begin{lemma}\label{lem:mutlicolindsub}
Let $C$ be a class of graphs. If the induced matching number $\imn(C)$ of $C$ is bounded, then $\mcolindsubsprobd(C)$ is fixed-parameter tractable and can be solved in time $k^{O(k)} d^{k+1} |V(G)|^{\imn(C)+1}$.
\end{lemma}
\begin{proof}
Let $H \in C$ and let $G$ be a $d$-degenerate graph with $k$-vertex colouring $c_G$. First, we compute $\imn(H)$, which can be done in time $O(1.43^k \poly(k))$~\cite{XT17-InducedMatching}. If $\imn(H) = 0$, then we have an instance of the multicoloured independent set problem, which by Lemma~\ref{lem:multicoloured_IS} can be solved in time $O(d^k k^k (k^2 + |G|))$, which is in $k^{O(k)} d^{k+1} |V(G)|^{\imn(C)+1}$ since $|G| \le d |V(G)|$. If $\imn(H) \ge 1$, then we continue below.

As usual, we sort $G$ in degeneracy order in linear time. Now, for each acyclic orientation $\orient{H}$ of $H$, we search for a multicoloured induced copy of $\orient{H}$ in $G$ that respects the orientation of the edges, that is, we search a \emph{multicoloured strong embedding}: a function $\varphi:V(H)\rightarrow V(G)$ such that $(u,v)\in E(\orient{H})$ if and only if $(\varphi(u),\varphi(v))\in E(G)$ for all $u,v\in V(H)$, and $c_G(\varphi(V(H)))=[k]$.  Let $S$ denote the sources of $\orient{H}$. By Lemma~\ref{lem:match_and_kernel}, $\orient{H}$ has a $t$-kernel $K\subseteq S$ with $t\leq \imn(H)$. That is, there exists a set $K$ of $t\leq\imn(H)$ sources from which all non-sources of $\orient{H}$ are reachable. We can find $K$ in time $O(2^k k^2)$ by brute force enumeration and graph search. We then let $\overline{K} = S \setminus K$ be the set of sources outside $K$, and we let $\orient{H}(K) = \orient{H}[V(H) \setminus \overline{K}]$ be the subgraph of $\orient{H}$ reachable from $K$. Now the idea is that detecting a colourful induced copy of $\orient{H}(K)$ in $G$ can be done in FPT time by guessing the image of $K$, and that detecting a colourful induced copy of $\overline{K}$ in $G$ can be done in FPT time by Lemma~\ref{lem:multicoloured_IS} since $\overline{K}$ is an independent set. The only complication is that we should search for the copy of $\overline{K}$ in a certain subgraph of $G$ so that, together with the copy of $\orient{H}(K)$, it forms a copy of $\orient{H}$.

We begin by enumerating all homomorphisms $\varphi' : \orient{H}(K) \rightarrow G$. Let $\ell=|\overline{K}|$, so $|V(\orient{H}(K))|=k-\ell$. Since $\orient{H}(K)$ is reachable from $t$ sources, and every vertex of $G$ has outdegree at most $d$, the enumeration takes time $O(|V(G)|^t d^{k-t-\ell} \poly(k))$, see~\cite{Bressan21}. For each such $\varphi'$ we check if:
\begin{enumerate}
    \item the image of $\varphi'$ is an induced copy of $\orient{H}(K)$ in $G$, and
    \item the vertices in the image of $\varphi'$ have pairwise different colours.
\end{enumerate}
in which case we say $\varphi'$ is valid. Note that testing if $\varphi'$ is valid takes time $O(k^2)$.

Now, for each valid $\varphi'$, we check whether $\varphi'$ can be extended to a multicoloured strong embedding, as follows. Let $\overline{K}=\{s_1,\dots,s_\ell\}$, and let $C' \subset [k]$ be the set of colours not covered by $\varphi'$. For every permutation $c_1,\ldots,c_{\ell}$ of $C'$, for each $i \in [\ell]$ let $V_i$ be the subset of vertices $v$ of $G$ such that:
\begin{enumerate}
    \item $c_G(v) = c_i$
    \item $N^{\text{out}}_G(v)\cap \mathsf{im}(\varphi') = \varphi'(N^{\text{out}}_H(s_i))$, that is, the out-neighbours of $v$ in the image of $\varphi'$ are precisely the images of the out-neighbours of $s_i$ (in $\orient{H}$) under $\varphi'$ 
    \item there are no in-neighbours of $v$ in the image of $\varphi'$.
\end{enumerate}
The sets $V_i$ can be computed in time $O(k^2 |G|)$. If any $V_i$ is empty, then we skip and continue with the next permutation of $C'$. Otherwise, let $G'$ be the subgraph of $G$ induced by $V_1,\dots,V_\ell$. Observe that $G'$ is $\ell$-vertex coloured with colours $c_1,\dots,c_\ell$, and that $G'$ can be constructed in time $|G|$ from the sets $V_i$.

Now, the crucial observation is the existence of a bijection between multicoloured $\ell$-independent sets in $G'$ and extensions $\varphi$ of $\varphi'$ such that:
\begin{enumerate}
    \item $\varphi$ is a multicoloured strong embedding from $\orient{H}$ to $G$, and
    \item for all $i\in [\ell]$, the colour of the image of $s_i$ is $c_i$.
\end{enumerate}
More precisely, the extension $\varphi$ is given by assigning $s_i$ the $c_i$-coloured element of the independent set in $G'$. Therefore, we need only to detect a multicoloured $\ell$-independent set in $G'$, which by Lemma~\ref{lem:multicoloured_IS} takes time $O(d^{\ell} k^k (k^2 + |G|))$, since $d(G') \leq d$. By adding the time taken to compute $G'$, the bound remains unchanged. 

Let us wrap up the bounds. The acyclic orientations of $H$ and the permutations of $\ell \le k$ colours give a $k^{O(k)}$ multiplicative factor. Then, if $\orient{H}$ has a $t$-kernel $K$, and if $|\overline{K}|=\ell$, then enumerating all valid $\varphi'$ takes time $O(d^{k-t-\ell} |V(G)|^{t} \poly(k))$. For every such $\varphi'$, we spend time $O(d^{\ell} k^k (k^2 + |G|))$. This yields a total running time of $k^{O(k)} d^{k-t} |V(G)|^{t} |G|$. Since $t \ge 1$ because we assumed $H$ is not an independent set, since $t \le \imn(H) \le \imn(C)$, and since $|G| \le d |V(G)|$, the total running time is bounded by $k^{O(k)} d^{k} |V(G)|^{\imn(C)+1}$.
\end{proof}

\begin{theorem}[Theorem~\ref{thm:intro_FPTRAS_indsubcount} restated]\label{thm:FPTRAS_indsubcount}
Let $C$ be a class of graphs. If the induced matching number of $C$ is bounded, then the problem $\#\indsubsprobd(C)$ has an FPTRAS with running time $O(\varepsilon^{-2} k^{O(k)} d^{k+1} |V(G)|^{\imn(C)+1+o(1)})$.
\end{theorem}
\begin{proof}
Follows immediately by Lemma~\ref{lem:mutlicolindsub} and Theorem~\ref{thm:coldecapprox}.
\end{proof}

\textbf{Examples.} Let $C$ be the class of complete bipartite graphs. Then, $\imn(C)=1$. Therefore, by virtue of Theorem~\ref{thm:FPTRAS_indsubcount}, we can approximately count the number of induced copies of any given $H \in C$ in time $O(\varepsilon^{-2} k^{O(k)} d^{k+1} |V(G)|^{2+o(1)})$. However, $C$ contains arbitrarily large graphs, and therefore, when parameterized only by $k$, the problem of approximately counting the induced copies of $H \in C$ in $G$ is $\W{1}$-hard (under randomised reductions) by reducing from the decision version~\cite{ChenTW08}.

\subsection{Generalised Induced Subgraph Patterns}
In the last part, we consider the complexity of approximating $\#\indsubsprobd(\Phi)$ for graph properties~$\Phi$ that have certain closure properties. Recall that $\#\indsubsprobd(\Phi)$ expects as input a graph~$G$ and a positive integer $k$, and the goal is to compute the number of induced subgraphs of size $k$ in $G$ that satisfy $\Phi$; the parameterization is given by $k+d(G)$.
Recall further that we established $\#\indsubsprobd(\Phi)$ to be hard in case of exact counting for a variety of natural properties in Theorem~\ref{thm:indsubsphi_hard_cor}, including connectivity, claw-freeness and all minor-closed properties of unbounded independence number.

In the non-degenerate case, that is, if we parameterize by $k$ only, it is known that the problem admits an FPTRAS for the properties of being connected and claw-free~\cite{JerrumM15,Meeks16}. In contrast, there is no FPTRAS in the non-degenerate case for the minor-closed property of being planar, unless $\mathrm{FPT}$ coincides with $\W{1}$ under randomised reductions~\cite{KhotR02}.

Therefore, we investigate to which extent small degeneracy can lead to new tractability results for approximate counting, and, in particular, whether we obtain an FPTRAS for all minor-closed properties. It turns out that a surprisingly simple Monte Carlo algorithm achieves something even stronger.

\begin{theorem}\label{thm:fptras_phi_excludedIS}
Let $\Phi$ be a computable graph property that is true on all but finitely many independent sets. Then $\#\indsubsprobd(\Phi)$ has an FPTRAS.
\end{theorem}
\begin{proof}
If $\Phi$ is true on all but finitely many independent sets, then $\Phi$ is true on all independent sets of size at least $c$ for some constant $c$. Now let $n=|V(G)|$. If $k < c$, then we simply enumerate all $k$-vertex subsets of $G$ and count exactly how many of them satisfy $\Phi$. Hence, in this case we solve $\#\indsubsprobd(\Phi)$ in time $f(k) \cdot O(n^k) \le f(k) \cdot O(n^c)$. If instead $k \ge c$, then we proceed as follows.

First, we observe that $G$ must have an induced independent set of size $\lceil\frac{n}{d+1}\rceil$. To see this, take the vertex of $G$ of smallest degree, delete it from $G$, and continue recursively. Since at each step we delete at most $d+1$ nodes, we can continue for at least $\frac{n}{d+1}$ steps, and thus for at least $\lceil\frac{n}{d+1}\rceil$ steps. Now, if $\lceil\frac{n}{d+1}\rceil < k$, then $n < k(d+1)$, and we solve the problem in time $f(k) \cdot O(n^k) = f(k) \cdot O(k^k (d+1)^k)$ by a brute-force enumeration. If instead $\lceil\frac{n}{d+1}\rceil \ge k$, then the number of induced $k$-independent sets of $G$ is at least $\binom{\lceil\frac{n}{d+1}\rceil}{k} \ge 1$.
Since $k \ge c$, then $\Phi$ holds for any such $k$-independent set. Hence, the fraction of $k$-node induced subgraphs of $G$ that satisfy~$\Phi$ is at least (recalling that $\frac{n}{d+1} \ge k$):
\begin{linenomath*}
\begin{equation}
\frac{\binom{\lceil\frac{n}{d+1}\rceil}{k}}{\binom{n}{k}} \geq \frac{\left(\frac{n}{(d+1)k}\right)^k}{n^k} = k^{-k}(d+1)^{-k}
\end{equation}
\end{linenomath*}
By standard concentration bounds (see e.g.\ Theorem~11.1 in~\cite{MitzenmacherU17}), by drawing $O(\varepsilon^{-2} (dk+k)^k)$ independent uniform samples,  we can compute with probability at least $2/3$ an $\varepsilon$-approximation of the number of $k$-vertex induced subgraphs of $G$ that satisfy $\Phi$.

Since verifying whether a particular sample satisfies $\Phi$ takes time $f(k)$ for some computable function $f$, the overall running time of the algorithm is always bounded by
\begin{linenomath*}
\[O(\varepsilon^{-2} k^k (d+1)^k)\cdot f(k) \cdot n^c  \,,\]
\end{linenomath*}
yielding the desired FPTRAS.
\end{proof}

As an easy consequence, we always get an FPTRAS for minor-closed properties:
\begin{corollary}[Theorem~\ref{thm:intro_approx_minor_closed} restated]
Let $\Phi$ be a minor-closed graph property. Then the problem $\#\indsubsprobd(\Phi)$ has an FPTRAS.
\end{corollary}
\begin{proof}
If $\Phi$ is true for all independent sets, then $\#\indsubsprobd(\Phi)$ has an FPTRAS by Theorem~\ref{thm:fptras_phi_excludedIS}. Otherwise, $\Phi$ is false for an independent set on $c$ vertices for some constant $c$. But then $\Phi$ is false on all graphs with at least $c$ vertices, since $\Phi$ is minor-closed. Thus, given any $G$ and $k$, if $k > c$ we immediately return $0$, otherwise, since $f$ is (poly-time) computable by the Robertson-Seymour theorem~\cite{RobertsonS04}, then in time $f(k) \cdot O(n^k) = f(k) \cdot O(n^c)$ we can enumerate all $k$-vertex subsets of $G$ and return how many of them satisfy $\Phi$. 
\end{proof}

%% file: main.bbl
\begin{thebibliography}{10}
\expandafter\ifx\csname url\endcsname\relax
  \def\url#1{\texttt{#1}}\fi
\expandafter\ifx\csname doi\endcsname\relax
  \def\doi#1{\burlalt{doi:#1}{http://dx.doi.org/#1}}\fi
\expandafter\ifx\csname urlprefix\endcsname\relax\def\urlprefix{URL }\fi
\expandafter\ifx\csname href\endcsname\relax
  \def\href#1#2{#2}\fi
\expandafter\ifx\csname burlalt\endcsname\relax
  \def\burlalt#1#2{\href{#2}{#1}}\fi

\bibitem{Nogaetat08}
N.~Alon, P.~Dao, I.~Hajirasouliha, F.~Hormozdiari, and S.~C. Sahinalp.
\newblock {{Biomolecular network motif counting and discovery by color
  coding}}.
\newblock {\em Bioinformatics}, 24(13):i241--i249, 07 2008.
\newblock \doi{10.1093/bioinformatics/btn163}.

\bibitem{AndersenDSV92}
L.~D. Andersen, S.~Ding, G.~Sabidussi, and P.~D. Vestergaard.
\newblock Edge orbits and edge-deleted subgraphs.
\newblock {\em Graphs Comb.}, 8(1):31--44, 1992.
\newblock \doi{10.1007/BF01271706}.

\bibitem{Arenasetal20}
M.~Arenas, L.~A. Croquevielle, R.~Jayaram, and C.~Riveros.
\newblock {When is Approximate Counting for Conjunctive Queries Tractable?},
  2020, \burlalt{2005.10029}{http://arxiv.org/abs/2005.10029}.
\newblock \urlprefix\url{https://arxiv.org/abs/2005.10029}.

\bibitem{ArvindR02}
V.~Arvind and V.~Raman.
\newblock Approximation {A}lgorithms for {S}ome {P}arameterized {C}ounting
  {P}roblems.
\newblock In {\em Proc.\ of ISAAC}, pages 453--464, 2002.
\newblock \doi{10.1007/3-540-36136-7\_40}.

\bibitem{Bera-ITCS20}
S.~K. Bera, N.~Pashanasangi, and C.~Seshadhri.
\newblock {Linear Time Subgraph Counting, Graph Degeneracy, and the Chasm at
  Size Six}.
\newblock In {\em Proc.\ of ITCS}, volume 151, pages 38:1--38:20, 2020.
\newblock \doi{10.4230/LIPIcs.ITCS.2020.38}.

\bibitem{Bera-SODA21}
S.~K. Bera, N.~Pashanasangi, and C.~Seshadhri.
\newblock Near-linear time homomorphism counting in bounded degeneracy graphs:
  The barrier of long induced cycles.
\newblock In {\em Proc.\ of ACM-SIAM SODA}, pages 2315--2332, 2021.
\newblock \doi{10.1137/1.9781611976465.138}.

\bibitem{Bressan19}
M.~Bressan.
\newblock {Faster Subgraph Counting in Sparse Graphs}.
\newblock In {\em Proc.\ of IPEC}, volume 148, pages 6:1--6:15, 2019.
\newblock \doi{10.4230/LIPIcs.IPEC.2019.6}.

\bibitem{Bressan21}
M.~Bressan.
\newblock Faster algorithms for counting subgraphs in sparse graphs.
\newblock {\em Algorithmica}, 2021.
\newblock \doi{10.1007/s00453-021-00811-0}.

\bibitem{Bringmann21}
K.~Bringmann and J.~Slusallek.
\newblock Current algorithms for detecting subgraphs of bounded treewidth are
  probably optimal.
\newblock {\em CoRR}, abs/2105.05062, 2021,
  \burlalt{2105.05062}{http://arxiv.org/abs/2105.05062}.
\newblock \urlprefix\url{https://arxiv.org/abs/2105.05062}.

\bibitem{BulatovZ20}
A.~A. Bulatov and S.~Zivn{\'{y}}.
\newblock {Approximate Counting {CSP} Seen from the Other Side}.
\newblock {\em {ACM} Trans. Comput. Theory}, 12(2):11:1--11:19, 2020.
\newblock \doi{10.1145/3389390}.

\bibitem{ChenM15}
H.~Chen and S.~Mengel.
\newblock A {T}richotomy in the {C}omplexity of {C}ounting {A}nswers to
  {C}onjunctive {Q}ueries.
\newblock In {\em 18th International Conference on Database Theory, {ICDT}
  2015, March 23-27, 2015, Brussels, Belgium}, pages 110--126, 2015.
\newblock \doi{10.4230/LIPIcs.ICDT.2015.110}.

\bibitem{ChenM16}
H.~Chen and S.~Mengel.
\newblock Counting {A}nswers to {E}xistential {P}ositive {Q}ueries: {A}
  {C}omplexity {C}lassification.
\newblock In {\em Proc.\ of ACM PODS}, pages 315--326, 2016.
\newblock \doi{10.1145/2902251.2902279}.

\bibitem{Chenetal05}
J.~Chen, B.~Chor, M.~Fellows, X.~Huang, D.~W. Juedes, I.~A. Kanj, and G.~Xia.
\newblock Tight lower bounds for certain parameterized {N}{P}-hard problems.
\newblock {\em Information and Computation}, 201(2):216--231, 2005.
\newblock \doi{10.1016/j.ic.2005.05.001}.

\bibitem{Chenetal06}
J.~Chen, X.~Huang, I.~A. Kanj, and G.~Xia.
\newblock Strong computational lower bounds via parameterized complexity.
\newblock {\em J. Comput. Syst. Sci.}, 72(8):1346--1367, 2006.
\newblock \doi{10.1016/j.jcss.2006.04.007}.

\bibitem{ChenTW08}
Y.~Chen, M.~Thurley, and M.~Weyer.
\newblock Understanding the {C}omplexity of {I}nduced {S}ubgraph
  {I}somorphisms.
\newblock In {\em Proc.\ of ICALP}, pages 587--596, 2008.
\newblock \doi{10.1007/978-3-540-70575-8\_48}.

\bibitem{Chiba&1985}
N.~Chiba and T.~Nishizeki.
\newblock Arboricity and subgraph listing algorithms.
\newblock {\em SIAM J. Comput.}, 14(1):210--223, 1985.
\newblock \doi{10.1137/0214017}.

\bibitem{Courcelle2000-CW}
B.~Courcelle and S.~Olariu.
\newblock Upper bounds to the clique width of graphs.
\newblock {\em Discrete Applied Mathematics}, 101(1):77--114, 2000.
\newblock \doi{10.1016/S0166-218X(99)00184-5}.

\bibitem{Curticapean13}
R.~Curticapean.
\newblock Counting matchings of size k is {W}[1]-hard.
\newblock In {\em Proc.\ of ICALP}, volume 7965 of {\em Lecture Notes in
  Computer Science}, pages 352--363. Springer, 2013.
\newblock \doi{10.1007/978-3-642-39206-1\_30}.

\bibitem{Curticapean15}
R.~Curticapean.
\newblock {\em The simple, little and slow things count: {O}n parameterized
  counting complexity}.
\newblock PhD thesis, Saarland University, 2015.
\newblock
  \urlprefix\url{http://scidok.sulb.uni-saarland.de/volltexte/2015/6217/}.

\bibitem{CurticapeanDM17}
R.~Curticapean, H.~Dell, and D.~Marx.
\newblock Homomorphisms are a good basis for counting small subgraphs.
\newblock In {\em Proc.\ of ACM STOC}, pages 210--223, 2017.
\newblock \doi{10.1145/3055399.3055502}.

\bibitem{CurticapeanM14}
R.~Curticapean and D.~Marx.
\newblock Complexity of {C}ounting {S}ubgraphs: {O}nly the {B}oundedness of the
  {V}ertex-{C}over {N}umber {C}ounts.
\newblock In {\em Proc.\ of IEEE FOCS}, pages 130--139, 2014.
\newblock \doi{10.1109/FOCS.2014.22}.

\bibitem{CyganFKLMPPS15}
M.~Cygan, F.~V. Fomin, L.~Kowalik, D.~Lokshtanov, D.~Marx, M.~Pilipczuk,
  M.~Pilipczuk, and S.~Saurabh.
\newblock {\em Parameterized {A}lgorithms}.
\newblock Springer, 2015.
\newblock \doi{10.1007/978-3-319-21275-3}.

\bibitem{DalmauJ04}
V.~Dalmau and P.~Jonsson.
\newblock The complexity of counting homomorphisms seen from the other side.
\newblock {\em Theoretical Computer Science}, 329(1-3):315--323, 2004.
\newblock \doi{10.1016/j.tcs.2004.08.008}.

\bibitem{DLM20}
H.~Dell, J.~Lapinskas, and K.~Meeks.
\newblock Approximately counting and sampling small witnesses using a colourful
  decision oracle.
\newblock In {\em Proc.\ of ACM-SIAM SODA}, pages 2201--2211, 2020.
\newblock \doi{10.1137/1.9781611975994.135}.

\bibitem{DorflerRSW19}
J.~D{\"{o}}rfler, M.~Roth, J.~Schmitt, and P.~Wellnitz.
\newblock {Counting Induced Subgraphs: An Algebraic Approach to
  {\#}W[1]-hardness}.
\newblock In {\em Proc.\ of MFCS}, volume 138, pages 26:1--26:14, 2019.
\newblock \doi{10.4230/LIPIcs.MFCS.2019.26}.

\bibitem{DurandM15}
A.~Durand and S.~Mengel.
\newblock Structural {T}ractability of {C}ounting of {S}olutions to
  {C}onjunctive {Q}ueries.
\newblock {\em Theory of Computing Systems}, 57(4):1202--1249, 2015.
\newblock \doi{10.1007/s00224-014-9543-y}.

\bibitem{DyerGJ10}
M.~E. Dyer, L.~A. Goldberg, and M.~Jerrum.
\newblock An approximation trichotomy for {B}oolean {\#}{CSP}.
\newblock {\em J. Comput. Syst. Sci.}, 76(3-4):267--277, 2010.
\newblock \doi{10.1016/j.jcss.2009.08.003}.

\bibitem{Eden20-CliqueGap}
T.~Eden, D.~Ron, and W.~Rosenbaum.
\newblock Almost optimal bounds for sublinear-time sampling of $k$-cliques:
  Sampling cliques is harder than counting.
\newblock 2020, \burlalt{2012.04090}{http://arxiv.org/abs/2012.04090}.

\bibitem{Eden20-CliqueArboricity}
T.~Eden, D.~Ron, and C.~Seshadhri.
\newblock Faster sublinear approximation of the number of k-cliques in
  low-arboricity graphs.
\newblock In {\em Proc.\ of ACM-SIAM SODA}, pages 1467--1478.
\newblock \doi{10.1137/1.9781611975994.89}.

\bibitem{Eppstein99}
D.~Eppstein.
\newblock Subgraph {I}somorphism in {P}lanar {G}raphs and {R}elated {P}roblems.
\newblock {\em J. Graph Algorithms Appl.}, 3(3):1--27, 1999.
\newblock \doi{10.7155/jgaa.00014}.

\bibitem{Eppstein2011-ListClique}
D.~Eppstein and D.~Strash.
\newblock Listing all maximal cliques in large sparse real-world graphs.
\newblock In P.~M. Pardalos and S.~Rebennack, editors, {\em Experimental
  Algorithms}, pages 364--375, 2011.

\bibitem{ErmanKKW10}
R.~Erman, L.~Kowalik, M.~Krnc, and T.~Walen.
\newblock Improved induced matchings in sparse graphs.
\newblock {\em Discret. Appl. Math.}, 158(18):1994--2003, 2010.
\newblock \doi{10.1016/j.dam.2010.08.026}.

\bibitem{FlumG04}
J.~Flum and M.~Grohe.
\newblock The {P}arameterized {C}omplexity of {C}ounting {P}roblems.
\newblock {\em {SIAM} J. Comput.}, 33(4):892--922, 2004.
\newblock \doi{10.1137/S0097539703427203}.

\bibitem{FlumG06}
J.~Flum and M.~Grohe.
\newblock {\em {P}arameterized {C}omplexity {T}heory}.
\newblock Texts in Theoretical Computer Science. An {EATCS} Series. Springer,
  2006.
\newblock \doi{10.1007/3-540-29953-X}.

\bibitem{Gishboliner20}
L.~Gishboliner, Y.~Levanzov, and A.~Shapira.
\newblock Counting subgraphs in degenerate graphs.
\newblock {\em Electron. Colloquium Comput. Complex.}, 27:177, 2020.
\newblock \urlprefix\url{https://eccc.weizmann.ac.il/report/2020/177}.

\bibitem{Grohe07}
M.~Grohe.
\newblock The complexity of homomorphism and constraint satisfaction problems
  seen from the other side.
\newblock {\em J. {ACM}}, 54(1):1:1--1:24, 2007.
\newblock \doi{10.1145/1206035.1206036}.

\bibitem{Grohe&2009expansion}
M.~Grohe and D.~Marx.
\newblock On tree width, bramble size, and expansion.
\newblock {\em Journal of Combinatorial Theory, Series B}, 99(1):218 -- 228,
  2009.
\newblock \doi{10.1016/j.jctb.2008.06.004}.

\bibitem{ImpagliazzoP01}
R.~Impagliazzo and R.~Paturi.
\newblock On the {C}omplexity of k-{S}{A}{T}.
\newblock {\em J. Comput. Syst. Sci.}, 62(2):367--375, 2001.
\newblock \doi{10.1006/jcss.2000.1727}.

\bibitem{JerrumM15}
M.~Jerrum and K.~Meeks.
\newblock The parameterised complexity of counting connected subgraphs and
  graph motifs.
\newblock {\em J. Comput. Syst. Sci.}, 81(4):702--716, 2015.
\newblock \doi{10.1016/j.jcss.2014.11.015}.

\bibitem{Kasteleyn63}
P.~W. Kasteleyn.
\newblock Dimer {S}tatistics and {P}hase {T}ransitions.
\newblock {\em Journal of Mathematical Physics}, 4(2):287--293, 1963.
\newblock \doi{10.1063/1.1703953}.

\bibitem{KhotR02}
S.~Khot and V.~Raman.
\newblock Parameterized complexity of finding subgraphs with hereditary
  properties.
\newblock {\em Theor. Comput. Sci.}, 289(2):997--1008, 2002.
\newblock \doi{10.1016/S0304-3975(01)00414-5}.

\bibitem{Komarath20}
B.~Komarath, A.~Pandey, and C.~S. Rahul.
\newblock Graph homomorphism polynomials: Algorithms and complexity.
\newblock {\em CoRR}, abs/2011.04778, 2020,
  \burlalt{2011.04778}{http://arxiv.org/abs/2011.04778}.
\newblock \urlprefix\url{https://arxiv.org/abs/2011.04778}.

\bibitem{Lovasz12}
L.~Lov{\'{a}}sz.
\newblock {\em Large {N}etworks and {G}raph {L}imits}, volume~60 of {\em
  Colloquium Publications}.
\newblock American Mathematical Society, 2012.
\newblock \urlprefix\url{http://www.ams.org/bookstore-getitem/item=COLL-60}.

\bibitem{Marx10}
D.~Marx.
\newblock Can {Y}ou {B}eat {T}reewidth?
\newblock {\em Theory of Computing}, 6(1):85--112, 2010.
\newblock \doi{10.4086/toc.2010.v006a005}.

\bibitem{Marx13}
D.~Marx.
\newblock Tractable hypergraph properties for constraint satisfaction and
  conjunctive queries.
\newblock {\em J. {ACM}}, 60(6):42:1--42:51, 2013.
\newblock \doi{10.1145/2535926}.

\bibitem{Meeks16}
K.~Meeks.
\newblock The challenges of unbounded treewidth in parameterised subgraph
  counting problems.
\newblock {\em Discrete Applied Mathematics}, 198:170--194, 2016.
\newblock \doi{10.1016/j.dam.2015.06.019}.

\bibitem{Miloetal02}
R.~Milo, S.~Shen-Orr, S.~Itzkovitz, N.~Kashtan, D.~Chklovskii, and U.~Alon.
\newblock {Network Motifs: Simple Building Blocks of Complex Networks}.
\newblock {\em Science}, 298(5594):824--827, 2002.
\newblock \doi{10.1126/science.298.5594.824}.

\bibitem{MitzenmacherU17}
M.~Mitzenmacher and E.~Upfal.
\newblock {\em {Probability and Computing: Randomized Algorithms and
  Probabilistic Analysis}}.
\newblock Cambridge University Press, second edition, 2017.

\bibitem{NesetrildM12}
J.~Nesetril and P.~O. de~Mendez.
\newblock {\em Sparsity - Graphs, Structures, and Algorithms}, volume~28 of
  {\em Algorithms and combinatorics}.
\newblock Springer, 2012.
\newblock \doi{10.1007/978-3-642-27875-4}.

\bibitem{RobertsonS86-ExGrid}
N.~Robertson and P.~D. Seymour.
\newblock Graph minors. {V.} {E}xcluding a planar graph.
\newblock {\em J. Comb. Theory, Ser. {B}}, 41(1):92--114, 1986.
\newblock \doi{10.1016/0095-8956(86)90030-4}.

\bibitem{RobertsonS04}
N.~Robertson and P.~D. Seymour.
\newblock Graph minors. {XX.} {W}agner's conjecture.
\newblock {\em J. Comb. Theory, Ser. {B}}, 92(2):325--357, 2004.
\newblock \doi{10.1016/j.jctb.2004.08.001}.

\bibitem{Roth19}
M.~Roth.
\newblock {\em Counting {P}roblems on {Q}uantum {G}raphs: {P}arameterized and
  {E}xact {C}omplexity {C}lassifications}.
\newblock PhD thesis, Saarland University, 2019.
\newblock
  \urlprefix\url{https://scidok.sulb.uni-saarland.de/bitstream/20.500.11880/27575/1/thesis.pdf}.

\bibitem{RothS18}
M.~Roth and J.~Schmitt.
\newblock Counting induced subgraphs: {A} {T}opological {A}pproach to
  {\#}{W}[1]-hardness.
\newblock In {\em Proc.\ of IPEC}, pages 24:1--24:14, 2018.
\newblock \doi{10.4230/LIPIcs.IPEC.2018.24}.

\bibitem{RothSW20}
M.~Roth, J.~Schmitt, and P.~Wellnitz.
\newblock Counting small induced subgraphs satisfying monotone properties.
\newblock In {\em Proc.\ of IEEE FOCS}, pages 1356--1367. {IEEE}, 2020.
\newblock \doi{10.1109/FOCS46700.2020.00128}.

\bibitem{RothW20}
M.~Roth and P.~Wellnitz.
\newblock {Counting and Finding Homomorphisms is Universal for Parameterized
  Complexity Theory}.
\newblock In {\em Proc.\ of ACM-SIAM SODA}, pages 2161--2180. {SIAM}, 2020.
\newblock \doi{10.1137/1.9781611975994.133}.

\bibitem{Schilleretal15}
B.~Schiller, S.~Jager, K.~Hamacher, and T.~Strufe.
\newblock {StreaM - A Stream-Based Algorithm for Counting Motifs in Dynamic
  Graphs}.
\newblock In {\em Algorithms for Computational Biology}, pages 53--67, Cham,
  2015.

\bibitem{TemperleyF61}
H.~N.~V. Temperley and M.~E. Fisher.
\newblock Dimer problem in statistical mechanics-an exact result.
\newblock {\em The Philosophical Magazine: A Journal of Theoretical
  Experimental and Applied Physics}, 6(68):1061--1063, 1961.
\newblock \doi{10.1080/14786436108243366}.

\bibitem{Hofetal12}
P.~van~'t Hof, M.~Kaminski, D.~Paulusma, S.~Szeider, and D.~M. Thilikos.
\newblock On graph contractions and induced minors.
\newblock {\em Discret. Appl. Math.}, 160(6):799--809, 2012.
\newblock \doi{10.1016/j.dam.2010.05.005}.

\bibitem{XT17-InducedMatching}
M.~Xiao and H.~Tan.
\newblock Exact algorithms for maximum induced matching.
\newblock {\em Information and Computation}, 256:196--211, 2017.
\newblock \doi{10.1016/j.ic.2017.07.006}.

\end{thebibliography}
